\documentclass[11pt]{article}
\usepackage{afterpage,amsfonts,amsmath,amssymb,bbm,color,lscape,epstopdf,float,graphicx,tabularx,url,times,theorem,natbib,multirow,rotating,hyperref,xcolor,setspace}
\usepackage[bottom]{footmisc}
\usepackage[adobe-utopia]{mathdesign}
\usepackage[scaled]{helvet}
\usepackage{breakurl}

\renewcommand{\baselinestretch}{1.5}

\definecolor{dbblue}{RGB}{10,65,155}				
\definecolor{blue}{RGB}{0,113.9850,188.9550} 			
\definecolor{red}{RGB}{216.7500,82.8750,24.9900} 		
\definecolor{green}{RGB}{118.8300,171.8700,47.9400} 	
\definecolor{grey}{RGB}{110,110,110}				
\definecolor{c1}{RGB}{0,113.9850,188.9550}			
\definecolor{c2}{RGB}{216.7500,82.8750,24.9900}		
\definecolor{c3}{RGB}{236.8950,176.9700,31.8750}		
\definecolor{c4}{RGB}{125.9700,46.9200,141.7800}		
\definecolor{c5}{RGB}{118.8300,171.8700,47.9400}		
\definecolor{c6}{RGB}{76.7550,189.9750,237.9150}		
\definecolor{c7}{RGB}{161.9250,19.8900,46.9200}		

\hypersetup{colorlinks,urlcolor=dbblue,citebordercolor=dbblue,linkbordercolor=dbblue,filecolor=dbblue,filebordercolor=dbblue,citecolor=dbblue,linkcolor=dbblue,colorlinks=true}

\newtheorem{theorem}{Theorem}
\newtheorem{assumption}{Assumption}

\newtheorem{lemma}{Lemma}
\newtheorem{remark}{Remark}

\allowdisplaybreaks
\DeclareMathOperator{\sign}{sign}
\newcommand{\overbar}[1]{\mkern 1.5mu\overline{\mkern-1.5mu#1\mkern-1.5mu}\mkern 1.5mu}

\newcommand{\qed}{\hfill $\blacksquare$}
\newenvironment{proof}[1][Proof.]{\begin{trivlist} \item[\hskip \labelsep {\bfseries #1}]}{\end{trivlist}}

\makeatletter
\newcommand{\doublewidetilde}[1]{{%
  \mathpalette\double@widetilde{#1}%
}}
\newcommand{\double@widetilde}[2]{%
  \sbox\z@{$\m@th#1\widetilde{#2}$}%
  \ht\z@=.9\ht\z@
  \widetilde{\box\z@}%
}
\makeatother

\begin{document}

\title{The drift burst hypothesis\vspace*{-0.50cm}}
\author{Kim Christensen \and Roel Oomen \and Roberto Ren\`{o}\thanks{Christensen: Department of Economics and Business
Economics, CREATES, Aarhus University, \texttt{kim@econ.au.dk}. Oomen: Deutsche Bank, London and Department of Statistics, London School of Economics. Ren\`{o}: Department of Economics, University of Verona, \texttt{roberto.reno@univr.it}. We thank Maria Flora, Alessandro Gnoatto, Martino Grasselli, Frederich Hubalek, Aleksey Kolokolov, Sebastian Laurent, Oliver Linton, Cecilia Mancini, Nour Meddahi, Per Mykland, Thorsten Rheinlander, anonymous referees, and participants at the 9th Annual SoFiE Conference in Hong Kong, XVII--XVIII Workshop in Quantitative Finance in Pisa and Milan, 3rd Empirical Finance Workshop at ESSEC, Paris, 10th CFE Conference in Seville, annual conference on Market Microstructure and High Frequency Data at the Stevanovich Center, U. of Chicago, and at seminars in DCU Dublin, Rady School of Management (UCSD), SAFE Frankfurt, Toulouse, TU Wien, Unicredit, U. of Venice, and CREATES for helpful comments and suggestions. Christensen received funding from the Danish Council for Independent Research (DFF -- 4182-00050) and was supported by CREATES. We thank the CME Group and Deutsche Bank AG for access to the data. The views and opinions rendered in this paper reflect the authors' personal views about the subject and do not necessarily represent the views of Deutsche Bank AG, any part thereof, or any other organization. This article is necessarily general and is not intended to be comprehensive, nor does it constitute legal or financial advice in relation to any particular situation. \textsc{Matlab} code to compute the proposed drift burst $t$-statistic is available on request. Commercial usage of the code is prohibited without explicit consent from the authors.}}
\date{November 2020}
\maketitle

\vspace*{-1.00cm}

\begin{abstract}
The drift burst hypothesis postulates the existence of short-lived locally explosive trends in the price paths of financial assets.
The recent U.S. equity and treasury flash crashes can be viewed as two high-profile manifestations of such dynamics, but we argue that drift bursts of varying magnitude are an expected and regular occurrence in financial markets that can arise through established mechanisms of liquidity provision. We show how to build drift bursts into the continuous-time It\^{o} semimartingale model, discuss the conditions required for the process to remain arbitrage-free, and propose a nonparametric test statistic that identifies drift bursts from noisy high-frequency data. We apply the test and demonstrate that drift bursts are a stylized fact of the price dynamics across equities, fixed income, currencies and commodities. Drift bursts occur once a week on average, and the majority of them are accompanied by subsequent price reversion and can thus be regarded as ``flash crashes.'' The reversal is found to be stronger for negative drift bursts with large trading volume, which is consistent with endogenous demand for immediacy during market crashes.


\end{abstract}

\thispagestyle{empty}

\pagebreak

\section{Introduction} \setcounter{page}{1}

The orderly functioning of financial markets is viewed by most regulators as their first and foremost objective. It is therefore unsurprising that the recent flash crashes in the U.S. equity and treasury markets are subject to intense debate and scrutiny, not least because they raise concerns around the stability of the market and the integrity of its design \citep*[see e.g.][and Figure \ref{figure:flash_crash} for an illustration]{cftc-sec:10a,cftc-sec:11a,cftc-sec:15a}. Moreover, there is growing consensus that flash crashes of varying magnitude are becoming more frequent across financial markets.\footnote{Nanex Research has reported hundreds of flash crashes across all major financial markets, see \url{http://www.nanex.net/NxResearch/}. Related work includes \citet*{golub-keane-poon:12a}.} The distinct price evolution over such events -- with highly directional and sustained price moves -- poses three direct challenges to the academic community. First, how can one formally model such dynamics? The literature on continuous-time finance has focused extensively on the volatility and jump components of the price process, but as we show, these are not sufficient to explain the observed dynamics. Secondly, how can one identify or test for the presence of such features in the data? And third, can such events be reconciled within the theory of price formation in the presence of market frictions? This paper addresses all these challenges.

Suppose the log-price of a traded asset, $X = (X_{t})_{t \geq 0}$, has the dynamic:
\begin{equation} \label{equation:main-model}
\text{d}X_{t} =  \mu_{t} \text{d}t + \sigma_{t} \text{d}W_{t} + \text{d} J_{t},
\end{equation}
where $\mu_{t}$ is the drift, $\sigma_{t}$ is the volatility, $W_{t}$ a Brownian motion, and $J_{t}$ is a jump process. In a conventional setup with locally bounded coefficients, over a vanishing time interval $\Delta \rightarrow 0$, the drift is $O_{p}( \Delta)$ (as is the jump term) and swamped by a diffusive component of larger order $O_{p}( \sqrt{ \Delta})$. Hence, much of the infill asymptotics is unaffected by the presence of a drift and the theory therefore invariably neglects it. Also, in empirical applications, particularly those relying on intraday data over short horizons, the drift term is generally small and estimates of it are subject to considerable measurement error \citep*[e.g.][]{merton:80a}. Thus, the common recommendation is to ignore it. Yet, to explain such events as those in Figure \ref{figure:flash_crash}, it is hard to see how the drift component can be dismissed. Our starting point is therefore -- what we refer to as -- the drift burst hypothesis, which postulates the existence of short-lived locally explosive trends in the price paths of financial assets. The objective of this paper is to build theoretical and empirical support for the hypothesis, thereby contributing towards a better understanding of financial market dynamics. We show how drift bursts can be embedded in Eq. \eqref{equation:main-model}. Next, we develop a feasible nonparametric identification strategy that enables the online detection of drift bursts from high-frequency data. What is tested here is a drift explosion against the null hypothesis of a jump-diffusion model.\footnote{We assume that volatility is strictly positive, so pure-jump It\^{o} semimartingale processes are ruled out.} The empirical application demonstrates that drift bursts are a stylized fact of the price process.

\begin{figure}[t!]
\begin{center}
\caption{The U.S. S\&P500 equity index and treasury market flash crash. \label{figure:flash_crash}}
\begin{tabular}{cc}
\small{Panel A: S\&P500 equity index.} & \small{Panel B: Treasury market.}\\
\includegraphics[height=0.4\textwidth,width=0.5\textwidth]{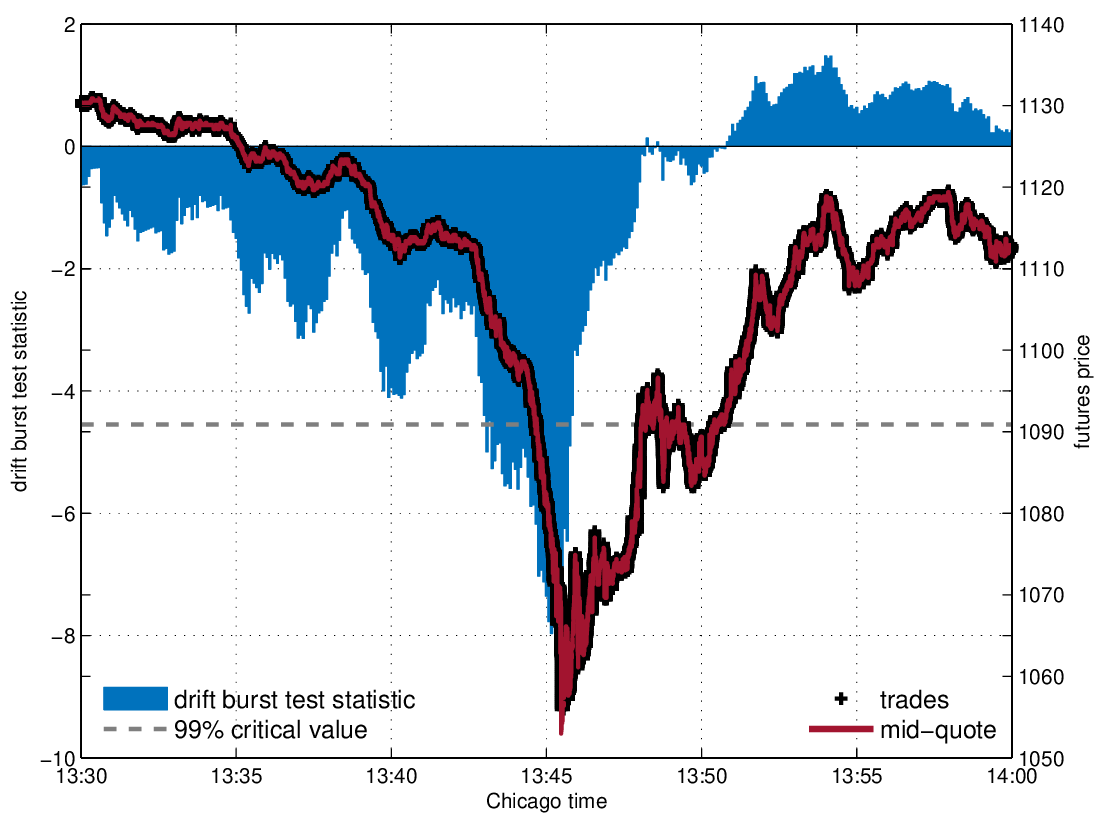} &
\includegraphics[height=0.4\textwidth,width=0.5\textwidth]{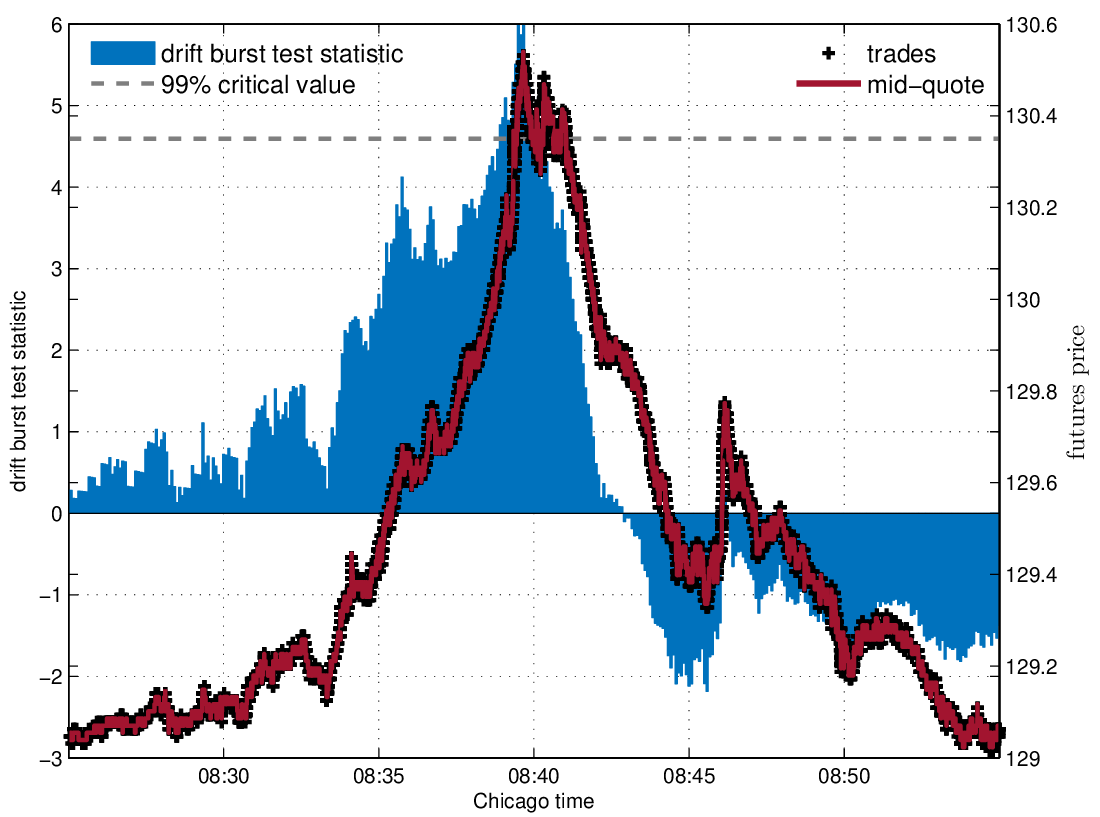} \\
\end{tabular}
\begin{scriptsize}
\parbox{\textwidth}{\emph{Note.} This figure draws the mid-quote and traded price (right axis) of the E-mini S\&P500 (in Panel A) and 10-Year Treasury Note (in Panel B) futures contracts over the flash crash episodes of May 6, 2010 and October 15, 2014. Superimposed is the nonparametric drift burst $t$-statistic (left axis) proposed in this paper.}
\end{scriptsize}
\end{center}
\end{figure}

An exploding drift is unconventional in continuous-time finance, but there are a number of theoretical models of price formation that back the idea. \citet{grossman-miller:88a} examine a collection of risk-averse market makers that provide immediacy in exchange for a positive expected (excess) return $\mu = E(P_1/P_0-1)$ of the form:
\begin{equation} \label{eqn:GM88}
\frac{ \mu}{ \sigma} = \frac{s \gamma}{1+M} \sigma P_{0},
\end{equation}
where $P_{0}$ is the initial price, $P_{1}$ is the price at which the market maker trades, $\sigma$ is the standard deviation of the price move, $s$ is the size of the order that requires execution, $\gamma$ is the risk aversion of market makers, and $M$ is the number of market makers competing for the order. Eq. \eqref{eqn:GM88} illustrates that the drift can dominate the volatility when liquidity demand ($s$) is unusually high or the willingness or capacity of the collective market makers to absorb order flow is impaired (i.e. increased risk aversion $\gamma$ or fewer active market makers $M$). This prediction fits the 2010 equity flash crash in that the price drop appeared to be accompanied by increased risk aversion and a rapid decline in the number of participating market makers. \citet*{cftc-sec:10a} write ``some market makers and other liquidity providers widened their quote spreads, others reduced offered liquidity, and a significant number withdrew completely from the markets.'' The subsequent price reversal observed in Figure \ref{figure:flash_crash} is also predicted by this model as the excess return is only temporary and the long-run price level returns to $P_{0}$. In related work, \citet{campbell-grossman-wang:93a} show that as liquidity demand increases (as measured by trading volume), the price reaction and subsequent reversal grow in size. Alternative mechanisms that can generate these price dynamics include trading frictions as in \citet{huang-wang:09a}, predatory trading and forced liquidation as in \citet{brunnermeier-pedersen:05a,brunnermeier-pedersen:09a}, or agents with tournament-type preferences and an aversion to missing out on trends as in \citet{Johnson:16a}.\footnote{There are also a number of practical mechanisms that cause or amplify price drops and surges, such as margin calls on leveraged positions (i.e. forced liquidation), dynamic hedging of short-gamma positions, stop-loss orders, or technical momentum trading strategies.} While this literature provides valuable insights and hypotheses regarding price dynamics, the testable implications often relate to confounding measures such as the unconditional serial correlation of price returns. Moreover, because the theory is typically cast as a two-period model, it does not readily translate into an econometric identification strategy of the impacted sample paths in continuous-time. In addition to the above, there is an extensive econometrics literature on bubble detection \citep*[see e.g.][]{phillips-yu:11a}, which is potentially related to our paper, but it operates on a much longer time scale than the one analyzed here. We deliver a foundation that increases the depth of the empirical work that can be conducted in these areas.

With the drift burst hypothesis and the corresponding It\^{o} semimartingale price process in place, we develop a nonparametric identification strategy for the online detection of drift burst sample paths from intraday noisy high-frequency data. Our approach aims to establish whether the observed price movement is generated by the drift rather than be the result of volatility. Unsurprisingly, it requires estimation of the local drift and volatility coefficients which is nontrivial for several reasons. First, from \citet*{merton:80a} we know that even when the drift term is a constant, it cannot be estimated consistently over a bounded time interval. Secondly, while infill asymptotics do provide consistent volatility estimates, in practice microstructure effects complicate inference. Building on the \citet*{bandi:02a,kristensen:10a} for coefficient estimation and \citet*{newey-west:87a,andrews:91a,barndorff-nielsen-hansen-lunde-shephard:08a,jacod-li-mykland-podolskij-vetter:09a} for the robustification to microstructure noise, we formulate a nonparametric kernel-based filtering approach that delivers estimates of the local drift and volatility on the basis of which we construct the test statistic. Under the null hypothesis of no drift burst, the test is asymptotically standard normal, but it diverges -- and therefore has power under the alternative -- when the drift explodes sufficiently fast. When calculated sequentially and using potentially overlapping data, the critical values of the test are determined on the basis of extreme value theory, as in \citet{lee-mykland:08a}. A simulation study confirms that the test is capable of identifying drift bursts. Interestingly, applying the test to high-frequency data for the days of the U.S. equity and treasury market flash crashes, displayed in Figure \ref{figure:flash_crash}, we find that they constitute highly significant drift bursts.

The introduction of drift bursts via an exploding drift coefficient provides an essential ingredient, which helps to reconcile a number of phenomena observed in financial markets. The first is the occurrence of flash crashes, where highly directional and sustained price movements are reversed shortly after. While there is a substantial body of research that looks at the 2010 equity market flash crash \citep*[a partial list includes][]{easley-prado-ohara:11a, madhavan:12a, andersen-bondarenko-kyle-obizhaeva:15a, kirilenko-kyle-samadi-tuzun:17a, menkveld-yueshen:18a}, there has been no attempt thus far to analyze these events in a more systematic fashion. Our test procedure lays down a framework that makes this possible. The second is that of ``gradual jumps'' -- in \citet*{barndorff-nielsen-hansen-lunde-shephard:09a} terminology -- where the price converges in a rapid but continuous fashion to a new level. This relates to a puzzle put forward by \citet{christensen-oomen-podolskij:14a}, who find that the total return variation that can be attributed to the jump component is an order of magnitude smaller than previously reported by extensive empirical literature \citep*[see also][]{bajgrowicz-scaillet-treccani:16a}. In particular, they show that jumps identified using data sampled at a five-minute frequency often vanish when viewed at the highest available tick frequency and instead appear as sharp but continuous price movements. They show that spurious detection of jumps at low frequency can be explained by an erratic volatility process. However, because volatility merely leads to wider price dispersion, it fails to reconcile the often steady and directional price evolution over such episodes. We argue that the drift burst hypothesis constitutes a more intuitive and appealing mechanism that can explain the reported over-estimation of the total jump variation.

We undertake an empirical analysis to determine the prevalence of drift bursts in practice and to characterize their basic features. To that end, we employ a comprehensive set of high-quality tick data covering some of the most liquid futures contracts across the equity, fixed income, currency, and commodity markets. We calculate the drift burst test statistic at five-second intervals over a multi-year sample. Our findings demonstrate that drift bursts are an integral part of the price process across all asset classes. Over the full sample we identify more than one thousand highly significant events, or roughly one per week. For most of the drift bursts we detect, the magnitude of the price drop or surge typically ranges between 25 and 200 basis points, with only a handful of extreme moves between 3\% and 8\%. We find that roughly two thirds of the drift bursts are followed by price reversion, which means that many of the identified events resemble (mini) flash crashes that are symptomatic of liquidity shocks. Consistent with the literature on price formation, particularly that in \citet{huang-wang:09a} of endogenous trading imbalances generated by costly market presence, we find that trading volume during a drift burst is highly correlated with the subsequent price reversal. The post-drift burst return can therefore -- as predicted by the theory -- be interpreted as a compensation for supplying immediacy during times of substantial market stress.

The remainder of the paper is organized as follows. Section \ref{section:dbh} introduces the drift burst hypothesis and describes the theoretical framework. Section \ref{section:identification} develops the identification strategy on the basis of noisy high-frequency data. Section \ref{section:simulation} includes an extensive simulation study that demonstrates the properties of the test. The empirical application is found in Section \ref{section:empirical}, while Section \ref{section:conclusion} concludes.

\section{The hypothesis} \label{section:dbh}

The log-price process $X = (X_{t})_{t \geq 0}$, introduced in Eq. \eqref{equation:main-model}, is an It\^{o} semimartingale defined on a filtered probability space $( \Omega, \mathcal{F}, (\mathcal{F}_{t})_{t \geq 0}, \mathcal{P})$. We assume the following about $X$.

\begin{assumption} \label{assumption:model}
$X$ is described by the dynamics in Eq. \eqref{equation:main-model}, where $X_{0}$ is $\mathcal{F}_{0}$-measurable, $\mu = ( \mu_{t})_{t \geq 0}$ is a locally bounded and predictable drift, $\sigma = ( \sigma_{t})_{t \geq 0}$ is an adapted, c\`{a}dl\`ag, locally bounded and almost surely (a.s.) strictly positive volatility, $W = (W_{t})_{t \geq 0}$ is a standard Brownian motion and $J = (J_{t})_{t \geq 0}$ is a pure-jump process.
\end{assumption}

The above is a standard formulation for continuous-time arbitrage-free price processes. It represents our frictionless null, where the price is in a ``normal'' state with non-explosive (locally bounded) coefficients $\mu_{t}$ and $\sigma_{t}$. We do not restrict the model in any essential way, other than by imposing mild regularity conditions on the driving terms, which are listed in Assumption \ref{assumption:model2} in Appendix \ref{appendix:proof}. As such, it encompasses a wide range of specifications and is compatible with time-varying expected returns, stochastic volatility, leverage effects, and jumps (both of finite and infinite activity) in the log-price and volatility. Below, we further add pre-announced jumps, explosive volatility, and additive microstructure noise. Note, however, that pure-jump processes without a Brownian term are not analyzed, because we require $\sigma_{t}$ to be strictly positive in Assumption \ref{assumption:model}.

To introduce the drift burst hypothesis, we momentarily enforce that $X$ has continuous sample paths, i.e.  $\text{d}J_{t} = 0$. The jump process is fully reactivated and driving functions are generalized in our theoretical results.

As $\mu$ and $\sigma$ are locally bounded under Assumption \ref{assumption:model}, it follows that for a fixed time point $\tau_{\textrm{db}}$:
\begin{equation} \label{equation:order_p}
\int_{ \tau_{ \textrm{db}}- \Delta}^{ \tau_{ \textrm{db}}+ \Delta} |\mu_{s}| \text{d}s = O_{p} \left( \Delta \right) \qquad \text{and} \qquad \int_{ \tau_{ \textrm{db}}- \Delta}^{ \tau_{ \textrm{db}}+ \Delta} \sigma_{s} \text{d}W_{s} = O_{p} \left( \sqrt{ \Delta} \right),
\end{equation}
as $\Delta \rightarrow 0$. Thus, the drift is much smaller than the volatility, because $\Delta \ll \sqrt{ \Delta}$. This is consistent with the notion that over short time intervals the main contributor to the log-return is volatility. It is this feature that has led the econometrics literature to largely neglect the drift.

However, the drift can prevail in an alternative model where, in a neighborhood of $\tau_{\textrm{db}}$, $\mu$ is allowed to diverge in such a way that:
\begin{equation} \label{equation:drift-burst}
\int_{ \tau_{ \textrm{db}}- \Delta}^{ \tau_{ \textrm{db}}+ \Delta} |\mu_{s}| \text{d}s = O_{p} \left( \Delta^{ \gamma_{ \mu}} \right),
\end{equation}
with $0 < \gamma_{ \mu} < 1/2$. We refer to an exploding drift coefficient as a \textit{drift burst} and to $\tau_{\textrm{db}}$ as a \textit{drift burst time}. The condition $\gamma_{ \mu} > 0$ ensures continuity of the sample path.

An example of an exploding drift leading to a drift burst is:
\begin{equation} \label{equation:model-db}
\mu_{t}^{\text{db}} = \left\{ \begin{array}{ll}a_1\left(\tau_{ \text{db}}-t\right)^{-\alpha} & t<\tau_{ \text{db}}  \\
a_2\left(t-\tau_{ \text{db}}\right)^{-\alpha} & t>\tau_{ \text{db}}
\end{array}\right.,\end{equation}
with $1/2 < \alpha < 1$ and $a_1$, $a_2$  constants. Setting $\gamma_{ \mu} = 1 - \alpha$, this formulation is consistent with Eq. \eqref{equation:drift-burst}. This specification of the drift can capture flash crashes when $a_1$ and $a_2$ are of the opposite sign (see e.g. Panel A in Figure \ref{figure:drift-burst}). It can also accommodate gradual jumps without reversion, e.g. when $a_2=0$.\footnote{The drift burst specification in Eq. \eqref{equation:model-db} also allows for gradual jumps that start off strong and then decelerate (when $a_1=0$ and $a_2\neq0$), akin to price behavior observed around, for instance, scheduled news announcements, where the first order price impact tends to be realized immediately, but it may then be followed by a gradual continuation as the market interprets and fully incorporates the shock.}

The log-price with explosive coefficients, replacing the drift in Eq. \eqref{equation:main-model} with a drift as in Eq. \eqref{equation:model-db}, is denoted by $\widetilde{X}$. As a process it is still a semimartingale, which is necessary -- but not sufficient -- to exclude arbitrage from the model \citep*[e.g.][]{delbaen-schachermayer:94a}.\footnote{While explosive drift does not impede the semimartingale structure, it can on the other hand negatively affect nonparametric estimation of volatility from high-frequency data, as unveiled by Example 3.4.2 in \citet*{jacod-protter:12a}, because the volatility is completely swamped by the drift. This is consistent with the findings of \citet*{li-todorov-tauchen:15a}, who note that standard OLS estimation of their proposed jump regression is seriously affected by the inclusion of two outliers in the sample. Incidentally, these are the equity flash crash of May 6, 2010 (in Figure \ref{figure:flash_crash}) and the hoax tweet of April 23, 2014 (in Figure \ref{figure:hoax}).} To prevent arbitrage, a further condition (imposed by Girsanov's Theorem) is necessary for the existence of an equivalent martingale measure:
\begin{equation} \label{equation:structural-condition}
\int_{\tau_{\textrm{db}}- \Delta}^{ \tau_{ \textrm{db}} + \Delta} \left( \frac{ \mu_s}{ \sigma_s} \right)^{2} \text{d}s < \infty ~(a.s.),
\end{equation}
which is known as a ``structural condition''.\footnote{The structural condition is sufficient together with the exponential moment condition $E\bigg[ \exp \Big\{- \int_{ \tau_{\textrm{db}}- \Delta}^{ \tau_{ \textrm{db}} + \Delta} \frac{\mu_{s}}{ \sigma_{s}} \text{d}W_{s} - \frac{1}{2} \int_{ \tau_{ \textrm{db}}- \Delta}^{ \tau_{\textrm{db}} + \Delta} \left( \frac{ \mu_{s}}{ \sigma_{s}} \right)^{2} \text{d}s \Big\} \bigg] = 1$, see Theorem 4.2 in \citet*[][]{karatzas-shreve:98a}.} This cannot hold if the drift explodes in the neighborhood of $\tau_{\textrm{db}}$, but the volatility remains bounded, as it allows for a so-called ``free lunch with vanishing risk,'' see Definition 10.6 in \citet*{bjork:03a}. Thus, explosive volatility is a necessary condition for drift bursts in a market free of arbitrage.

We say there is a volatility burst, if
\begin{equation} \label{equation:volatility_burst}
\int_{\tau_{\textrm{db}}- \Delta}^{\tau_{\textrm{db}}+ \Delta} \sigma_{s} \text{d}W_{s} = O_{p} \left( \Delta^{ \gamma_{ \sigma}} \right),
\end{equation}
with $0<\gamma_{\sigma}<1/2$. As above, an example of a bursting volatility is:
\begin{equation} \label{equation:model-vb}
\sigma_{t}^{\text{vb}} =  b \left|\tau_{ \text{db}}-t\right|^{-\beta}
\end{equation}
with $0 < \beta < 1/2$ and $b>0$. We here restrict $\beta$ to ensure that $\int_{\tau_{\textrm{db}}- \Delta}^{ \tau_{ \textrm{db}} + \Delta}  \left( \sigma_{s}^{ \text{vb}} \right)^{2} \text{d}s < \infty$, so that the stochastic integral in Eq. \eqref{equation:volatility_burst} can be defined. In this case,
\begin{equation}
\int_{\tau_{\textrm{db}}- \Delta}^{\tau_{\textrm{db}}+ \Delta}  \sigma_{s}^{\text{vb}} \text{d}W_{s} = O_{p} \left( \Delta^{1/2-\beta} \right),
\end{equation}
so that $\gamma_{ \sigma} = 1/2- \beta$.

We can thus introduce a ``canonical'' alternative model:
\begin{equation}\label{equation:canonical}
\text{d}\widetilde{X}_{t} =  \mu_{t}^{\text{db}} \text{d}t + \sigma_{t}^{\text{vb}}\text{d}W_{t},
\end{equation}
for which $\mu_{t}^{\text{db}} / \sigma_{t}^{\text{vb}} \rightarrow \infty$ as $t \rightarrow \tau_{ \text{db}}$ if $\alpha > \beta$. The structural condition in Eq. \eqref{equation:structural-condition} is readily fulfilled when $\alpha - \beta <  1/2$. Thus, this example shows that the drift coefficient can explode locally (even after normalizing by an exploding volatility) either preserving absence of arbitrage (when $0< \alpha - \beta < 1/2$), or allowing local arbitrage opportunities (when $\alpha - \beta> 1/2$). However, as our econometric analysis in Theorem \ref{theorem:alternative} reveals, we are only able to detect an explosive drift with $\alpha - \beta > 1/2$, that is when a short-lived absence of arbitrage is found. The intuition for this result is that under the condition in Eq. \eqref{equation:structural-condition} we can always switch to an equivalent probability measure with no drift and unaltered volatility. Hence, there is no way to detect an explosive drift in the arbitrage-free setting.

In Panel B of Figure \ref{figure:drift-burst}, we show simulated sample paths of the associated log-price in this framework. A flash crash is generated from Eq. \eqref{equation:canonical}. We set the drift burst rate at $\alpha = 0.65$ and $\alpha = 0.75$ with a volatility burst parameter $\beta = 0.2$. While the former parametrization preserves absence of arbitrage (since $\alpha- \beta<1/2$), the latter does not. Visually, however, the price dynamics in both scenarios is pronounced and qualitatively identical.

The notion that volatility can burst during market turbulence or dislocation is uncontroversial. \citet*{kirilenko-kyle-samadi-tuzun:17a} and \citet*{andersen-bondarenko-kyle-obizhaeva:15a} report elevated levels of volatility during the equity flash crash \citep*[see also][]{bates:18a}. However, as illustrated by Figure \ref{figure:drift-burst} and proved formally in Theorem \ref{theorem:alternative}, a volatility burst in itself is not sufficient to capture the gradual jump or flash crash dynamics regularly observed in practice. The introduction of a separate drift burst component as we propose in this paper is a convenient and effective tool to reconcile continuous-time semimartingale theory with such empirical observations.

\begin{figure}[t!]
\begin{center}
\caption{Illustration of a log-price with a drift burst. \label{figure:drift-burst}}
\begin{tabular}{cc}
\small{Panel A: Drift coefficient.} & \small{Panel B: Simulated log-return.} \\
\includegraphics[height=0.4\textwidth,width=0.48\textwidth]{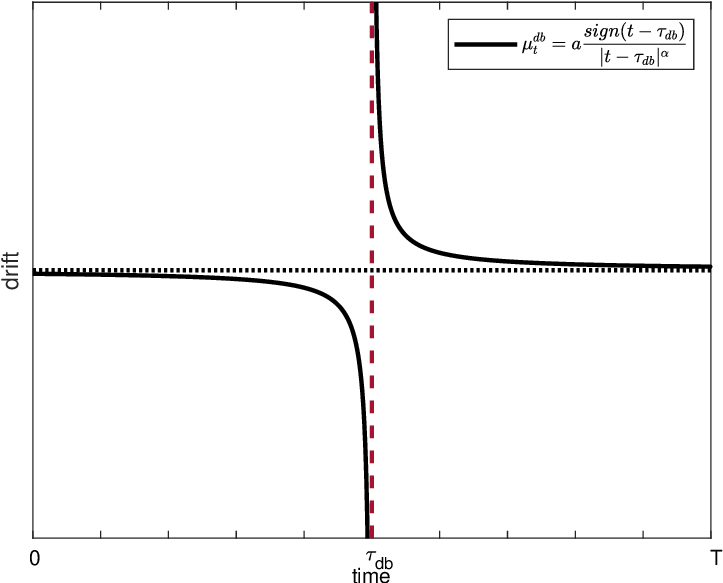} &
\includegraphics[height=0.4\textwidth,width=0.48\textwidth]{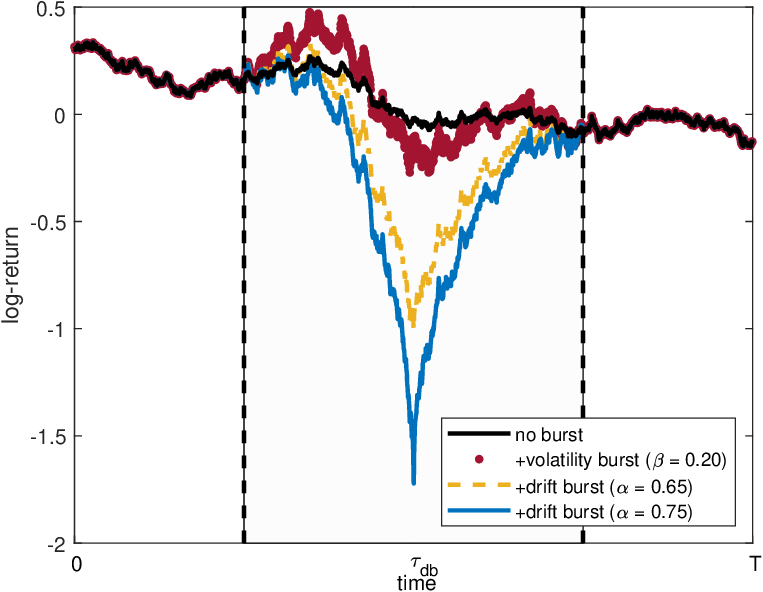} \\
\end{tabular}
\begin{scriptsize}
\parbox{\textwidth}{\emph{Note.} In Panel A, the drift coefficient is shown against time, while Panel B shows the evolution of a simulated log-price with a burst in: (i) nothing, (ii) volatility, and (iii) drift and volatility. The latter is based on Eq. \eqref{equation:model-db} and \eqref{equation:model-vb} with $-a_{1} = a_{2} = 3$, $b = 0.15$, $\alpha = 0.65$ or 0.75 and $\beta = 0.2$.}
\end{scriptsize}
\end{center}
\end{figure}

\section{Identification} \label{section:identification}

We now develop a nonparametric approach to detect drift bursts in real data. We propose a test statistic that exploits the message of Eq. \eqref{equation:drift-burst}, namely if there is a drift burst in the price process at time $\tau_{\text{db}}$, the drift can prevail over volatility and locally dominate log-returns in the vicinity of $\tau_{\text{db}}$. The test statistic thus compares a suitably rescaled estimate of $\mu_{t} / \sigma_{t}$ based on high-frequency data in a neighborhood of $t$. Later, we prove that our ``signal-to-noise'' measure uncovers drift bursts if they are sufficiently strong.

We extend existing work on nonparametric kernel-based estimation of the coefficients of diffusion processes to estimate $\mu_{t}$ and $\sigma_{t}$ \citep[e.g.][]{bandi:02a,kristensen:10a}. We assume that $X$ is recorded at times $0 = t_{0} < t_{1} < \ldots < t_{n} = T$, where $\Delta_{i,n} = t_{i} - t_{i-1}$ is the time gap between observations and $T$ is fixed. The sampling times are potentially irregular, as formalized in Assumption \ref{assumption:times} in Appendix \ref{appendix:proof}. The discretely sampled log-return over $[t_{i-1},t_{i}]$ is $\Delta_{i}^{n} X = X_{t_{i}} - X_{t_{i-1}}$, and we define:
\begin{equation}
\label{equation:drift_estimate_X}
\hat{ \mu}_{t}^{n} = \frac{1}{h_{n}} \sum_{i = 1}^{n} K \left( \frac{t_{i-1}-t}{h_{n}} \right) \Delta_{i}^{n} X, \quad \text{for }
t \in (0,T],
\end{equation}
where $h_{n}$ is the bandwidth of the mean estimator and $K$ is a kernel. We also set:
\begin{equation} \label{equation:variance_estimate_X}
\hat{ \sigma}_{t}^{n} = \left( \frac{1}{h_{n}'} \sum_{i = 1}^{n} K \left( \frac{t_{i-1} - t}{h_{n}'} \right) \left( \Delta_{i}^{n} X \right)^{2} \right)^{1/2}, \quad \text{for } t \in (0,T],
\end{equation}
where $h_{n}'$ is the bandwidth of the volatility estimator.

The bandwidths $h_{n}$, $h_{n}'$ and the kernel $K$ are assumed to fulfill some weak regularity conditions that are succinctly listed in Assumption \ref{assumption:kernel} in Appendix \ref{appendix:proof}.

In absence of a drift burst, the proof of Theorem \ref{theorem:null} stated below shows that, as $n \rightarrow \infty$:
\begin{equation}
\label{equation:null_distribution}
\sqrt{h_{n}} \left( \hat{ \mu}_{t}^{n} - \mu_{t-}^{*} \right) \overset{d}{ \rightarrow}
N \left(0, K_{2} \sigma_{t-}^{2} \right),
\end{equation}
where $\mu_{t}^{*} = \mu_{t} + \int_{ \mathbb{R}} \delta(t,x) I_{ \left\{| \delta(t,x)| > 1 \right\}} \lambda( \text{d}x)$, $K_{2}$ is a kernel-dependent constant, and the above convergence is stable in law.

As shown by Eq. \eqref{equation:null_distribution}, $\hat{ \mu}_{t}^{n}$ is asymptotically unbiased for the (jump compensated) drift term. It is inconsistent, because the variance explodes as $h_{n} \rightarrow 0$. This appears to rule out drift burst detection via $\hat{ \mu}_{t}^{n}$. On the other hand, if we rescale the left-hand side of Eq. \eqref{equation:null_distribution} with $\hat{ \sigma}_{t}^{n}\sqrt{K_{2}}$,  it appears the right-hand side has a standard normal distribution.\footnote{Lemma \ref{lemma:spot-volatility} in Appendix \ref{appendix:proof} shows that $\hat{ \sigma}_{t}^{n}$ is a consistent estimator of $\sigma_{t-}$.} It is this insight that facilitates the construction of a test statistic that can identify drift bursts, as we prove in Theorem \ref{theorem:null}, which describes the behavior of the $t$-statistic under the null of locally bounded coefficients, whereas Theorem \ref{theorem:alternative} does it under the alternative of explosive drift and volatility.

The test statistic is defined as:
\begin{equation} \label{equation:t-statistic}
T_{t}^{n} = \sqrt{ \frac{h_{n}}{K_{2}}} \frac{\hat{ \mu}_{t}^{n}}{
\hat{ \sigma}^{n}_{t}}.
\end{equation}
$T_{t}^{n}$ has an intuitive interpretation with the indicator kernel. Here, it is the ratio of the drift part to the volatility part of the log-return over the interval $[t-h_{n},t]$ (as $h_{n}, h_{n}' \rightarrow 0$, this holds for any valid kernel).

\begin{theorem}
\label{theorem:null} Assume that $X$ is a semimartingale as defined by Eq. \eqref{equation:main-model}, and that Assumption \ref{assumption:model} and \ref{assumption:model2} -- \ref{assumption:kernel} are fulfilled. As $n \rightarrow \infty$, it holds that:
\begin{equation}
\label{equation:H0}
T_{t}^{n} \overset{d}{ \rightarrow} N(0,1).
\end{equation}
\end{theorem}
\begin{proof}
See Appendix \ref{appendix:proof}. \qed
\end{proof}
Theorem \ref{theorem:null} shows that, in absence of a drift burst, the $t$-statistic in Eq. \eqref{equation:t-statistic} has a limiting standard normal distribution.\footnote{While this statement appears to follow trivially from Eq. \eqref{equation:null_distribution} -- i.e., via application of Slutsky's Theorem -- this is not true. In general, we can only use Eq. \eqref{equation:null_distribution} to deduce Eq. \eqref{equation:H0}, if $\sigma_{t-}$ is a constant. In our paper, where $\sigma_{t-}$ is a random variable, the definition of convergence in distribution does not support such a conclusion. We therefore prove in Appendix \ref{appendix:proof} that the convergence in Eq. \eqref{equation:null_distribution} is in law stably, which is a stronger form of convergence that helps to recover this feature \citep*[the concept is explained in e.g.][]{jacod-protter:12a}. Moreover, we allow for leverage effects and jumps. In both these directions, Theorem \ref{theorem:null} extends \citet{kristensen:10a}.} We note that under the null, $T_{t}^{n}$ does not depend on $\mu_{t-}^{*}$ and $\sigma_{t-}$ for large $n$. Thus, although it is not possible to consistently estimate $\mu_{t-}^{*}$, we can exploit its asymptotic distribution to form a test of the drift burst hypothesis, since a large $t$-statistic signals that the realized log-return is mostly induced by drift.\footnote{\citet*{todorov-tauchen:14a} build a related type of self-normalized statistic to test whether the price process is jump-diffusion or of pure-jump type. Their procedure has power against the pure-jump alternative. It is therefore important in future work to establish the properties of our procedure in a pure-jump setting to ensure it has size control in such models.}

The drift burst alternative is formalized by an exploding $\mu_{t}$ term. We note that the alternative is broad enough to allow $\sigma_{t}$ to co-explode with the drift.

\begin{theorem} \label{theorem:alternative} Assume that $\widetilde{X}$ is of the form:
\begin{equation}
\label{equation:alternative}
\text{\upshape{d}} \widetilde{X}_{t} =  \text{\upshape{d}} X_{t} +  \frac{c_{1,t}}{ (\tau_{\text{\upshape{db}}} - t)^{ \alpha}} \text{\upshape{d}}t + \frac{c_{2,t}}{ (\tau_{\text{\upshape{db}}} - t)^{ \beta}} \text{\upshape{d}}W_{t},
\end{equation}
where $\tau_{\text{\upshape{db}}}>0$, $\text{\upshape{d}}X_{t}$ is the model with locally bounded coefficients in Eq. \eqref{equation:main-model} for which the conditions of Theorem \ref{theorem:null} hold, $c_{1,t}$ and $c_{2,t}$ are adapted stochastic processes adhering to identical conditions as for $\mu_{t}$ and $\sigma_{t}$ listed in Assumption \ref{assumption:model2}. Moreover, $\alpha$ and $\beta$ are constants such that $0 \leq \beta < 1/2$ and $0 < \alpha < 1$.
Then, as $n \rightarrow \infty$, it holds that:
\begin{equation}
T_{ \tau_{\text{\upshape{db}}}}^{n}\left\{
\begin{array}{lc}
\overset{a.s.}{ \rightarrow} \pm \infty, & \text{if} \quad \alpha - \beta > 1/2, \\[0.10cm]
\overset{d}{\rightarrow} c_{K, \beta} N(0,1) + d_{K, \beta, c_{1}, c_{2}}, & \text{if} \quad \alpha - \beta = 1/2, \\[0.10cm]
\overset{d}{\rightarrow} c_{K, \beta} N(0,1), & \text{if} \quad \alpha - \beta < 1/2,
\end{array}
\right.
\end{equation}
where
\begin{equation}
c_{K, \beta} = \sqrt{ \frac{ \int_{ \mathbb{R}}K^{2}(x)|x|^{-2\beta} \text{\upshape{d}}x}{K_{2} \int_{\mathbb{R}}K(x)|x|^{-2 \beta} \text{\upshape{d}}x}} \quad \text{and} \quad d_{K, \beta,c_{1},c_{2}} = \frac{c_{1, \tau_{\text{\upshape{db}}}}}{c_{2, \tau_{\text{\upshape{db}}}}} \frac{ \int_{ \mathbb{R}}K^2(x)|x|^{- \beta-1/2} \text{\upshape{d}}x}{ \sqrt{ K_{2} \int_{ \mathbb{R}}K(x)|x|^{-2 \beta} \text{\upshape{d}}x}}.
\end{equation}
\end{theorem}

\begin{proof}
See Appendix \ref{appendix:proof}. \qed
\end{proof}
Theorem \ref{theorem:alternative} implies the explosion of the $t$-statistic under the alternative, when the drift term explodes fast enough relative to the volatility (i.e. $\alpha- \beta> 1/2$). The condition $\alpha- \beta> 1/2$ is equivalent to require that, in a neighborhood of $\tau_{ \text{\upshape{db}}}$, the log-return is dominated by drift. In a frictionless economy this allows for a short-lived arbitrage around $\tau_{ \text{\upshape{db}}}$. In practice, to make the test statistic large, it is of course enough that the mean log-return is significantly nonzero over a small time interval. Our simulations confirm that the condition  $\alpha- \beta>1/2$ is not needed to achieve power in small samples.\footnote{As detailed in Appendix \ref{appendix:parametric}, we can estimate $\alpha$ and $\beta$ using a parametric maximum likelihood approach based on Eq. \eqref{equation:canonical}. The sample averages across detected events in the empirical high-frequency data analyzed in Section \ref{section:empirical} are $\bar{ \hat{ \alpha}}_{ \text{ML}} = 0.6250$ and $\bar{ \hat{ \beta}}_{ \text{ML}} = 0.1401$. Hence, assuming that bursts are of the form in Eq. \eqref{equation:canonical}, the process is found to be right at the margin of being arbitrage-free.}

An implication of Theorem \ref{theorem:alternative} is that $T_{ \tau_{\text{\upshape{db}}}}^{n}$ does not explode because of a volatility burst, even if it occurs without a drift burst. With $\alpha - \beta < 1/2$ the $t$-statistic is normally distributed in the limit with a controllable asymptotic variance that depends on the kernel shape and $\beta$. With a left-sided exponential kernel $K(x) = \exp(-|x|)$, for $x \leq 0$, adopted below, $c_{K, \beta} =2^{ \beta} \leq \sqrt{2}$. At the border, $\alpha - \beta=1/2$, the $t$-statistic is bounded but can take arbitrarily large values as $\alpha \rightarrow 1$ and $\beta \rightarrow 1/2$, because $\lim_{ \beta \rightarrow 1/2} d_{K, \beta, c_{1}, c_{2}} = \infty$.\footnote{\citet*{hengartner-linton:96a} study nonparametric regression estimation at design poles and zeros. A nearly identical set of ``bias'' and ``variance'' kernel constants to $d_{K, \beta, c_{1}, c_{2}}$ and $c_{K, \beta}$ also appear in the asymptotic distribution of their estimator.}

Thus, the large values of the test statistic frequently observed in the empirical application are neither explained by volatility bursts nor by jumps. We conclude that -- in our setting -- a significant $t$-statistic can only be induced by a drift explosion.\footnote{A likelihood ratio test based on the model in Eq. \eqref{equation:canonical}, also reported in Appendix \ref{appendix:parametric}, strongly rejects $\mathcal{H}_{0}^{ \prime} : \beta = 0$ against $\mathcal{H}_{1}^{ \prime} : \beta > 0$ during a drift burst, yielding empirical support for the volatility co-exploding with the drift.} These theoretical statements are also corroborated by our simulation analysis.

\begin{remark} \label{remark:alternative}
In this paper, we are testing against drift explosions. However, other mechanisms can possibly cause the $t$-statistic to deviate from the asymptotic normal distribution. First, pure-jump processes without a Brownian component can be a potential alternative \citep*[e.g.][]{todorov-tauchen:14a}. Second, it is possible to work with the persistence of jumps in a Hawkes-type process to construct a ``jump burst.'' Third, one can develop a market microstructure model with endogenous noise that depends on the fundamental value of the asset \citep*[e.g.][]{li-linton:20a}. These ideas to describe the departures from the jump-diffusion It\^{o} semimartingale framework that our paper reveals, via the drift burst hypothesis, constitute promising avenues for future research. In the end, it is the economic application that dictates how we label such deviations.
\end{remark}

\subsection{Inference via the maximum statistic}

The asymptotic theory asserts that $T_{t}^{n}$ is standard normal in absence of a drift burst, whereas it grows arbitrarily large under the alternative, as we approach a drift burst time. It suggests that a viable detection strategy is to compute the $t$-statistic progressively over time and reject the null when $|T_{t}^{n}|$ gets significantly large. This leads to a multiple testing problem, which can cause size distortions, if the quantile function of the standard normal distribution is used to determine a critical value of $T_{t}^{n}$.

To control the family-wise error rate, we evaluate a standardized version of the maximum of the absolute value of our $t$-statistic using extreme value theory.\footnote{\citet*{bajgrowicz-scaillet-treccani:16a} and \citet*{lee-mykland:08a} also exploit these ideas in the high-frequency framework to devise an unbiased jump-detection test, while in a related context \citet*{andersen-bollerslev-dobrev:07a} propose a Bonferroni correction. The latter was another viable tool to avoid systematic overrejection of the null hypothesis.} We compute $\left(T_{t_{i}^{*}}^{n} \right)_{i=1}^{m}$ at $m$ equispaced time points $t_{i}^{*} \in (0,T]$, where $T$ is fixed. We set:
\begin{equation} \label{equation:maxabsTn}
T_{m}^{*} = \max_{t_{i}^{*}} |T_{t_{i}^{*}}^{n}|, \quad i = 1, \ldots, m.
\end{equation}
The crucial point is that, in addition to $T_{t_{i}^{*}}^{n} \overset{d}{ \rightarrow} N(0,1)$ under the null, the $T_{t_{i}^{*}}^{n}$'s are also independent -- up to error terms that are asymptotically negligible -- if $m$ does not grow too fast. It follows that a normalized version of $T_{m}^{*}$ has a limiting Gumbel distribution, as $m \rightarrow \infty$ at a suitable rate.

\begin{theorem} \label{theorem:gumbel-distribution}
Assume that the conditions of Theorem \ref{theorem:null} hold. Then, if $n \rightarrow \infty$, $m \rightarrow \infty$ such that $mh_{n} \rightarrow 0$ and $\sqrt{ \log(m)} \bigg( \frac{m}{ \sqrt{nh_{n}}}+(mh_{n})^{-B} \sqrt{ \log(m)} + m^{- \Gamma/2} \sqrt{ \log(m)} \bigg) \rightarrow 0$, it further holds that:
\begin{equation}
\label{equation:gumbel}
(T_{m}^{*} - b_{m})a_{m} \overset{d}{ \rightarrow} \xi,
\end{equation}
where
\begin{equation}
a_{m} = \sqrt{2 \log(m)}, \qquad b_{m} = a_{m} - \frac{1}{2} \frac{ \log(\pi \log(m))}{a_{m}},
\end{equation}
and the CDF of $\xi$ is the Gumbel, i.e. $P( \xi \leq x) = \exp( -\exp(-x))$.
\end{theorem}

\begin{proof}
See Appendix \ref{appendix:proof}. \qed
\end{proof}

The alternative to Theorem \ref{theorem:gumbel-distribution} is that there is at least one drift burst (as formalized in Theorem \ref{theorem:alternative}) in the time interval $[0,T]$, whence the maximum statistic $T_{m}^{*}$ diverges rendering the test consistent.

In practice, the Gumbel distribution is conservative, as the convergence in \eqref{equation:gumbel} is slow and because of residual dependence in the $t$-statistics due to small sample effects, microstructure noise and pre-averaging (introduced in Section \ref{section:noise}). In Appendix \ref{appendix:critical-value}, we propose a simulation-based procedure to determine data-driven critical values.

\subsection{Robustness to a pre-announced jump} \label{section:fixed-jump}

We here study an extended model that -- on top of the drift, volatility and jump component in Eq. \eqref{equation:main-model} -- has a ``pre-announced'' jump \citep*[see e.g.][]{jacod-li-zheng:17a,dubinsky-johannes-kaeck-seeger:18a}, where the jump time is fixed across sample paths. We show these types of jumps do not compromise drift burst detection.
\begin{theorem}
\label{theorem:fixed-jump} Assume that $\doublewidetilde{X}$ is of the form:
\begin{equation} \label{equation:fixed-jump}
\text{\upshape{d}} \doublewidetilde{X}_{t} =  \text{\upshape{d}}X_{t} + \text{\upshape{d}}J_{t}',
\end{equation}
where $\text{\upshape{d}}X_{t}$ is the model in Eq. \eqref{equation:main-model} such that the conditions of Theorem \ref{theorem:null} hold, while $J_{t}' =  J \cdot I_{\{0< \tau_{J} \leq t\}}$, $\tau_{J}$ is a stopping time, and $J$ is $\mathcal{F}_{ \tau_{J}}$-measurable. Then, as $n \rightarrow \infty$, it holds that:
\begin{equation}
\displaystyle T_{ \tau_{J}}^{n} \overset{p}{ \rightarrow} \sqrt{ \frac{K(0)}{K_{2}}} \cdot \sign(J).
\end{equation}
\end{theorem}
\begin{proof}
See Appendix \ref{appendix:proof}. \qed
\end{proof}
We can readily select a kernel that can tell apart the occurrence of a fixed jump from a drift explosion. In particular, the left-sided exponential kernel advocated above has $\displaystyle \sqrt{ K(0)/K_{2}} = \sqrt{2}$, so that $|T_{ \tau_{J}}^{n}| \overset{p}{ \rightarrow} \sqrt{2}$. Thus, our proposed $t$-statistic is -- asymptotically -- small under the null (standard normal distributed) and pre-announced jump alternative ($\sqrt{2}$ in absolute value), while it is large (diverging) under the drift burst alternative.

\subsection{Robustness to microstructure noise} \label{section:noise}

In practice, we do not measure the true, efficient log-price $X_{t_{i}}$, because transaction and quotation data are disrupted by multiple layers of ``noise'' or ``friction'' \citep*[e.g.][]{black:86a,stoll:00a}. In this section, we show how to modify our test for drift bursts, so it is resistant to such features of the market microstructure at the tick level.
To incorporate noise, we suppose that:
\begin{equation}
Y_{t_{i}} = X_{t_{i}} + \epsilon_{t_{i}}, \qquad \text{for } i = 0, 1, \ldots, n,
\end{equation}
where $\epsilon_{t_{i}}$ is an additive error term.
\begin{assumption} \label{assumption:noise}
$( \epsilon_{t_{i}})_{i=0}^{n}$ is adapted and independent of $X$. Moreover, $E[ \epsilon_{t_{i}}] = 0$, $E \big[ \epsilon_{t_{i}}^{4} \big] < \infty$, and denoting the autocovariance function by $\gamma_{k} = E[\epsilon_{t_{i}}\epsilon_{t_{i+k}}]$ for any integer $k\geq 0$, we further assume $\gamma_{k}$ is finite, independent of $i$ and $n$, and such that $\gamma_{k} = 0$ for $k > Q$, where $Q \geq 0$ is an integer (i.e., $Q$-dependent noise).
\end{assumption}

This is a standard noise model in financial econometrics. It allows for autocorrelation in the noise process, but it rules out dependence between the noise and fundamental price, as in \citet*{li-linton:20a}. The difficulty brought by noise is then that in order to do inference about drift bursts in $X$, we are forced to work with the contaminated high-frequency record of $Y$.

A direct application of the $t$-statistic in Eq. \eqref{equation:t-statistic} to the noise-contaminated returns $\Delta_{i}^{n} Y$ is powerless, because the noise asymptotically dominates the other shocks and overwhelms the signal of an exploding drift. A solution to this problem is to slightly slow down the accumulation of noise by pre-averaging $Y_{t_{i}}$, as in \citet*{jacod-li-mykland-podolskij-vetter:09a}. We define a pre-averaged increment for any stochastic process $V$:
\begin{equation} \label{equation:preavg-Y}
\Delta_{i}^{n} \overbar{V} = \sum_{j=1}^{k_{n}-1}g_{j}^{n}  \Delta_{i+j}^{n} V = - \sum_{j=0}^{k_{n}-1}H_{j}^{n}V_{t_{i+j}},
\end{equation}
where $k_{n}$ is the pre-averaging window, $g_{j}^{n} = g(j/k_{n})$ and $H_{j}^{n} = g_{j+1}^{n}-g_{j}^{n}$ with $g: [0,1] \mapsto \mathbb{R}$ continuous and piecewise continuously differentiable with a piecewise Lipschitz derivative $g^{\prime}$, such that $g(0)=g(1)=0$ and $\int_{0}^{1}g^{2}(s) \text{d}s< \infty$.
Absent a drift burst, it follows that:
\begin{equation} \label{equation:order}
\Delta_{i}^{n} \overbar{X} = O_{p} \left( \sqrt{ \frac{k_{n}}{n}} \right) \qquad \text{and} \qquad \Delta_{i}^{n} \overbar{ \epsilon} = O_{p} \left( \frac{1}{\sqrt{k_{n}}} \right),
\end{equation}
As Eq. \eqref{equation:order} shows, the noise is reduced by a factor $\sqrt{k_{n}}$. The drift and volatility of $X$ are enhanced by $\sqrt{k_{n}}$, while leaving their relative order unchanged. Intuitively, it therefore suffices with minimal pre-averaging to bring down the noise enough and make a fast drift burst with $\alpha$ close to one dominate the divergence of the asymptotic variance in the drift estimator.

The drift burst $t$-statistic in Eq. \eqref{equation:t-statistic} is then refined as a noise-robust version, where the drift estimator and its long-run variance are computed from the pre-averaged return series:
\begin{equation} \label{equation:t-statistic-robust}
\overbar{T}_{t}^{n} = \sqrt{h_n} \frac{\hat{ \overbar{ \mu}}_{t}^{n}}{
\sqrt{\hat{ \overbar{ \sigma}}^{n}_{t}}},
\end{equation}
with
\begin{equation} \label{equation:mu-noise}
\hat{ \overbar{ \mu}}_{t}^{n} = \frac{1}{h_{n}} \sum_{i=1}^{n-k_{n}+2} K \left( \frac{t_{i-1}-t}{h_{n}} \right) \Delta_{i-1}^{n} \overbar{Y},
\end{equation}
and
\begin{equation} \label{equation:SigmatHAC}
\hat{ \overbar{ \sigma}}_{t}^{n} = \frac{1}{h_{n}'} \left[ \sum_{i=1}^{n-k_{n}+2} \left( K \left( \frac{t_{i-1}-t}{h_{n}'} \right) \Delta_{i-1}^{n} \overbar{Y} \right)^{2} + 2 \sum_{L=1}^{L_{n}}w \left( \frac{L}{L_{n}} \right) \sum_{i=1}^{n-k_{n}-L+2}K \left( \frac{t_{i-1}-t}{h_{n}'} \right)K \left( \frac{t_{i+L-1}-t}{h_{n}'} \right) \Delta_{i-1}^{n} \overbar{Y} \Delta_{i-1+L}^{n} \overbar{Y} \right],
\end{equation}
where $w: \mathbb{R}_{+} \rightarrow \mathbb{R}$ is a kernel with $w(0)=1$ and $w(x) \rightarrow 0$ as $x \rightarrow \infty$, and $L_{n}$ is the lag length that determines the number of autocovariances in \eqref{equation:SigmatHAC}.
$\hat{ \overbar{ \sigma}}_{t}^{n}$ is a heteroscedasticity and autocorrelation consistent (HAC)-type statistic \citep*[e.g.][]{newey-west:87a, andrews:91a}. The extra complexity is required to account for any noise dependence and the serial correlation induced by pre-averaging to consistently estimate the asymptotic variance of $\hat{ \overbar{ \mu}}_{t}^{n}$.

\begin{theorem} \label{theorem:null-noise}
Set $Y_{t_{i}} = X_{t_{i}} + \epsilon_{t_{i}}$, where $X$ is defined by Eq. \eqref{equation:main-model} and $\epsilon$ is defined by Assumption \ref{assumption:noise}. Suppose that Assumption \ref{assumption:model} -- \ref{assumption:kernel} are fulfilled. For every fixed $t \in (0,T]$, as $n \rightarrow \infty$, $k_{n} \rightarrow \infty$, $L_{n} \rightarrow \infty$ such that $k_{n}h_{n} \rightarrow 0$, $k_{n} h_{n}' \rightarrow 0$, $\frac{k_{n}}{nh_{n}} \rightarrow 0$, $\frac{k_{n}}{nh_{n}'} \rightarrow 0$, and $\frac{L_{n}}{nh_{n}'} \rightarrow 0$, it holds that:
\begin{equation*}
\overbar{T}_{t}^{n} \overset{d}{ \rightarrow} N(0,1).
\end{equation*}
\end{theorem}
\begin{proof}
See Appendix \ref{appendix:proof}. \qed
\end{proof}

The conditions $k_{n}h_{n} \rightarrow 0$ and $\frac{k_{n}}{nh_{n}} \rightarrow 0$ (and the associated ones with $h_{n}'$) call for moderate pre-averaging, so that the number of pre-averaged terms is not too large. The condition $\frac{L_{n}}{nh_{n}'} \rightarrow 0$ means the lag length also cannot grow too fast when estimating the long-run variance of the drift estimator.

\begin{theorem} \label{theorem:alternative-noise}
Set $Y_{t_{i}} = \widetilde{X}_{t_{i}} + \epsilon_{t_{i}}$, where $\widetilde{X}$ is defined as in Theorem \ref{theorem:alternative} and everything else is maintained as in Theorem \ref{theorem:null-noise}. As $n \rightarrow \infty$, $k_{n} \rightarrow \infty$, $L_n\rightarrow\infty$ such that $k_{n} h_{n} \rightarrow 0$, $k_{n} h_{n}' \rightarrow 0$, $\frac{k_{n}}{nh_{n}} \rightarrow 0$, $\frac{k_{n}}{nh_{n}'} \rightarrow 0$, and $\frac{L_{n}}{nh_{n}'} \rightarrow 0$, and $k_{n}h_{n}^{2(1- \alpha)} \rightarrow \infty$, it holds that $| \overbar{T}_{ \tau_{ \text{\upshape{db}}}}^{n}| \overset{p}{ \rightarrow} \infty$ for $\alpha- \beta > 1/2$.
\end{theorem}
\begin{proof}
See Appendix \ref{appendix:proof}. \qed
\end{proof}

This confirms that in presence of noise the pre-averaged test statistic converges in law to a standard normal under the null of no drift burst, while it again diverges at a drift burst time, if $\alpha - \beta > 1/2$. Under the alternative, pre-averaging cannot be too moderate ($k_{n}h_{n}^{2(1- \alpha)} \rightarrow \infty$), otherwise the signal of an exploding drift is not enhanced enough relative to the noise. The condition also shows that with $\alpha$ closer to $1$, we need less pre-averaging.

\section{Simulation study} \label{section:simulation}

In this section, we adopt a Monte Carlo approach to further explore the $t$-statistic proposed in Eq. \eqref{equation:t-statistic-robust} as a tool to uncover drift bursts in $X$. The overall goal is to investigate the size and power properties of our test and figure out how ``small'' drift bursts we are able to detect under the alternative, amid also an exploding volatility, the presence of infinity-activity price jump processes, and microstructure noise.

We simulate a driftless \citet{heston:93a}-type stochastic volatility (SV) model:
\begin{align} \label{equation:heston-model}
\begin{split}
\text{d}X_{t} &= \sigma_{t} \text{d}W_{t}, \\[0.10cm]
\text{d} \sigma_{t}^{2} &= \kappa \left( \theta - \sigma_{t}^{2} \right) \text{d}t + \xi \sigma_{t} \text{d}B_{t}, \quad t \in [0,1],
\end{split}
\end{align}
where $W$ and $B$ are standard Brownian motions with $E( \text{d}W_{t} \text{d}B_{t}) = \rho \text{d}t$. Thus, the drift-to-volatility ratio of the efficient log-price is $\mu_{t}/\sigma_{t} = 0$.

We configure the variance process to match key features of real financial high-frequency data. As consistent with prior work \citep*[e.g.][]{ait-sahalia-kimmel:07a}, we assume the annualized parameters of the model are $( \kappa, \theta, \xi, \rho) = (5, 0.0225, 0.4, -\sqrt{0.5})$. We note $\theta$ implies an unconditional standard deviation of log-returns of roughly 15\% p.a., which aligns with what we observe across assets in our empirical study. A total of $\text{1,000}$ repetitions is generated via an Euler discretization. In each simulation, $\sigma_{t}^{2}$ is initiated at random from its stationary law $\sigma_{t}^{2} \sim \text{Gamma}(2 \kappa \theta \xi^{-2}, 2 \kappa \xi^{-2})$. The sample size is $n = \text{23,400}$, which is representative of the liquidity in the futures contracts analyzed in Section \ref{section:empirical} (see Table \ref{Table:summaryStatsCME}). It corresponds to second-by-second sampling in a 6.5 hours trading session.

We create drift and volatility bursts with the parametric model:
\begin{equation} \label{equation:parametric-model}
\mu_{t}^{\text{db}} = a \frac{ \sign(t- \tau_{\text{db}})}{| \tau_{\text{db}}-t|^\alpha}, \quad \quad \sigma_{t}^{\text{vb}} = b \frac{ \sqrt{\theta}}{| \tau_{\text{db}}-t|^\beta}, \quad \text{for } t \in [0.475,0.525],
\end{equation}
with $\tau_{\text{db}} = 0.5$ fixed. Here, the price experiences a short-lived flash crash at $\tau_{\text{db}}$, as consistent with our empirical finding that most of the identified drift bursts are followed by partial or full recovery.\footnote{To ensure $X$ reverts during a pure volatility burst, we recenter the log-return series associated with $\sigma_{t}^{\text{vb}}$, so that $\int_{0}^{T} \sigma_{t}^{\text{vb}} \text{d}W_{t} = 0$. This has almost no impact on the outcome of the $t$-statistic, but it makes the price processes comparable across settings.} The window $[0.475, 0.525]$ can be interpreted as making the duration of the bursts last about 20 minutes. We set $\alpha = (0.55, 0.65, 0.75)$ and $\beta = (0.1, 0.2, 0.3, 0.4)$ to gauge their impact on our $t$-statistic.\footnote{Note that as $\alpha - \beta > 1/2$ for some of these combinations, the model is not always devoid of arbitrage.} In particular, fixing $a = 3$ we induce a cumulative return $\int_{0}^{ \tau_{\text{db}}} \mu_{t}^{ \text{db}} \text{d}t$ of about $-0.5\%$ (with opposite sign after the crash) for $\alpha = 0.55$ to slightly less than $-1.5\%$ for $\alpha = 0.75$, as comparable to what we observe in the real data. Also, with $b = 0.15$ our choices of $\beta$ produce a 25\% ($\beta = 0.1$) to more than 100\% ($\beta = 0.4$) increase in the standard deviation of log-returns in the drift burst window relative to its unconditional level across simulations. A drift burst is therefore accompanied by highly elevated volatility, making it challenging to detect the signal.

To explicate the robustness of our $t$-statistic, in each simulation we superposition on top of the continuous sample path model in Eq. \eqref{equation:heston-model} a L\'{e}vy process with jump density given by:
\begin{equation}
\nu( \text{d}x) = \psi \frac{e^{- \lambda x}}{x^{1 + \upsilon}} \text{d}x,
\end{equation}
where $\upsilon > 0$ and $\psi > 0$. This defines a pure-jump tempered stable process with activity index $\upsilon$. We assume $\upsilon = 0.5$, such that an infinite-activity finite-variation process with both small and large jumps appearing randomly in $X$ is produced. We further set $\lambda = 3$ and calibrate $\psi$ so that on average 20\% of the quadratic variation is induced by the jump component. This is broadly in line with previous studies, e.g. \citet*{ait-sahalia-jacod-li:12a, ait-sahalia-xiu:16a}. The process is simulated as the difference between two positive tempered stable processes, as outlined in \citet*{todorov-tauchen-grynkiv:14a}. These are generated with the acceptance-rejection algorithm of \citet*{baeumer-meerschaert:10a}. Note that the discretization is exact for the selected value of $\upsilon$. Also, in agreement with our theoretical exposition the jump process is active both under the null and alternative hypothesis.

The noisy log-price is:
\begin{equation}
Y_{i/n} = X_{i/n} + \epsilon_{i/n}, \quad i = 0, 1, \ldots, n,
\end{equation}
where $\epsilon_{i/n} \sim N \left(0, \omega_{i/n}^{2} \right)$ and $\displaystyle \omega_{i/n} = \gamma \frac{ \sigma_{i/n}}{ \sqrt{n}}$, so the noise is both conditionally heteroscedastic, serially dependent (via $\sigma$), and positively related to the riskiness of the efficient log-price \citep*[e.g.][]{bandi-russell:06a, oomen:06a, kalnina-linton:08a}. We should note that this is a more general microstructure noise than allowed by Assumption \ref{assumption:noise}. $\gamma$ is the noise-to-volatility ratio. We set $\gamma = 0.5$, which amounts to a medium contamination level \citep*[e.g.][]{christensen-oomen-podolskij:14a}. To reduce the noise, we pre-average $Y_{i/n}$ locally within a block of length $k_{n} = 3$ and based on the weight function $g(x)=\min(x,1-x)$.\footnote{In the Online Appendix, we present a comprehensive analysis with $\gamma = 0.5, 2$ and $5$ and pre-averaging horizon $k_{n} = 1, \ldots, 10$. The results do not differ materially from those reported here. A modest loss of power is noted, however, if $h_{n}$ is small and $k_{n}$ is large.}\textsuperscript{,}\footnote{With equidistant data, it follows that if $k_{n}$ is even and $g(x) = \min(x,1-x)$, the pre-averaged return in Eq. \eqref{equation:preavg-Y} can be rewritten as $\Delta_{i}^{n} \bar{Y} = \frac{1}{k_{n}} \sum_{j=1}^{k_{n}/2} Y_{ \frac{i+k_{n}/2+j}{n}} - \frac{1}{k_{n}} \sum_{j=1}^{k_{n}/2} Y_{ \frac{i+j}{n}}$. Thus, the sequence $(2 \Delta_{i}^{n} \bar{Y})_{i=1}^{n-k_{n}+2}$ can be interpreted as constituting a new set of increments from a price process that is constructed by averaging of the rescaled noisy log-price series, $(Y_{i/n})_{i=0}^{n}$, in a neighbourhood of $i/n$, thus making the use of the term pre-averaging and the associated notation transparent.} $\hat{ \overbar{ \mu}}_{t}^{n}$ and $\hat{ \overbar{ \sigma}}_{t}^{n}$ are constructed from $( \Delta_{i}^{n} \overbar{Y})_{i=0}^{n-2k_{n}+1}$ based on Eq. \eqref{equation:mu-noise} and \eqref{equation:SigmatHAC} with a left-sided exponential kernel $K(x) = \exp(-|x|)$, for $x \leq 0$.

A Parzen kernel is selected for $w$:
\begin{equation} \label{equation:parzen}
w(x) =
\begin{cases}
1 - 6x^{2} + 6|x|^{3}, & \text{for } 0 \leq |x| \leq 1/2, \\[-0.25cm]
2(1 - |x|)^{3}, & \text{for } 1/2 < |x| \leq 1, \\[-0.25cm]
0, & \text{otherwise}.
\end{cases}
\end{equation}
This choice has some profound advantages in our framework. First, the Parzen kernel ensures that $\hat{ \overbar{ \sigma}}_{t}^{n}$ is positive, so we can always compute the $t$-statistic, which is not true for a general weight function. Secondly, the efficiency of the Parzen kernel is near-optimal, e.g. \citet*{andrews:91a, barndorff-nielsen-hansen-lunde-shephard:09a}. The slight loss of efficiency brings the distinct merit that $\hat{ \overbar{ \sigma}}_{t}^{n}$ can be computed on the back of the first $L_{n}$ lags of the autocovariance function, whereas more efficient weight functions typically require all $n$ lags. In the high-frequency framework $n$ is large, so the latter can be prohibitively slow to compute. In contrast, $L_{n}$ is typically small compared to $n$ in practice, rendering our choice of kernel much less time-consuming.

We set $L_{n} = Q^{*} + 2(k_{n}-1)$ to estimate $\hat{ \overbar{ \sigma}}_{t}^{n}$. Here, $2(k_{n}-1)$ is due to pre-averaging and we compute $Q^{*}$ from $( \Delta_{i}^{n} Y)_{i=1}^{n}$ as a data-driven measure of noise dependence based on automatic lag selection \citep*[see, e.g.][]{newey-west:94a, barndorff-nielsen-hansen-lunde-shephard:09a}.\footnote{In our simulations, the average value of $Q^{*}$ is 11.9, while its interquartile range is 8 -- 17.} The bandwidth for $\hat{ \overbar{ \mu}}_{t}^{n}$ is varied in $h_{n} = (120,300,600)$ seconds. We use a larger bandwidth of $h_{n}' = 5h_{n}$ for $\hat{ \overbar{ \sigma}}_{t}^{n}$ to better capture persistence in volatility and estimate the microstructure-induced return variation.

A new value of $\overbar{T}_{t}^{n}$ is recorded at every 60th transaction update.\footnote{We set a burn-in period of a full volatility bandwidth before testing to allow for a sufficient number of observations to construct $\overbar{T}_{t}^{n}$.} We extract $T_{m}^{*} = \max_{i = 1, \ldots, m} | \overbar{T}_{t_{i}^*}^{n}|$ based on the resulting $m = 341$ tests in each sample. The simulation-based approach explained in Appendix \ref{appendix:critical-value} is adopted to find a critical value of $T_{m}^{*}$.

\begin{figure}
\caption{Q--Q plot of $\overbar{T}_{t}^{n}$ without drift burst.\label{figure:t-statistic}}
\begin{center}
\begin{tabular}{cc}
\small{Panel A: No burst.} & \small{Panel B: Volatility burst ($\beta = 0.4$).} \\
\includegraphics[height=0.4\textwidth,width=0.475\textwidth]{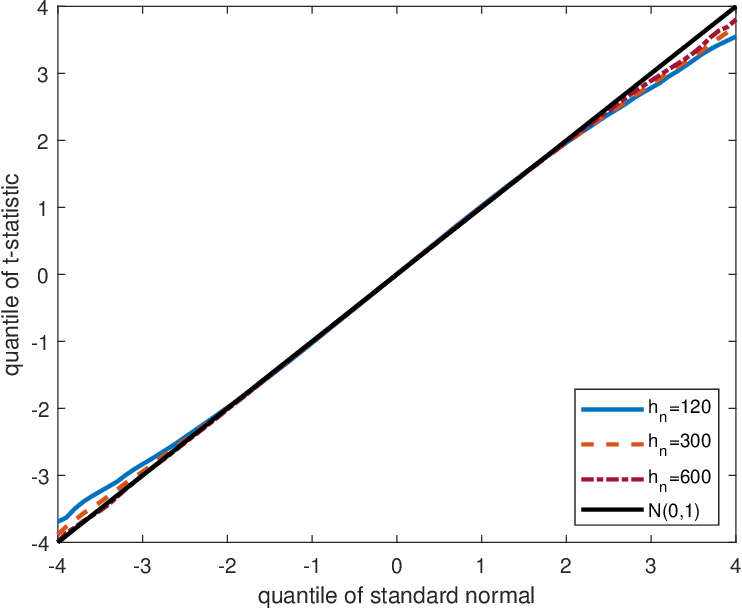} &
\includegraphics[height=0.4\textwidth,width=0.475\textwidth]{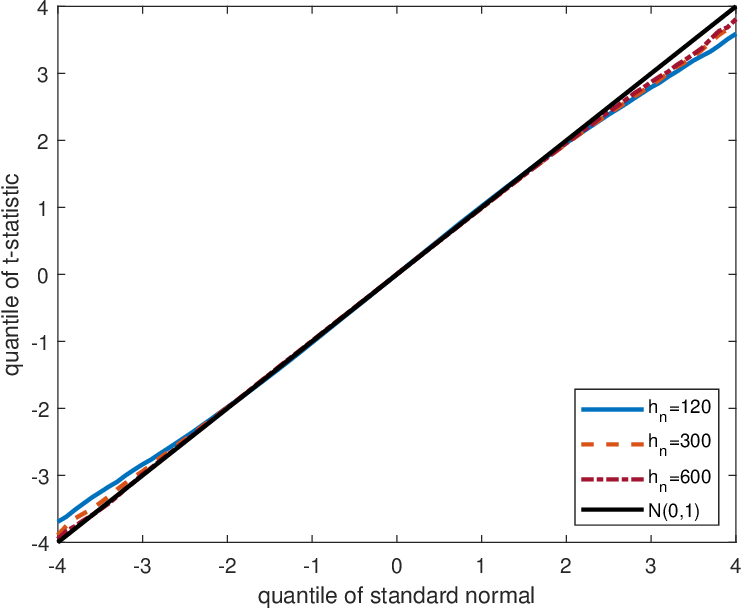} \\
\end{tabular}
\begin{scriptsize}
\parbox{\textwidth}{\emph{Note.} We present a Q-Q plot of $\overbar{T}_{t}^{n}$ under the null hypothesis of no drift burst. Panel A is from the \citet*{heston:93a}-type SV+jump model with no burst in neither drift nor volatility, while Panel B adds a volatility burst using the parametric model in Eq. \eqref{equation:parametric-model} with $b = 0.15$ and $\beta = 0.4$.}
\end{scriptsize}
\end{center}
\end{figure}

Figure \ref{figure:t-statistic} reports Q-Q plots of the distribution of $\overbar{T}_{t}^{n}$ in absence of a drift burst. Panel A is based on the \citet*{heston:93a}-type SV+jump model. As readily seen, the Gaussian curve is an accurate description of the sampling variation of $\overbar{T}_{t}^{n}$, although the $t$-statistic is slightly thin-tailed with $h_{n} = 120$. This is caused by a modest correlation between $\hat{ \overbar{ \mu}}_{t}^{n}$ and $\hat{ \overbar{ \sigma}}_{t}^{n}$, which are computed from partly overlapping data; an effect that is more pronounced for small bandwidths. In Panel B, the process featuring a large volatility burst ($\beta = 0.4$) is plotted.\footnote{The results for other values of $\beta$ fall in between those of Panel A and B and are therefore not reported.} While the volatility burst puts more mass into the tails of the distribution of $\overbar{T}_{t}^{n}$, the normal continues to be a good approximation.

\begin{table}[ht!]
\setlength{ \tabcolsep}{0.15cm}
\begin{center}
\caption{Size and power of drift burst $t$-statistic $T_{m}^{*}$ with $\gamma = 0.5$ and $k_{n} = 3$.
\label{table:sim-drift-burst-gamma=0.5-kn=3}}
\begin{small}
\begin{tabular}{llcrrrrrrrrrrrrrrrrrrr}
\hline \hline
& & \multicolumn{6}{c}{Pr$(T_{m}^{*} > q_{0.950})$} & & \multicolumn{6}{c}{Pr$(T_{m}^{*} > q_{0.990})$} & & \multicolumn{6}{c}{Pr$(T_{m}^{*} > q_{0.995})$} \\
\cline{3-8} \cline{10-15} \cline{17-22}
& & \multicolumn{1}{c}{vb (size)} & & \multicolumn{4}{c}{db (power)} & & \multicolumn{1}{c}{vb (size)} & & \multicolumn{4}{c}{db (power)} & & \multicolumn{1}{c}{vb (size)} & & \multicolumn{4}{c}{db (power)}\\
\cline{3-3} \cline{5-8} \cline{10-10} \cline{12-15} \cline{17-17} \cline{19-22}
& & \multicolumn{1}{c}{$\mu_{t}^{ \text{db}} \equiv 0$} & & $\alpha = $ & 0.55 & 0.65 & 0.75  & & \multicolumn{1}{c}{$\mu_{t}^{ \text{db}} \equiv 0$} & & $\alpha = $ & 0.55 & 0.65 & 0.75  & & \multicolumn{1}{c}{$\mu_{t}^{ \text{db}} \equiv 0$} & & $\alpha = $ & 0.55 & 0.65 & 0.75 \\
\hline
\multicolumn{3}{l}{Panel A: $h_{n} = 120$}\\
$\beta = $& 0.0$^{*}$ & 0.8 &&& 44.9 & 83.0 & 98.0 && 0.1 &&& 29.0 & 66.6 & 90.7 && 0.0 &&& 23.9 & 58.1 & 85.8\\
& 0.1 & 0.9 &&& 27.9 & 75.7 & 96.8 && 0.1 &&& 13.8 & 53.4 & 87.6 && 0.0 &&& 9.9 & 44.4 & 80.6\\
& 0.2 & 1.0 &&& 18.2 & 66.4 & 94.9 && 0.3 &&& 6.4 & 44.1 & 84.4 && 0.0 &&& 4.2 & 35.2 & 75.9\\
& 0.3 & 0.9 &&& 9.5 & 48.9 & 92.1 && 0.3 &&& 2.1 & 26.9 & 76.7 && 0.1 &&& 1.2 & 19.1 & 64.0\\
& 0.4 & 0.9 &&& 5.4 & 25.3 & 80.5 && 0.2 &&& 0.8 & 8.2 & 51.2 && 0.0 &&& 0.4 & 4.6 & 38.6\\
\\
\multicolumn{3}{l}{Panel B: $h_{n} = 300$}\\
$\beta = $& 0.0$^{*}$ & 2.1 &&& 51.6 & 86.7 & 98.9 && 0.6 &&& 39.7 & 78.8 & 96.9 && 0.3 &&& 33.6 & 74.5 & 95.7\\
& 0.1 & 2.0 &&& 41.8 & 82.4 & 98.5 && 0.6 &&& 29.2 & 71.8 & 95.8 && 0.3 &&& 24.8 & 67.1 & 93.7\\
& 0.2 & 2.1 &&& 36.0 & 78.3 & 98.1 && 0.5 &&& 21.9 & 66.3 & 94.1 && 0.3 &&& 18.5 & 60.1 & 92.5\\
& 0.3 & 2.6 &&& 26.1 & 70.3 & 96.7 && 0.7 &&& 13.2 & 54.4 & 92.4 && 0.3 &&& 10.5 & 48.9 & 90.2\\
& 0.4 & 4.1 &&& 16.7 & 52.5 & 93.4 && 0.9 &&& 7.9 & 34.6 & 84.4 && 0.5 &&& 5.5 & 27.9 & 79.2\\
\\
\multicolumn{3}{l}{Panel C: $h_{n} = 600$}\\
$\beta = $& 0.0$^{*}$ & 4.6 &&& 46.4 & 82.6 & 98.0 && 1.0 &&& 33.8 & 73.3 & 95.6 && 0.4 &&& 30.0 & 68.8 & 94.5\\
& 0.1 & 4.2 &&& 40.1 & 79.4 & 97.6 && 1.0 &&& 28.6 & 68.6 & 94.6 && 0.3 &&& 25.2 & 63.7 & 93.1\\
& 0.2 & 4.4 &&& 36.4 & 77.1 & 97.2 && 1.0 &&& 24.6 & 64.7 & 93.8 && 0.2 &&& 21.5 & 60.2 & 92.6\\
& 0.3 & 4.7 &&& 30.3 & 71.1 & 96.0 && 1.1 &&& 19.0 & 57.4 & 92.2 && 0.3 &&& 14.0 & 52.2 & 90.1\\
& 0.4 & 5.8 &&& 23.2 & 57.9 & 93.9 && 1.8 &&& 11.9 & 42.6 & 87.6 && 0.9 &&& 9.2 & 37.3 & 83.6\\
\hline \hline
\end{tabular}
\end{small}\smallskip
\begin{scriptsize}
\parbox{\textwidth}{\emph{Note.} Pr$(T_{m}^{*} > q_{1 - c})$ is the rejection rate (in percent, across Monte Carlo replications) of the drift burst $t$-statistic $T_{m}^{*}$ defined in Eq. \eqref{equation:maxabsTn}, where $q_{1 - c}$ is a simulated $(1 - c)$-level quantile from the finite sample extreme value distribution of $T_{m}^{*}$ under the null of no drift burst, as explained in Appendix \ref{appendix:critical-value}. $\alpha$ is the explosion rate of the drift burst (db), while $\beta$ is the explosion rate of the volatility burst (vb). *$\beta = 0.0$ represents the \cite{heston:93a}-type SV+jump model with no volatility burst. $h_{n}$ is the bandwidth of $\hat{ \mu}_{t}^{n}$ (measured as effective sample size), while the bandwidth of $\hat{ \Sigma}_{t}^{n}$ is 5$h_{n}$. $\gamma$ is the level of noise-to-volatility (per increment) and $k_{n}$ is the pre-averaging horizon.}
\end{scriptsize}
\end{center}
\end{table}

This is corroborated by Table \ref{table:sim-drift-burst-gamma=0.5-kn=3}, where we compute how often $T_{m}^{*}$ leads to rejection of the null hypothesis of no drift burst for three significance levels $c = 5\%, 1\%, 0.5\%$. There are several interesting findings. Look at the columns with $\mu_{t}^{ \text{db}} \equiv 0$, which report the results in absence of a drift burst (i.e., size). Without a volatility burst ($\beta = 0.0$), the test is conservative compared to the nominal level if $h_{n}$ is small, as also reflected in Figure \ref{figure:t-statistic}. As $\beta$ increases, $T_{m}^{*}$ is mildly inflated yielding a tiny size distortion, but the effect is only present for large $\beta$ and $h_{n}$. Otherwise, the test is roughly unbiased. This suggests our $t$-statistic is adaptive and highly robust to even substantial shifts in spot variance, so that we do not falsely pick up an explosion in volatility as a significant drift burst.

Turn next to the alternative with a drift burst (i.e., power). As expected, the power is increasing in $\alpha$, holding $\beta$ fixed, while it is decreasing in $\beta$, holding $\alpha$ fixed. In general, the test has decent power and is capable of identifying a true explosion in the drift coefficient, except those causing a minuscule cumulative log-return and that are coupled with a large volatility burst. While the test has excellent ability to discover the largest drift bursts, which from a practical point of view are arguably also the most important, it is intriguing that we can uncover many of the smaller ones as well. At last, higher values of $h_{n}$ improve the rejection rate under the alternative, but the marginal gain of going from $h_{n} = 300$ to $h_{n} = 600$ is negligible. This suggests -- on the one hand -- that $h_{n}$ should not be too narrow, as it erodes the power, while -- on the other -- it should neither be too wide, as this creates a small size distortion.

\section{Drift bursts in financial markets} \label{section:empirical}

\subsection{Data}

\begin{table}[t!]
\setlength{\tabcolsep}{0.20cm}
\begin{center}
\caption{CME futures data summary statistics. \label{Table:summaryStatsCME}}
\smallskip
\begin{tabular}{llccccccc}
\hline
					& & 		&\multicolumn{2}{c}{volume} & \# quote & inside &\multicolumn{2}{c}{sub-sample retained} \\
\cline{4-5}\cline{8-9}
code & name		 	& \# days & \# contracts & notional (\$bn) & updates & spread (bps) & by volume & by quotes\\
\hline
6E & Euro FX        & 1,536 &   209,281 &  31.9 & 62,357 & 0.69 & 92.87\% & 87.25\%\\
CL & Crude oil      & 1,545 &   299,821 &  18.8 & 69,580 & 1.70 & 95.58\% & 92.28\%\\
ES & E-mini S\&P500 & 2,053 & 1,763,100 & 145.5 & 27,513 & 1.54 & 97.33\% & 91.94\%\\
GC & Gold           & 1,544 &   162,056 &  22.0 & 56,562 & 0.97 & 88.14\% & 84.59\%\\
ZC & Corn           & 1,528 &   112,612 &   2.6 &  4,870 & 5.89 & 87.23\% & 71.18\%\\
ZN & 10-Year T-Note & 1,542 & 1,121,122 & 112.1 &  6,372 & 1.22 & 94.33\% & 89.10\%\\
\hline
\end{tabular}
\smallskip
\begin{scriptsize}
\parbox{\textwidth}{\emph{Note}. This table reports for each futures contract, the number of days in the sample, the average daily volume by number of contracts and notional traded, the average daily number of top-of-book quote updates, and the average daily median spread in basis points calculated from 09:00 -- 10:00 Chicago time. The sample period is January 2012 -- December 2017 for all contracts, except ES where we start in January 2010. In the empirical analysis, we restrict attention to the most active trading hours from 01:00 -- 15:15 Chicago time for all contracts, except for ZC where the interval is restricted to 08:30 -- 13:20 Chicago time. The fraction of volume and quote updates retained after removing the most illiquid parts of the day is reported in the last two columns.}
\end{scriptsize}
\end{center}
\end{table}

We apply the drift burst test statistic developed above to a comprehensive set of intraday tick data, covering a broad range of financial assets. We use trades and quotes with milli-second precision timestamps for futures contracts traded on the Chicago Mercantile Exchange (CME). We select the most actively traded futures contract for each of the main assets classes, namely the Euro FX for currencies (6E), Crude oil for energy (CL), the E-mini S\&P500 for equities (ES), Gold for precious metals (GC), Corn for agricultural commodities (ZC), and the 10-Year Treasury Note for rates (ZN). These futures contracts are amongst the most liquid financial instruments in the world. To illustrate, the average daily notional volume traded in just a single ES contract on the CME is comparable to the trading volume of the entire U.S. cash equity market covering over 5,000 stocks traded across more than ten different exchanges.\footnote{See \url{https://batstrading.com/market_summary/} for daily U.S. equity market volume statistics.} The sample period is January 2012 -- December 2017,  except we backdate ES to January 2010 in order to capture the May 2010 flash crash. While the CME is open nearly all day, we restrict attention to the more liquid European and U.S. trading sessions: from 01:00 -- 15:15 Chicago time or 07:00 -- 21:15 London time. The only exception is Corn, where we use data from 08:30 -- 13:20 Chicago time. Outside of these hours, trading is minimal in this contract. Table \ref{Table:summaryStatsCME} provides informative summary statistics of the data.

\subsection{Implementation of test}

We construct for each series a mid-quote as the average of the best bid and offer available at any point in time. A quote update is retained if the mid-quote changes. The remaining data are pre-averaged with $k_{n} = 3$.\footnote{The Online Appendix contains the empirical analysis with $k_{n} = 1$ (i.e., no pre-averaging), 5 and 10. The results are broadly speaking in line with those reported here.} We then calculate the drift burst test statistic on a regular five-second grid and only include values that are preceded by a mid-quote revision. As our primary interest is to identify short-lived drift bursts, we continue with a five-minute bandwidth for the drift. We base the spot volatility on a 25-minute bandwidth with the Parzen kernel and $L_{n} = 2(k_{n}-1) + 10$ lags for the HAC-correction. As in the simulations, a left-sided exponential kernel is adopted: $K(x) = \exp(-|x|)$, for $x \leq 0$.\footnote{In practice, a backward-looking kernel is preferred to allow for real-time updating. Moreover, it enchances the power for testing the drift burst hypothesis. The intuition is that after a drift burst, the pre-dominant tendency of the price to reverse combined with a high and persistent level of volatility, delivers a lower value of the test statistic if a two-sided kernel is employed.}

\begin{table}[t!]
\setlength{\tabcolsep}{0.375cm}
\begin{center}
\caption{Drift burst test summary statistics. \label{Table:dbTest}}
\smallskip
\begin{tabular}{llcccrrrrrr}
\hline
	&& \multicolumn{3}{c}{empirical distribution} & & \multicolumn{5}{c}{\# of identified drift bursts} \\
\cline{3-5}\cline{7-11}
code & name & $\sigma$  & $\sigma_{q}$ & kurtosis & & $|T|>4.0$ & >4.5 & >5.0 & >5.5 & >6.0\\
\hline
6E & Euro FX         &   1.04 &   1.05 &    3.6 && 1329 & 615 & 292 & 145 &  64\\
CL & Crude Oil       &   1.04 &   1.07 &    3.8 && 1606 & 667 & 300 & 127 &  51\\
ES & E-mini S\&P500  &   1.07 &   1.07 &    3.1 &&  549 & 202 &  78 &  29 &  11\\
GC & Gold            &   1.02 &   1.04 &    3.7 && 1324 & 542 & 227 &  83 &  35\\
ZC & Corn            &   1.15 &   1.13 &    2.9 &&   97 &  29 &  15 &   5 &   2\\
ZN & 10-Year T-Note  &   1.22 &   1.10 &    2.4 &&  105 &  46 &  21 &  10 &   4\\
\hline
\end{tabular}
\smallskip
\begin{scriptsize}
\parbox{\textwidth}{\emph{Note}. This table reports for each futures contract, the standard deviation and kurtosis of the empirical drift burst $t$-statistic. We calculate the test every five seconds across the sample if there was a mid-quote update over that interval. The standard deviation is also calculated by rescaling the 5/95-percentile of the empirical distribution by that of a standard normal (``$\sigma_q$''). The number of drift bursts identified for critical values ranging between 4 and 6 is reported. The number of false positives we expect, which can be computed using the technique described in Appendix \ref{appendix:critical-value}, is virtually zero.}
\end{scriptsize}
\end{center}
\end{table}

Table \ref{Table:dbTest} reports selected descriptive measures of the calculated drift burst $t$-statistics over the full sample. Judging by the standard deviation and kurtosis of the test statistic, the distribution is close to standard normal, as consistent with the asymptotic theory under the null of no drift burst. This is remarkable, because the test is applied to relatively short intraday intervals across a wide range of asset classes and is therefore exposed to substantial changes in liquidity conditions, diurnal effects, or to futures contracts where the minimum price increment -- and hence the microstructure noise -- is relatively large (e.g. ZN). Concentrating on the tails, we identify a large number of drift bursts.\footnote{To account for the rolling calculation of the test statistic and avoid double counting of events, we allow for at most one drift burst to be established over any five-minute window at which the test statistic attains a local extremum and exceeds a set critical value.}  At a critical value of 4.5, for instance, there are 202 drift bursts in the E-mini S\&P500 futures, or about one every two weeks. They are more prevalent in Euro FX, Gold, and Oil contracts but less frequent in the Treasury and Corn futures. The number of expected false positives (computed as described in Appendix \ref{appendix:critical-value}) is practically zero.

\begin{figure}[t!]
\setlength{\tabcolsep}{0.1cm}
\caption{Time series of drift bursts.\label{figure:timeseries}}
\begin{center}
\vskip-0.4cm
\begin{tabular}{cc}
Panel A: Distribution over time. & Panel B: Monthly count. \\
\includegraphics[height=0.4\textwidth,width=0.475\textwidth]{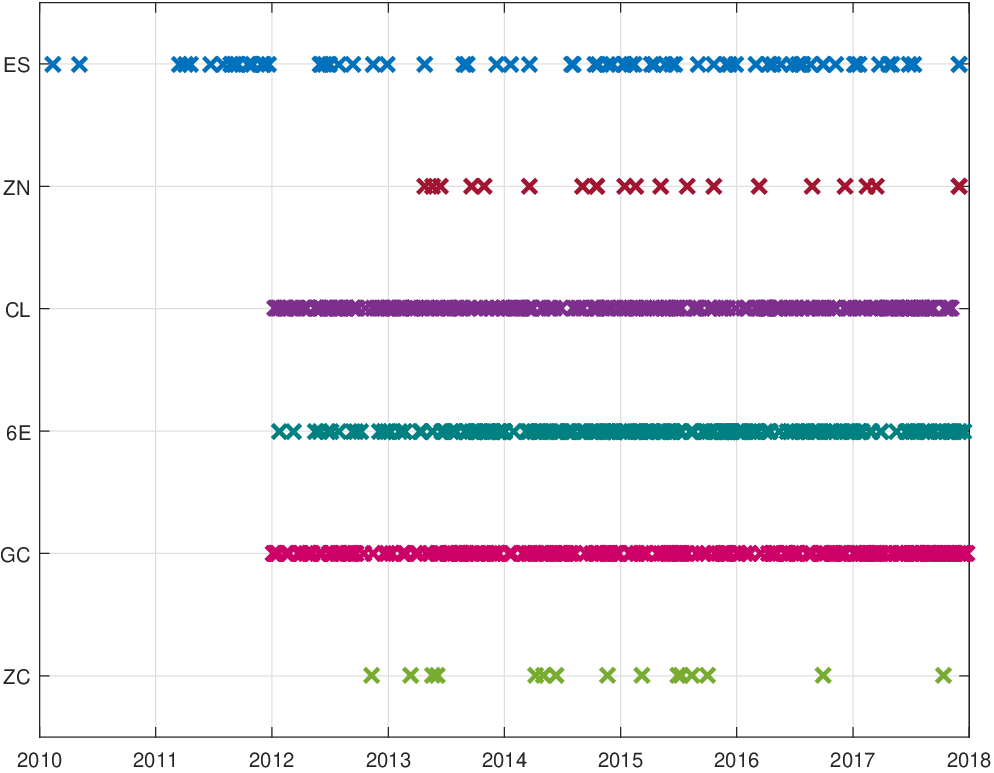} &
\includegraphics[height=0.4\textwidth,width=0.475\textwidth]{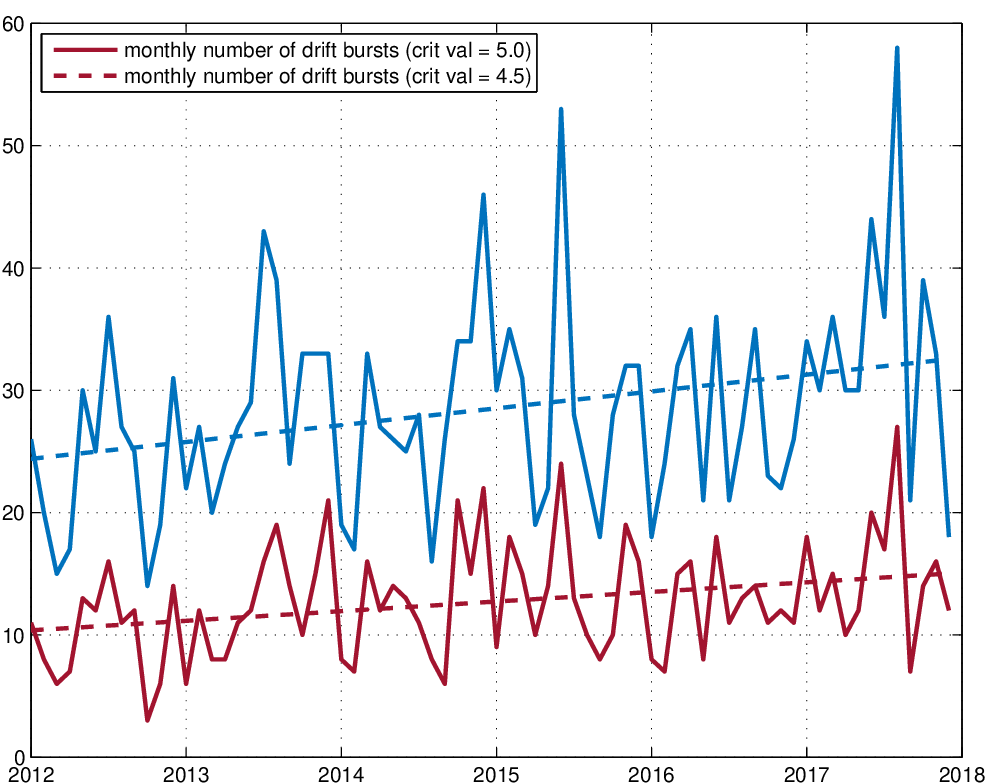} \\
\end{tabular}
\begin{scriptsize}
\parbox{\textwidth}{\emph{Note.} In Panel A, we plot the cross-sectional distribution of drifts bursts across asset classes and over time (each cross represents a day with $|T_{m}^{*}| > 5.0$), while Panel B shows the associated number of significant daily events aggregated to a monthly level.}
\end{scriptsize}
\end{center}
\end{figure}

Panel A of Figure \ref{figure:timeseries} indicates the location of drift bursts for the various securities, while Panel B reports the pooled monthly counts. The results lend support to the perception that flash crashes are increasing over time. A time trend is positive for each asset but lacks statistical significance for some due to the limited number of observations. The pooled results, however, indicate a statistically significant time trend of 5--10\% per annum in the number of identified drift bursts with a regression $t$-statistic of 2.5 for both scenarios drawn in Panel B.

\begin{figure}[t!]
\setlength{\tabcolsep}{0.1cm}
\caption{Drift burst examples. \label{figure:example}}
\begin{center}
\vskip-0.15cm
\begin{tabular}{cc}
Panel A: Euro FX (May 2, 2016). & Panel B: Gold (Mar 26, 2015). \\
\includegraphics[height=0.4\textwidth,width=0.475\textwidth]{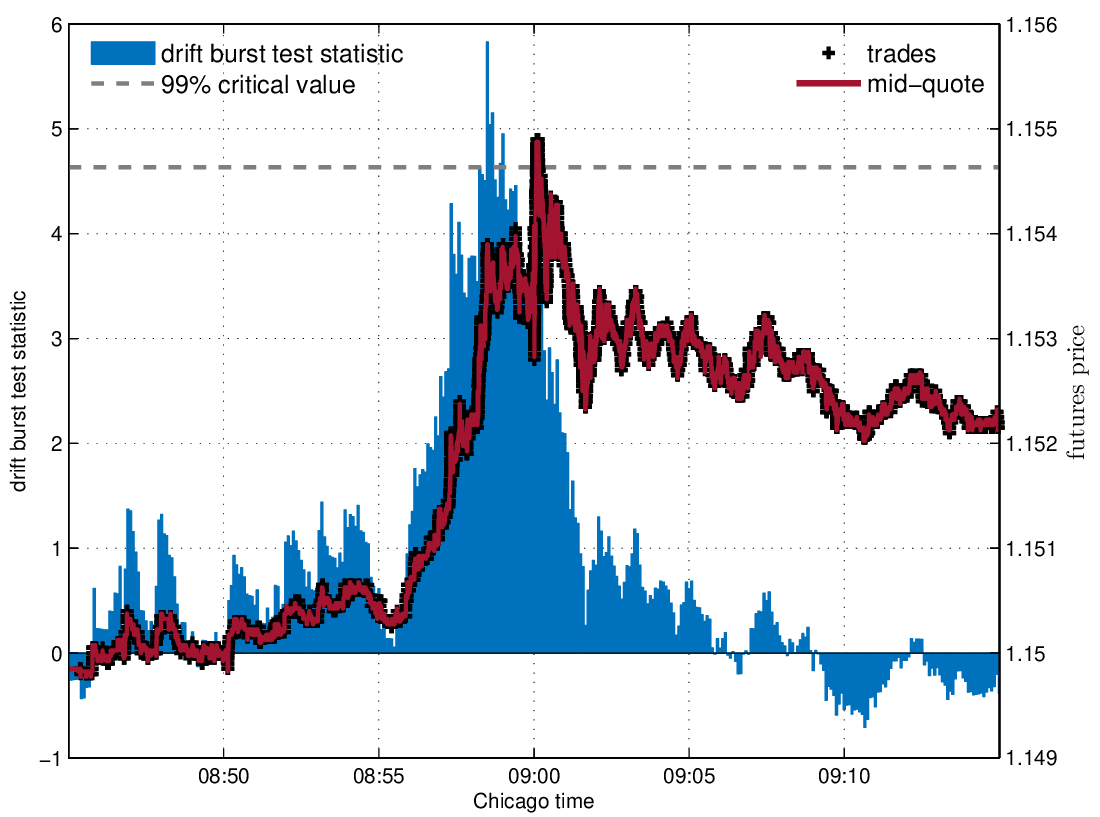} &
\includegraphics[height=0.4\textwidth,width=0.475\textwidth]{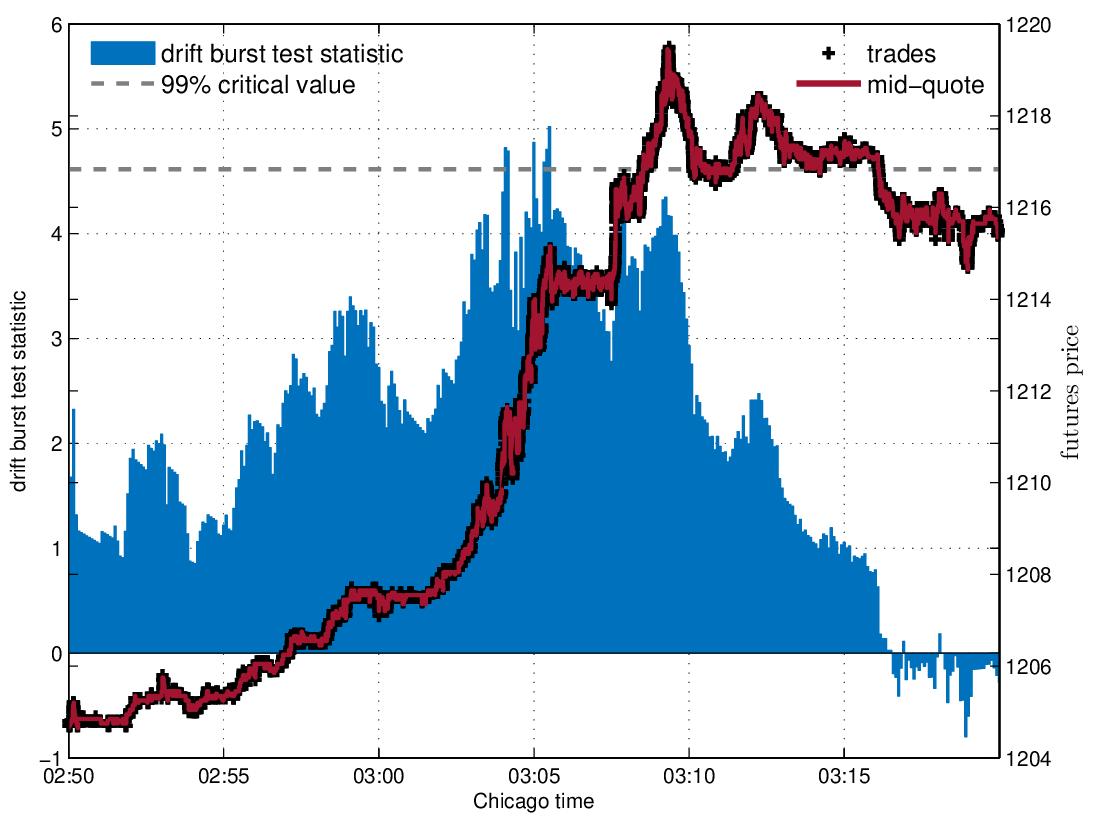}\\
Panel C: Crude oil (Sep 17, 2012). & Panel D: Corn (Mar 9, 2012). \\
\includegraphics[height=0.4\textwidth,width=0.475\textwidth]{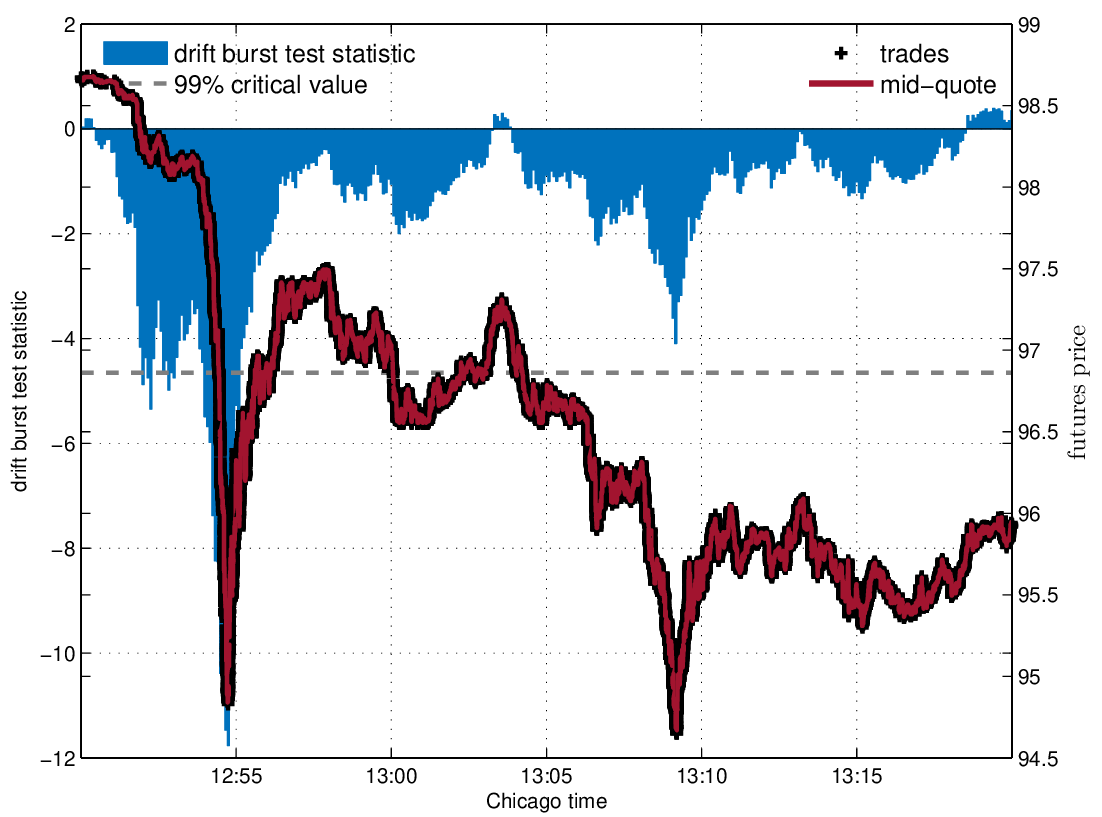} &
\includegraphics[height=0.4\textwidth,width=0.475\textwidth]{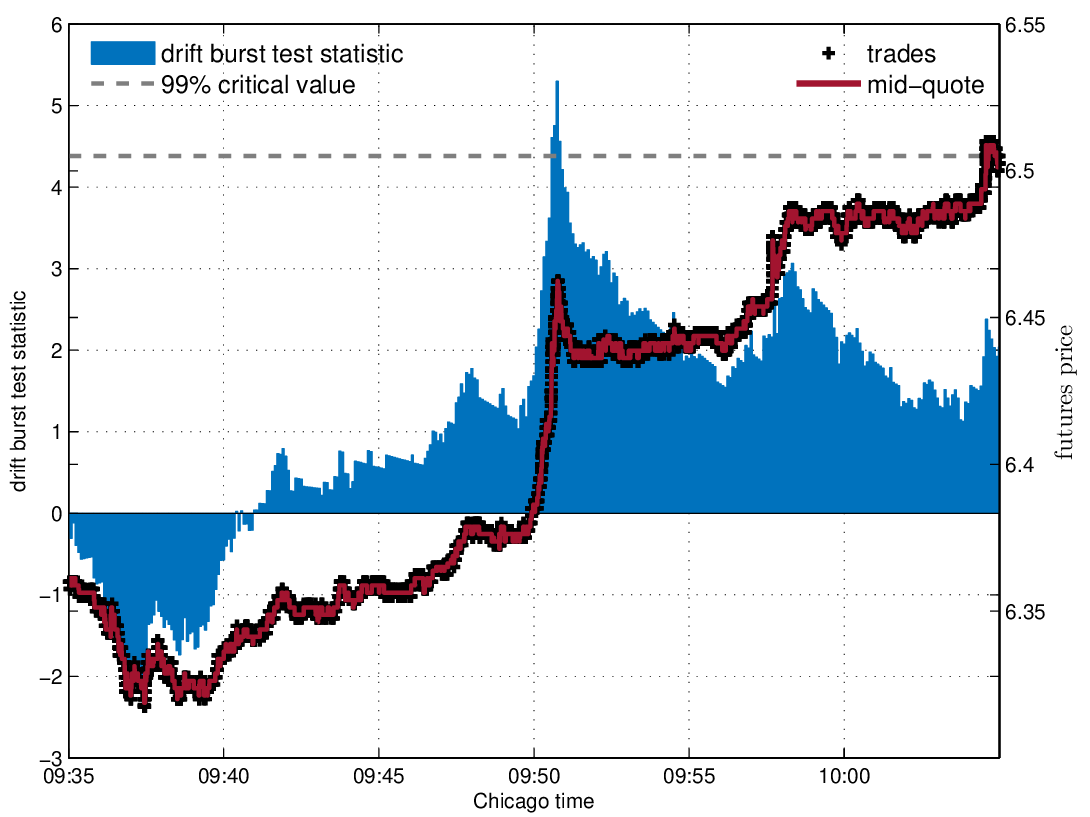}\\
\end{tabular}
\begin{scriptsize}
\parbox{0.95\textwidth}{\emph{Note}. This figure draws, for some identified drift bursts, the sample path of the mid-quote and traded price (right axis) together with the $t$-statistic (left axis) over a 30-minute window that includes the peak of the drift burst.}
\end{scriptsize}
\end{center}
\end{figure}

Figure \ref{figure:example} shows some examples of single-asset drift bursts identified by the test. A multi-asset drift burst is presented in Figure \ref{figure:hoax}, which plots the evolution of the six asset prices during the Twitter hoax flash crash of April 23, 2013. The figures show that drift bursts are a stylized feature of the price process and, in some instances, systemic to the market. It is evident that neither jumps nor volatility are driving the price dynamics observed in these examples. The drift burst hypothesis is a plausible alternative to model the data.

\begin{figure}[t!]
\setlength{\tabcolsep}{0.1cm}
\caption{The Twitter hoax of April 23, 2013: A synchronous flash event. \label{figure:hoax}}
\begin{center}
\begin{tabular}{cc}
Panel A: E-mini S\&P500. & Panel B: 10-Year T-Note. \\
\includegraphics[height=0.29\textwidth,width=0.475\textwidth]{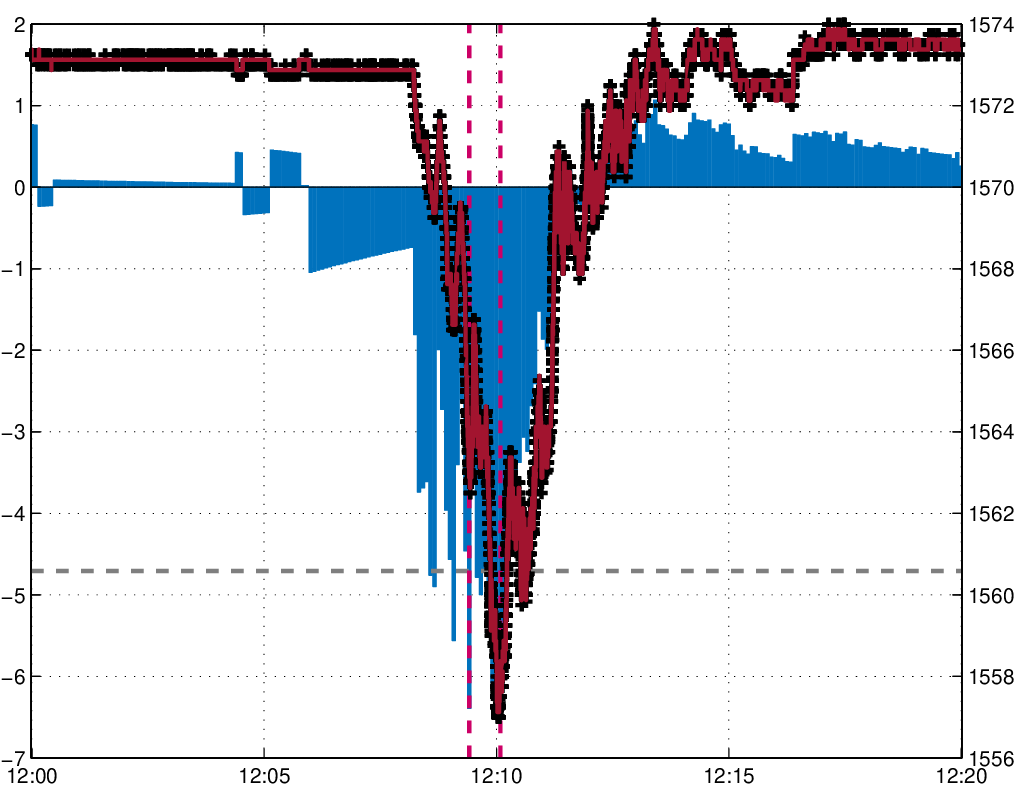} &
\includegraphics[height=0.29\textwidth,width=0.475\textwidth]{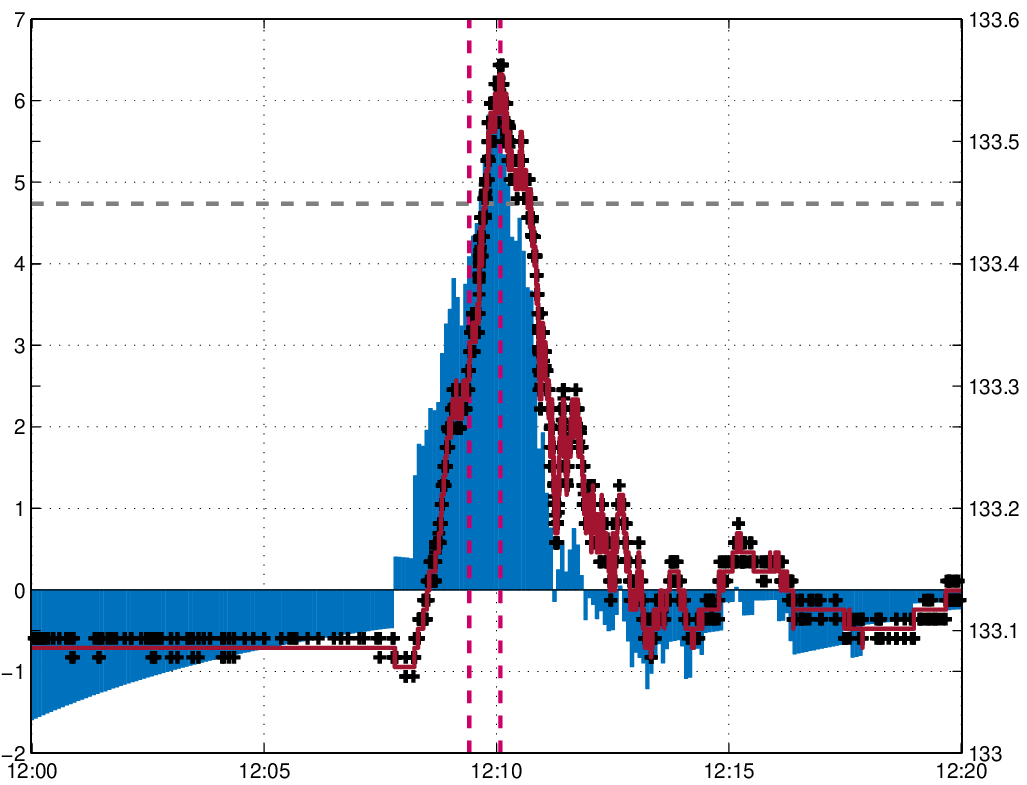} \\
 Panel C: Euro FX. & Panel D: Gold. \\
\includegraphics[height=0.29\textwidth,width=0.475\textwidth]{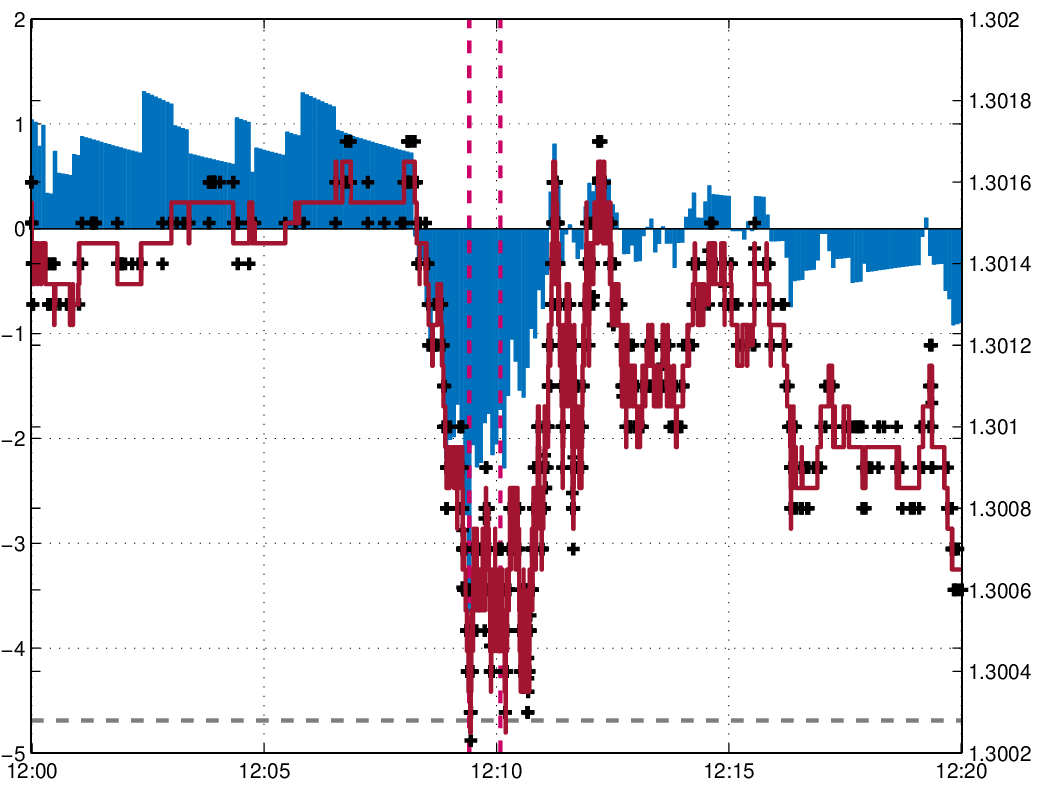} &
\includegraphics[height=0.29\textwidth,width=0.475\textwidth]{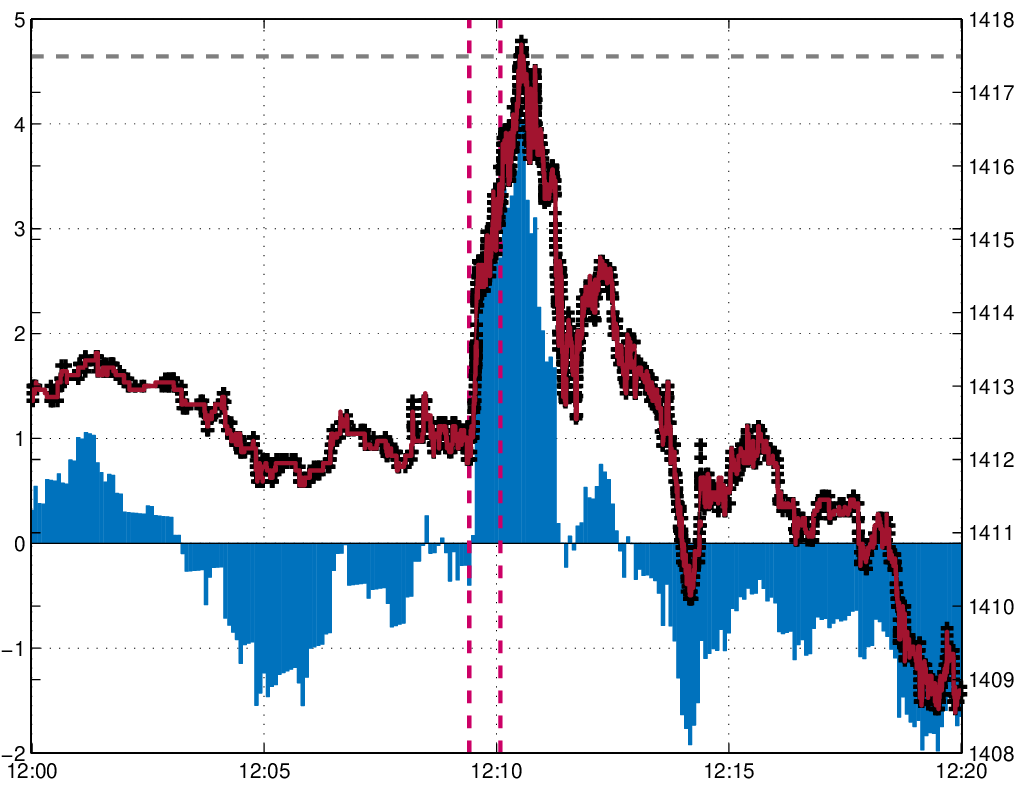} \\
 Panel E: Crude oil. & Panel F: Corn. \\
\includegraphics[height=0.29\textwidth,width=0.475\textwidth]{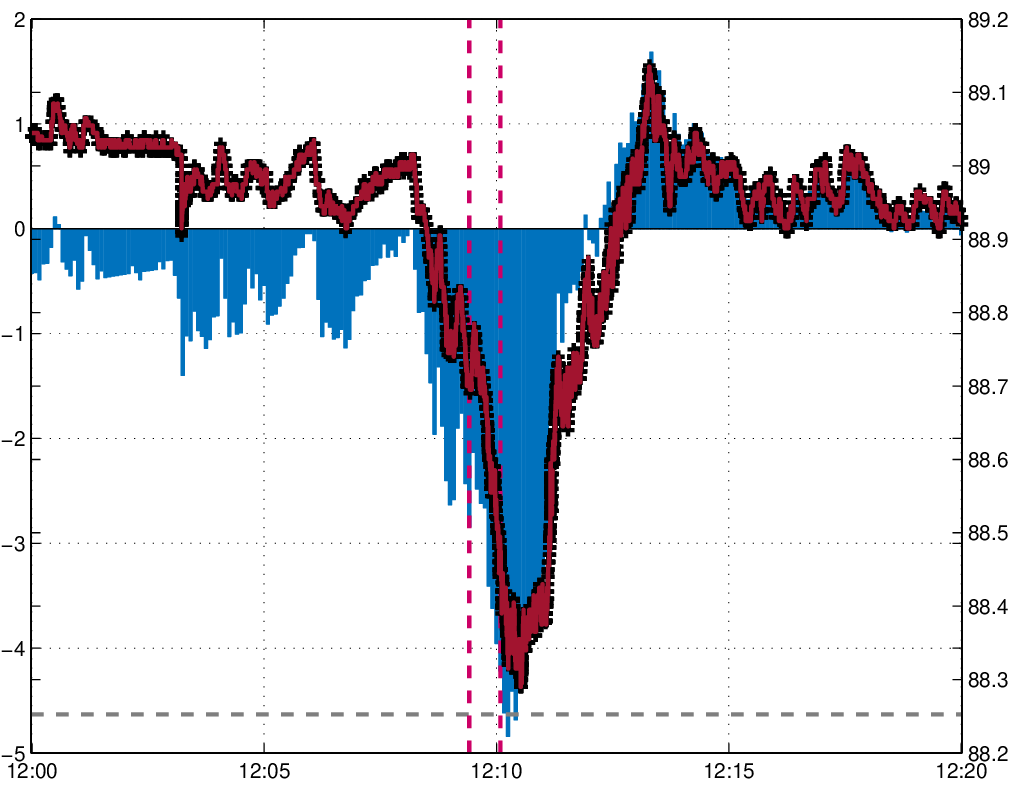} &
\includegraphics[height=0.29\textwidth,width=0.475\textwidth]{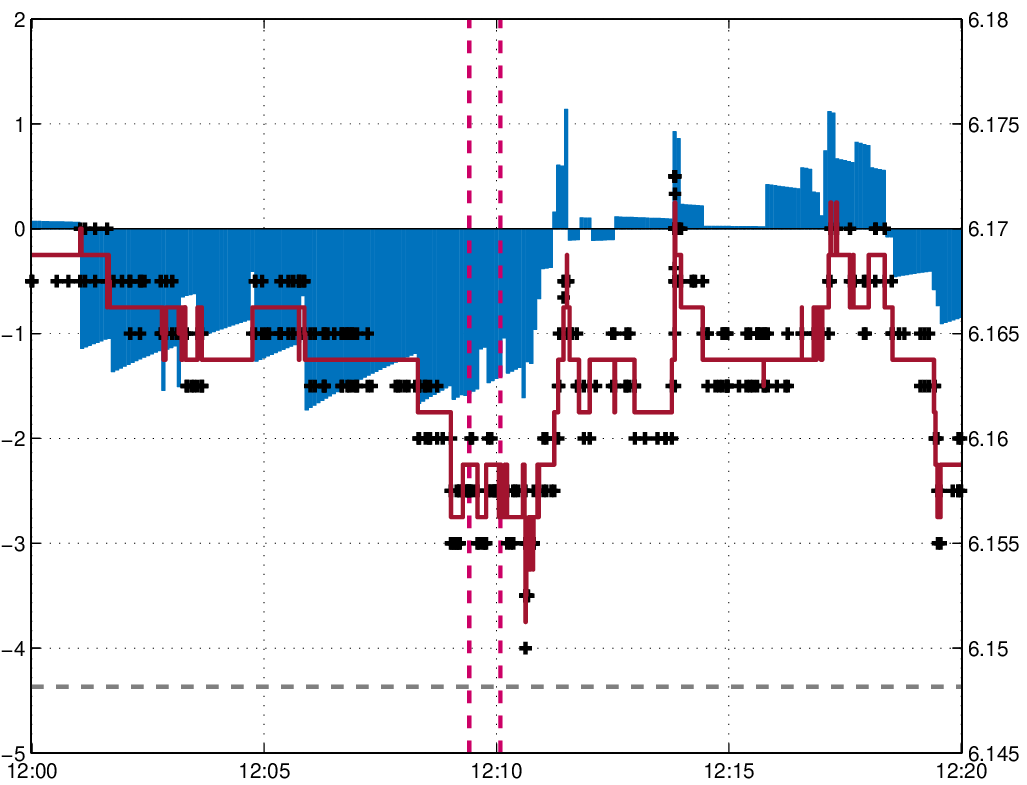} \\
\end{tabular}
\begin{scriptsize}
\parbox{0.95\textwidth}{\emph{Note}. This figure draws the mid-quote and traded price (right axis) together with the $t$-statistic (left axis) for our six asset classes over a 20-minute window around the Twitter flash crash of April 23, 2013. At 12:07 p.m. Chicago time, Associated Press' Twitter account was hacked and a fake news was released, suggesting there had been explosions in the White House and that President Obama was injured. This led to perilous but short-lived sell-offs in the S\&P500, Euro FX and Crude oil futures contracts, while the 10-Year T-Note and Gold contracts rallied. Only Corn was largely unaffected.}
\end{scriptsize}
\end{center}
\end{figure}

\subsection{Reversion after the drift burst} \label{section:reversion}

We now investigate in more detail whether the mean reversion experienced during the flash crashes in Figure \ref{figure:flash_crash} and Figure \ref{figure:hoax} is more broadly associated with the price dynamics of a drift burst. Let $\{ t_{j} \}$ denote the set of time points, where drift bursts are identified. The start of a drift burst is set to five minutes before $t_{j}$.\footnote{The 5-minute frequency is arbitrary, but often used in practice and consistent with our bandwidth. As a robustness check, we also applied an endogenous event window, where the duration of the drift burst is defined relative to the latest point in time prior to $t_{j}$, where the absolute value of the $t$-statistic is below one. The results are in line with those we report here and are available at request.} We sample the mid-quote process at this frequency pre- and post-drift burst and calculate:
\begin{equation} \label{equation:dr}
R_{t_{j}}^{-} = X_{t_{j}} - X_{t_{j} - 5\text{m}} \qquad \text{and} \qquad R_{t_{j}}^{+} = X_{t_{j} + 5\text{m}} - X_{t_{j}},
\end{equation}
where $R_{t_{j}}^{-}$ is the five-minute log-return during the $j$th drift burst, while $R_{t_{j}}^{+}$ is the corresponding post-drift burst log-return. In Figure \ref{figure:reversion}, we plot $R_{t_{j}}^{+}$ against $R_{t_{j}}^{-}$ pooled across the various asset markets. We observe that drift bursts can be associated with both positive and negative returns, but most of them are reversals and the percentage of ``gradual jumps'' is roughly one third.

To gauge the magnitude of the reversion and evaluate whether drift bursts are subject to short-term return predictability, we run the following regression:
\begin{equation}\label{eqn:revreg}
R_{t_{j}}^{+} = a + b R_{t_{j}}^{-} + \epsilon_{t_{j}}.
\end{equation}
A value of $b$ different from zero indicates predictability conditional on a drift burst, and $b<0$ means a drift burst tends to be followed by a retracement of the price.

\begin{figure}[t!]
\centering
\caption{Reversion in the drift burst. \label{figure:reversion}}
\includegraphics[width=0.475\textwidth]{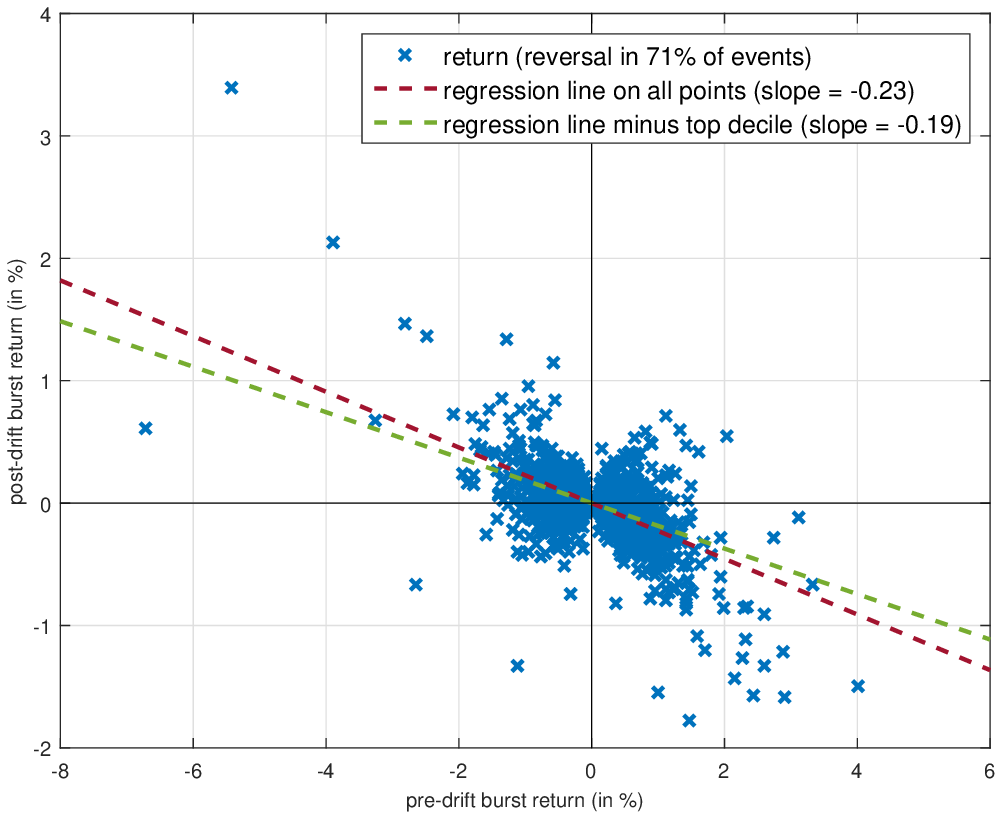}
\begin{scriptsize}
\parbox{0.75\textwidth}{\emph{Note}. This figure draws on the horizontal axis the percentage price move from the start to the peak of a drift burst against the subsequent price move (as defined by $R_{t_{j}}^{+}$ and $R_{t_{j}}^{-}$ in Eq. \eqref{equation:dr}) on the vertical axis. The horizon is five minutes. The number of drift bursts in each quadrant are Q1: 324, Q2: 670, Q3: 279 and Q4: 828.}
\end{scriptsize}
\end{figure}

In Table \ref{Table:reversion}, we show the regression results for each futures contract separately and with all assets pooled together. We find a strong mean reversion with estimates of $b$ that are negative and highly significant (the intercept estimate is insignificant and omitted). This is consistent across asset classes, also for those with a relatively small number of observations. The regression $R^{2}$ tends to lie in the 25\% - 40\% range, with a few exceptions, and is about 30\% on average. This suggests a substantial predictive power of the pre-drift burst return to forecast the post-drift burst return. The fraction of reversals, as measured by counting the relative occurrence where $R_{t_{j}}^{+}$ and $R_{t_{j}}^{-}$ are of opposite sign, rarely drops below 65\%. As a robustness check, we remove the top decile of the most significant drift bursts and rerun the regression. The outcome is in the right-hand side of the table. As expected, the regression $R^{2}$ drops but the finding remains overly evident in the data with frequent reversals and significant mean reversion. The critical value is set to 4.0 and 4.5 here, but the results do not change much for larger and more conservative values.

\begin{table}[t!]
\setlength{\tabcolsep}{0.45cm}
\begin{center}
\caption{Reversion in the drift burst. \label{Table:reversion}}
\smallskip
\begin{tabular}{lcrccclrccc}
\hline
&& \multicolumn{4}{c}{all} && \multicolumn{4}{c}{without top decile} \\
\cline{3-6} \cline{8-11}
code && \# & $b$ & $R^{2}$ & $\%R$ && \# & $b$ & $R^{2}$ & $\%R$ \\
\hline
\multicolumn{11}{l}{\emph{Panel A: Critical value = 4.0}}\\
ES  &&  549 & $\underset{(-12.05)}{ -0.26}$& $ 21.0\%$& $ 67.2\%$ &&  494 & $\underset{( -5.59)}{ -0.14}$& $  6.0\%$& $ 66.8\%$\\
ZN  &&  105 & $\underset{( -5.30)}{ -0.25}$& $ 21.3\%$& $ 65.7\%$ &&   95 & $\underset{( -3.09)}{ -0.16}$& $  9.2\%$& $ 64.2\%$\\
6E  && 1329 & $\underset{(-21.05)}{ -0.22}$& $ 25.0\%$& $ 68.2\%$ && 1196 & $\underset{(-19.64)}{ -0.21}$& $ 24.4\%$& $ 68.1\%$\\
GC  && 1324 & $\underset{(-20.65)}{ -0.19}$& $ 24.4\%$& $ 78.1\%$ && 1192 & $\underset{(-18.36)}{ -0.18}$& $ 22.1\%$& $ 78.2\%$\\
CL  && 1606 & $\underset{(-29.33)}{ -0.21}$& $ 34.9\%$& $ 71.3\%$ && 1445 & $\underset{(-24.30)}{ -0.18}$& $ 29.0\%$& $ 70.6\%$\\
ZC  &&   97 & $\underset{( -8.25)}{ -0.20}$& $ 41.5\%$& $ 76.3\%$ &&   87 & $\underset{( -8.05)}{ -0.21}$& $ 43.0\%$& $ 77.0\%$\\
ALL && 5010 & $\underset{(-46.62)}{ -0.21}$& $ 30.3\%$& $ 71.8\%$ && 4509 & $\underset{(-38.40)}{ -0.18}$& $ 24.6\%$& $ 71.6\%$\\
\\
\multicolumn{11}{l}{\emph{Panel B: Critical value = 4.5}}\\
ES  &&  202 & $\underset{(-14.45)}{ -0.36}$& $ 51.0\%$& $ 70.8\%$ &&  182 & $\underset{( -8.12)}{ -0.20}$& $ 26.7\%$& $ 69.8\%$\\
ZN  &&   46 & $\underset{( -3.89)}{ -0.29}$& $ 25.1\%$& $ 71.7\%$ &&   41 & $\underset{( -1.76)}{ -0.13}$& $  7.2\%$& $ 70.7\%$\\
6E  &&  615 & $\underset{(-15.65)}{ -0.23}$& $ 28.5\%$& $ 66.3\%$ &&  554 & $\underset{(-15.31)}{ -0.23}$& $ 29.8\%$& $ 65.0\%$\\
GC  &&  542 & $\underset{(-11.86)}{ -0.17}$& $ 20.6\%$& $ 77.9\%$ &&  488 & $\underset{( -9.57)}{ -0.15}$& $ 15.8\%$& $ 77.3\%$\\
CL  &&  667 & $\underset{(-20.51)}{ -0.23}$& $ 38.7\%$& $ 70.0\%$ &&  600 & $\underset{(-16.02)}{ -0.20}$& $ 30.0\%$& $ 68.7\%$\\
ZC  &&   29 & $\underset{( -4.61)}{ -0.16}$& $ 43.1\%$& $ 86.2\%$ &&   26 & $\underset{( -3.98)}{ -0.16}$& $ 38.8\%$& $ 84.6\%$\\
ALL && 2101 & $\underset{(-34.40)}{ -0.23}$& $ 36.0\%$& $ 71.3\%$ && 1891 & $\underset{(-26.91)}{ -0.19}$& $ 27.7\%$& $ 70.2\%$\\
\hline
\end{tabular}
\smallskip
\begin{scriptsize}
\parbox{0.975\textwidth}{\emph{Note}. This table reports for each security, the number of identified drift bursts $(\#)$ using a critical value of 4.0 and 4.5, the estimated slope coefficient $b$ of Eq. \eqref{eqn:revreg} with its associated $t$-statistic in parenthesis, the regression $R^{2}$, and the probability of reversion ($\% R$) calculated as the fraction of drift bursts, where the price direction after the $t$-statistic peaks is opposite in sign compared to that in the run-up. To confirm robustness of the results, the right-hand part of the table removes the decile of the strongest drift bursts as measured by the absolute value of the $t$-statistic.}
\end{scriptsize}
\end{center}
\end{table}

\subsection{Asymmetric reversal, trading volume and liquidity}

\citet*{grossman-miller:88a} suggest a large drift-to-volatility can be caused by exogenous demand for immediacy. In this framework, the return reversal represents a premium paid to market makers supplying liquidity against one-sided order flow \citep*[e.g.][also shows that returns on short-term reversal strategies can be thought of as a liquidity signature]{nagel:12a}. The empirical results in Section \ref{section:reversion} align with this interpretation. Their model, however, is symmetric in the order imbalance, as expressed by Eq. \eqref{eqn:GM88}.

\citet{huang-wang:09a} show how trading imbalances arise endogenously due to costly market presence. Such a mechanism always leads to selling pressure, which tends to attenuate rallies and exacerbate sell-offs, resulting in market crashes absent news on fundamentals (Result 1 -- 3 in their paper). As in \citet{grossman-miller:88a}, the model induces negative serial correlation in observed returns (Result 4a), but it is asymmetric. Negative returns exhibit stronger correlation than positive returns (Result 4b); volatility is higher during a negative return and high volume event (Result 5); abnormal volume implies larger future returns and stronger reversion (Result 6); and negative returns accompanied by high volume exhibit stronger reversals (Result 7). We here conduct an empirical assessment to check whether our sample of pre- and post-drift burst returns are consistent with these testable implications. We restrict attention to the E-mini S\&P500 futures contract ES \citep*[][also base their analysis on a stock paying a risky dividend]{huang-wang:09a}. The results for other assets are available in the Online Appendix.

We construct a $5$-minute grid as above. We then compute for each interval $j$:
\begin{itemize}
\item $R_{t_{j}}^{-}$ and $V_{t_{j}}^{-}$: the 5-minute log-return and traded volume,\footnote{We construct a normalized gross trading volume measure, which is defined by taking the gross trading volume in notional value and normalizing it by the average trading rate at that time of the day. This allows to depurate our volume series for the pronounced intraday swings in trading intensity, which makes the analysis more robust to time-of-the-day effects. Still, the results based on raw gross trading volume are broadly consistent with what we report here.}
\item $R_{t_{j}}^{+}$: the subsequent 5-minute log-return,
\item $\rho$: correlation between $R_{t_{j}}^{-}$ and $R_{t_{j}}^{+}$,
\item $\sigma_{p} = \sqrt{(R_{t_{j}}^{-})^2 + (R_{t_{j}}^{+})^2}$: a volatility measure.
\end{itemize}

\begin{table}[t!]
\setlength{\tabcolsep}{0.6cm}
\begin{center}
\caption{ES return and volume dynamics. \label{table:huang-wang-ES-1}}
\smallskip
\begin{tabular}{lcccccc}
\hline
Conditional information & & $E[V_{t_{j}}^{-}]$ & $E[R_{t_{j}}^{-}]$ & $E[R_{t_{j}}^{+}]$ & $\rho$ &  $\sigma_{p}$ \\
\cline{1-1} \cline{3-7}
\multicolumn{3}{l}{\emph{Panel A: Unconditional}}\\
all                                & &   1.00 &   0.03 &   0.02 &  -1.81\% &   9.41\\
\\[-0.25cm]
$R_{t_{j}}^{-} > E[R_{t_{j}}^{-}]$ & &   1.06 &   4.93 &  -0.08 &   0.05\% &   9.99\\
$R_{t_{j}}^{-} < E[R_{t_{j}}^{-}]$ & &   0.96 &  -3.51 &   0.10 &  -2.27\% &   8.97\\
\\[-0.25cm]
$V_{t_{j}}^{-} > E[V_{t_{j}}^{-}]$ & &   1.54 &  -0.01 &   0.01 &  -1.93\% &  12.15\\
$V_{t_{j}}^{-} < E[V_{t_{j}}^{-}]$ & &   0.46 &   0.06 &   0.04 &  -1.34\% &   5.42\\
\\[-0.25cm]
\multicolumn{3}{l}{\emph{Panel B: Negative drift burst}} \\
all                                & &   3.53 & -37.39 &   8.01 & -79.58\% &  40.15\\
\\[-0.25cm]
$R_{t_{j}}^{-} > E[R_{t_{j}}^{-}]$ & &   2.70 & -23.23 &   4.17 & -15.90\% &  24.78\\
$R_{t_{j}}^{-} < E[R_{t_{j}}^{-}]$ & &   5.46 & -70.34 &  16.94 & -85.53\% &  75.95\\
\\[-0.25cm]
$V_{t_{j}}^{-} > E[V_{t_{j}}^{-}]$ & &   6.22 & -55.25 &  13.29 & -84.26\% &  59.57\\
$V_{t_{j}}^{-} < E[V_{t_{j}}^{-}]$ & &   1.94 & -26.83 &   4.89 & -53.52\% &  28.67\\
\\[-0.25cm]
\multicolumn{3}{l}{\emph{Panel C: Positive drift burst}} \\
all                                & &   4.22 &  36.41 &  -5.44 & -25.73\% &  38.29\\
\\[-0.25cm]
$R_{t_{j}}^{-} > E[R_{t_{j}}^{-}]$ & &   6.07 &  56.42 &  -9.02 & -15.03\% &  59.58\\
$R_{t_{j}}^{-} < E[R_{t_{j}}^{-}]$ & &   3.09 &  24.15 &  -3.25 & -17.96\% &  25.25\\
\\[-0.25cm]
$V_{t_{j}}^{-} > E[V_{t_{j}}^{-}]$ & &   8.17 &  47.32 &  -8.52 & -18.94\% &  50.26\\
$V_{t_{j}}^{-} < E[V_{t_{j}}^{-}]$ & &   2.34 &  31.22 &  -3.98 & -24.87\% &  32.60\\
\hline
\end{tabular}
\smallskip
\begin{scriptsize}
\parbox{\textwidth}{\emph{Note.} This table reports descriptive measures of the pre- and post-drift burst return, $R_{t_{j}}^{-}$ and $R_{t_{j}}^{+}$, and a normalized volume measure $V_{t_{j}}^{-}$. $E[\cdot]$ refers to the average taken across the entire sample (unconditional in Panel A) or over the number of identified positive and negative drift bursts (Panel B and C). $\rho$ is the correlation between $R_{t_{j}}^{-}$ and $R_{t_{j}}^{+}$, while $\sigma_{p} = \sqrt{(R_{t_{j}}^{-})^{2}+(R_{t_{j}}^{+})^{2}}$ is the volatility in the drift burst window. Returns are expressed in basis points (bps).}
\end{scriptsize}
\end{center}
\end{table}

\citet{huang-wang:09a} report $R_{t_{j}}^{-}$, $R_{t_{j}}^{+}$, $V_{t_{j}}^{-}$, $\rho$ and $\sigma_p$ based on \textit{simulated} data.\footnote{They employ the notation $R_{1/2}$ for $R_{t_{j}}^{-}$, $R_{1}$ for $R_{t_{j}}^{+}$, and $V_{1/2}$ for $V_{t_{j}}^{-}$.} In contrast, Table \ref{table:huang-wang-ES-1} -- shaped as Table 1 in that paper for comparison -- is computed from \textit{actual} market returns. The full sample results in Panel A are calculated from all available 5-minute intervals. We observe a tiny negative return autocorrelation, which is slightly stronger if the preceding return is less than average or volume is above average. Volatility increases marginally with volume. The predicted effects are thus present, but extremely weak.

We now condition on $t_{j}$ being a drift burst time (with critical value set to 4.0), allowing at most for one significant event per day. Grouping by the sign of $R_{t_{j}}^{-}$, we are left with 253 negative and 208 positive drift bursts. In Panel B -- C of Table \ref{table:huang-wang-ES-1}, we tabulate the corresponding averages. The results are striking. A drift burst of either sign is associated with high volatility, large volume, and substantial negative serial correlation. The typical negative drift burst entails a price move of $-37.39$bps compared to $+36.41$bps for a positive one.\footnote{These price changes are far out in the tails of the return distribution. A $-37.39$bps drop in five minutes translates to a loss of $-31.41\%$ in seven hours. In our sample it represents a tenfold increase in magnitude compared to a typical 5-minute negative stock market return and corresponds to a four standard deviations draw, or about 1 in 25,000.} The subsequent $5$-minute log-return is $+8.01$bps and $-5.44$bps, on average. The serial correlation is $-0.7958$ ($-0.2573$) for negative (positive) drift bursts. Looking further at price drops, we split the sample into negative drift bursts with a smaller and larger return or volume than that of an average event. The negative serial correlation is most pronounced in drops that are below average (more negative) or accompanied by high trading volume. These effects are not present in positive drift bursts. At last, volatility is markedly higher for tail drift bursts (in either direction) and those with large volume.

\begin{table}[t!]
\setlength{\tabcolsep}{0.35cm}
\begin{center}
\caption{ES post-drift burst return conditional on pre-drift burst return and volume. \label{table:huang-wang-ES-2}}
\smallskip
\begin{tabular}{lccccccccc}
\hline
 & \multicolumn{4}{c}{negative drift burst} & & \multicolumn{4}{c}{positive drift burst} \\ \cline{2-5} \cline{7-10} \\[-0.25cm]
sorting variable & \multicolumn{9}{c}{$V_{t_{j}}^{-}$} \\[0.25cm]
$R_{t_{j}}^{-}$ & low & medium & high & high-low & & low & medium & high & high-low \\[0.25cm]
low & 14.3 &  8.5 & 42.8 & 28.4 & & -0.9 & -3.0 & -2.7 & -1.8 \\
medium &  5.1 &  4.9 &  6.8 &  1.7 & & -3.1 & -5.2 & -4.5 & -1.4 \\
high &  1.8 &  4.0 &  0.5 & -1.3 & & -8.1 & -12.0 & -9.6 & -1.5 \\
\hline
\end{tabular}
\smallskip
\begin{scriptsize}
\parbox{0.975\textwidth}{\emph{Note}. We report the average post-drift burst return $R_{t_{j}}^{+}$ after double-sorting first by pre-drift burst return $R_{t_{j}}^{-}$ and second by normalized gross trading volume $V_{t_{j}}^{-}$. As in \citet*{huang-wang:09a}, the sample is divided in three subsets: ``low'' is defined as the 1st quartile, ``medium'' is the inter-quartile range, and ``high'' the 4th quartile. The column ``high-low'' reports the difference between the high and low subgroup. Returns are expressed in basis points (bps).}
\end{scriptsize}
\end{center}
\end{table}

Table \ref{table:huang-wang-ES-2} reports conditional averages of the post-drift burst return $R_{t_{j}}^{+}$, which is double-sorted by pre-drift burst return $R_{t_{j}}^{-}$ and volume $V_{t_{j}}^{-}$. As before, the table is shaped as Table 2 in \citet{huang-wang:09a} to facilitate comparison. It reinforces that the reversion is stronger if the price change or volume is large. The results are again more pronounced for negative drift bursts.

As an alternative way to describe the double-conditioning of Table \ref{table:huang-wang-ES-2}, we estimate the forecasting equation of \citet{campbell-grossman-wang:93a}:
\begin{equation} \label{equation:cgw}
R_{t_{j}}^{+} = a + b R_{t_{j}}^{-} + c R_{t_{j}}^{-}V_{t_{j}}^{-} + \epsilon_{t_{j}},
\end{equation}
where an interaction term $R_{t_{j}}^{-}V_{t_{j}}^{-}$ is added to Eq. \eqref{eqn:revreg}. This approach is suggested by \citet*{huang-wang:09a} in order to assess if the return predictability is driven by trading volume. The parameter estimates of Eq. \eqref{equation:cgw} are displayed in Table \ref{table:campbell-grossmann-wang-ES} along with $t$-statistics (in parenthesis) based on \citet{newey-west:87a}-robust standard errors. Once we control for volume, the serial correlation observed during a negative drift burst is largely subsumed by the interaction term and the coefficient on $R_{t_{j}}^{-}$ is heavily reduced, albeit it remains significant for drops below average. This does not occur for positive drift bursts. We again split the sample in half according to the size of the pre-drift burst return and -- as predicted by \citet*{huang-wang:09a} -- we note that the volume-return variable is more important for large negative drops. On the other hand, for positive drift bursts it is largely irrelevant in explaining the post-drift burst return.

To summarize, Table \ref{table:huang-wang-ES-1} -- \ref{table:campbell-grossmann-wang-ES} confirm the theoretical predictions made by \citet*{grossman-miller:88a} and \citet*{huang-wang:09a}. The reversion in drift bursts is consistent with market makers absorbing large orders and charging a fee for this service. The asymmetry between negative and positive drift bursts is a symptom of endogenous demand for liquidity due to costly market participation. In the latter case, we confirm several testable implications on trading volume (regarded as a proxy of order imbalance), in that compensation increases with volume, which is larger during negative drift bursts. Overall, our findings suggest the regular occurrence of ``flash crashes'' documented here is consistent with existing theories of liquidity provision.

\begin{table}[t!]
\setlength{\tabcolsep}{0.60cm}
\begin{center}
\caption{ES post-drift burst return forecasting equation. \label{table:campbell-grossmann-wang-ES}}
\smallskip
\begin{tabular}{lcccccccc}
\hline
 & & \multicolumn{3}{c}{negative drift burst} & & \multicolumn{3}{c}{positive drift burst}\\
\cline{3-5} \cline{7-9}
                        & & $a$ & $b$ & $c$                           & & $a$ & $b$ & $c$\\
all           & & -11.06 & -0.51 & --               & &  0.56 & -0.16 & -- \\[-0.15cm]
              & & {\footnotesize (-3.50)} & {\footnotesize(-5.34)} &               & & {\footnotesize ( 0.20)} & {\footnotesize(-1.78)} \\[0.25cm]
     & & -2.69 & -0.17 & -0.02      & &  0.45 & -0.16 & -0.00 \\[-0.15cm]
     & & {\footnotesize (-1.41)} & {\footnotesize(-2.17)} & {\footnotesize(-6.06)}     & & {\footnotesize ( 0.16)} & {\footnotesize(-1.66)} & {\footnotesize(-0.46)} \\[0.25cm]
$R_{t_{j}}^{-} > E[R_{t_{j}}^{-}]$ & &  0.54 & -0.13 & -0.01      & & -1.34 & -0.13 & -0.00 \\[-0.15cm]
     & & {\footnotesize ( 0.28)} & {\footnotesize(-1.10)} & {\footnotesize(-0.78)}     & & {\footnotesize (-0.15)} & {\footnotesize(-0.72)} & {\footnotesize(-0.51)} \\[0.25cm]
$R_{t_{j}}^{-} < E[R_{t_{j}}^{-}]$ & & -14.58 & -0.34 & -0.01      & &  0.75 & -0.15 & -0.00 \\[-0.15cm]
     & & {\footnotesize (-3.52)} & {\footnotesize(-3.96)} & {\footnotesize(-3.86)}     & & {\footnotesize ( 0.36)} & {\footnotesize(-1.36)} & {\footnotesize(-0.37)} \\
\hline
\end{tabular}
\smallskip
\begin{scriptsize}
\parbox{0.975\textwidth}{\emph{Note}. This table shows the estimated OLS coefficients from the \citet{campbell-grossman-wang:93a} forecasting equation $R_{t_{j}}^{+} = a + b R_{t_{j}}^{-} + c R_{t_{j}}^{-}V_{t_{j}}^{-} + \epsilon_{t_{j}}$, where $R^+_{t_{j}}$ is the post-drift burst return, $R_{t_{j}}^{-}$ the pre-drift burst return, and $V_{t_{j}}^{-}$ is a pre-drift burst normalized gross trading volume. $t$-statistics of the parameter estimates -- based \citet*{newey-west:87a}-corrected standard errors -- are reported in parenthesis. The regression is fitted separately for negative (left panel) and positive (right panel) detected drift burst with $c=0$ (first row); with $c\neq 0$ (second row); after conditioning on $R_{t_{j}}^{-}$ being above (third row) and below (fourth row) its average value $E[R_{t_{j}}^{-}]$ in each of the two subsamples (negative and positive drift bursts).}
\end{scriptsize}
\end{center}
\end{table}

\section{Conclusion} \label{section:conclusion}

The drift burst hypothesis is proposed as a theoretical framework for the modeling of distinct and sustained trends in the price paths of financial assets. We show how drift bursts -- defined as a short-lived locally explosive drift coefficient -- can be embedded into standard continuous-time models and demonstrate that the arbitrage-free property is preserved if the volatility co-explodes during a drift burst, something we provide strong empirical support for. Applying a novel methodology for drift burst identification to a comprehensive set of tick data covering six major asset classes, we deliver unprecedented insights into potentially disruptive but poorly understood events, such as flash crashes. In contrast to the existing literature, which has mostly regarded these events as market glitches, we show that they are instead a regular, stylized fact in the markets, whose dynamic features match theoretical predictions from the market microstructure literature of liquidity provision.

The paper contributes towards a better understanding of the microstructure dynamics of financial markets and can help to inform the regulatory policy agenda going forward by shedding light on a number of important questions: Who triggers a flash crash? Who supplies liquidity during the event? What is the role played by high-frequency traders? And what, if anything, can be done to prevent flash crashes in the future?

\pagebreak

\appendix

\section{Mathematical Appendix} \label{appendix:proof}

In this section, $C$ is a generic positive constant, whose value may change from line to line.

\begin{assumption} \label{assumption:model2} With the notation of Assumption \ref{assumption:model}, we further impose that

\noindent i) The jump process $J_t$ is of the form:
\begin{equation}
J_{t} = \int_{0}^{t} \int_{ \mathbb{R}} \delta(s,x) I_{ \{| \delta(s,x)| \leq 1 \}} \left( \nu( \mathrm{d}s, \mathrm{d}x) - \tilde{ \nu}( \mathrm{d}s, \mathrm{d}x) \right) + \int_{0}^{t} \int_{ \mathbb{R}} \delta(s,x) I_{ \{| \delta(s,x)|> 1 \}} \nu( \mathrm{d}s, \mathrm{d}x),
\end{equation}
where $\nu$ is a Poisson random measure on $\mathbb{R}_{+} \times \mathbb{R}$, $\tilde{ \nu}( \mathrm{d}s, \mathrm{d}x) = \lambda( \mathrm{d}x) \mathrm{d}s$ a compensator, and $\lambda$ is a $\sigma$-finite measure on $\mathbb{R}$. Moreover, $\delta: \mathbb{R}_{+} \times \mathbb{R} \rightarrow \mathbb{R}$ is a predictable function such that there exists a sequence $(\tau_{n})_{n \geq 1}$ of $\mathcal{F}_{t}$-stopping times with $\tau_{n} \rightarrow \infty$ and, for each $n$, a deterministic and non-negative function $\Gamma_{n}$ with $\text{min}(| \delta(t,x)|,1) \leq \Gamma_{n}(x)$ and $\int_{ \mathbb{R}} \Gamma_{n}(x)^{2} \lambda( \mathrm{d}x) < \infty$ and, for all $\kappa\in (0,1)$, $\int_{\left\{x: \Gamma_{n}(x) \leq \kappa \right\} } \Gamma_{n}(x) \lambda( \mathrm{d}x) < \infty$, for all $(t,x)$ and $n \geq 1$.

\bigskip

\noindent ii) Fix $t \in (0,T]$. We assume there exists a $\Gamma > 0$ and a sequence $(\tau_{m})_{m \geq 1}$ of $\mathcal{F}_{t}$-stopping times with $\tau_{m} \rightarrow \infty$ and constants $C^{(m)}_{t}$ such that for all $m$,
\begin{equation} \label{equation:lipschitz}
E_{s} \left[| \mu_{u}- \mu_{s}|^{2}+| \sigma_{u}- \sigma_s|^{2} \right] \leq C^{(m)}_{t} |u - s|^{ \Gamma},
\end{equation}
for all $0 \leq s \leq u \leq T \wedge \tau_{m}$, where $E_{t} [\cdot] = E[ \cdot | \mathcal{F}_{t}]$.
\end{assumption}

\begin{remark}
With Assumption \ref{assumption:model} and \ref{assumption:model2} in hand, the localization procedure in \citet*[][Section 4.4.1]{jacod-protter:12a} implies that we can assume $\mu_{t}$, $\sigma_{t}$, and $\delta(t,x)$ are bounded, as $(\omega, t,x)$ vary within $\Omega \times [0,T] \times \mathbb{R}$, and that $| \delta(t,x)| \leq \overline{ \Gamma}(x)$, where $\overline{ \Gamma}(x)$ is bounded and such that $\int_{ \mathbb{R}} \overline{ \Gamma}(x)^{2} \lambda( \mathrm{d}x) < \infty$ and, for all $\kappa\in (0,1)$, $\int_{ \left\{x: \overline{ \Gamma}(x) \leq \kappa \right\}} \overline{ \Gamma}(x) \lambda( \mathrm{d}x) < \infty$.
\end{remark}

\begin{assumption} \label{assumption:times} $(t_{i})_{i=0}^{n}$ is a deterministic sequence. We denote by $\Delta_{i,n} = t_{i} - t_{i-1}$, $\Delta_{n}^{-} = \min_{i=1, \ldots, n} \{ \Delta_{i,n} \}$, $\Delta_{n}^{+} = \max_{i=1, \ldots, n} \{ \Delta_{i,n} \}$ and assume that, for a sufficiently large $n$, and suitable constants $c$ and $C$ (that do not depend on $n$),
\begin{equation*}
c \Delta_{n} \leq \Delta_{n}^{-} \leq \Delta_{n}^{+} \leq C \Delta_{n},
\end{equation*}
where $\Delta_{n} = T/n$. Moreover, denoting the ``quadratic variation of time up to $t$'' as $H(t) = \lim_{n \rightarrow \infty} H_{n}(t)$, where $H_{n}(t) = \frac{1}{ \Delta_{n}} \sum_{t_{i} \le t} \left( \Delta_{i,n} \right)^{2}$, we assume $H(t)$ exists and is Lebesgue-almost surely differentiable in $(0,T)$ with derivative $H'$ such that $\left|H'(t_{i})- \Delta_{i,n} / \Delta_{n} \right| \leq C \Delta_{i,n}$, for any $t_{i}$ in which $H$ is differentiable, where $C$ does not depend on $i$ and $n$.
\end{assumption}

\begin{assumption} \label{assumption:kernel} The bandwidths $h_{n}$, $h_{n}'$ are sequences of positive real numbers, such that, as $n \rightarrow \infty$, $h_{n} \rightarrow 0$, $h_{n}' \rightarrow 0$,  $nh_{n} \rightarrow \infty$, and $nh_{n}' \rightarrow \infty$. The kernel $K: \mathbb{R} \rightarrow \mathbb{R}_{+}$ is any function with the properties:
\begin{enumerate}
\item[(K0)] $K(x) = 0$ for $x>0$;
\item[(K1)] $K$ is bounded and differentiable with bounded first derivative. Further, $xK(x) \rightarrow 0$ and $xK'(x) \rightarrow 0$ as $x \rightarrow -\infty$;
\item[(K2)] $\int_{- \infty}^{0} K(x) \text{\upshape{d}}x = 1$ and $K_{2} = \int_{-\infty}^{0} K^{2}(x) \text{\upshape{d}}x < \infty$;
\item[(K3)] It holds that for every positive sequence $G_{n,t} \rightarrow \infty$, $\int_{-\infty}^{-G_{n,t}} K(x)\text{\upshape{d}}x \leq C G_{n,t}^{-B}$ for some $B> 0$ and $C>0$ (i.e., $K$ has a fast vanishing tail);
\item[(K4)] $m_{K}(\alpha) = \int_{-\infty}^0 K(x)|x|^{\alpha}\text{\upshape{d}}x<\infty$, for all $\alpha > -1$; $m'_{K}(\alpha) = \int_{-\infty}^0 K^2(x)|x|^{\alpha}\text{\upshape{d}}x<\infty$, for all $\alpha > -1$.
\end{enumerate}
\end{assumption}

\begin{remark} Condition (K0) can be extended to allow for two-sided kernels without changing any of the theoretical results and only minor modifications in the proofs. We impose it to make the mathematical exposition less cumbersome and aligned with the implementation in the empirical application.
\end{remark}
Without loss of generality, in the proofs we set $h_{n}' = h_{n}$.

\begin{lemma} \label{lemma:discretization}
Assume that the conditions of Assumptions \ref{assumption:model} and \ref{assumption:times} -- \ref{assumption:kernel} hold. Then, for every fixed $t \in (0,T]$ as $n \rightarrow \infty$ and $h_{n} \rightarrow$ 0, it holds that:
\begin{align*}
A_{n} &= \frac{1}{h_{n}} \sum_{i=1}^{n} K \left( \frac{t_{i-1}-t}{h_{n}} \right) \int_{t_{i-1}}^{t_{i}} \mu_{s} \text{\upshape{d}}s - \int_{0}^{T} \frac{1}{h_{n}} K \left( \frac{s-t}{h_{n}} \right) \mu_{s} \text{\upshape{d}}s = O_{p} \left( \frac{1}{nh_{n}} \right), \\[0.25cm]
B_{n} &= \frac{1}{ \Delta_{n}h_{n}} \sum_{i=1}^{n} K \left( \frac{t_{i-1}-t}{h_{n}} \right) \left( \int_{t_{i-1}}^{t_{i}} \mu_{s} \text{\upshape{d}}s \right)^{2} - \int_{0}^{T} \frac{1}{h_{n}} K \left( \frac{s-t}{h_{n}} \right) \mu_{s}^{2} \text{\upshape{d}}s = O_{p} \left( \frac{1}{nh_{n}} \right).
\end{align*}
This also applies if $\mu_{t}$ is replaced by $\sigma_{t}$.
\end{lemma}

\begin{proof}
Write:
\begin{equation*}
A_{n} = \frac{1}{h_{n}} \sum_{i=1}^{n} \int_{t_{i-1}}^{t_i} \left( K \left(\frac{t_{i-1}-t}{h_{n}}\right)-K \left(\frac{s-t}{h_{n}}\right)\right) \mu_s \text{d}s.
\end{equation*}
The mean value theorem -- together with the boundedness of $K'$ and $\mu_{t}$ -- implies that, for each interval $[t_{i-1}, t_{i}]$, there exists a $\xi_{t_{i-1},t_{i}}$ such that:
\begin{align*}
|A_{n}| &\leq \frac{1}{h_{n}} \sum_{i=1}^{n} \int_{t_{i-1}}^{t_{i}} \Big| K' \left( \frac{\xi_{t_{i-1},t_{i}}-t}{h_{n}} \right)(s-t_{i-1}) \Big| | \mu_{s}| \text{d}s \leq C \frac{T}{n} \frac{1}{h_{n}} \int_{0}^{T} | \mu_{s}| \text{d}s \leq C \frac{1}{nh_{n}},
\end{align*}
where the second inequality follows from $(K1)$. The proof for $B_{n}$ is nearly identical. \qed
\end{proof}

\begin{lemma} \label{lemma:localization}
Assume that the conditions of Assumption \ref{assumption:model} and \ref{assumption:times} -- \ref{assumption:kernel} hold. Then, for every fixed $t \in (0,T]$ as $n \rightarrow \infty$ and $h_{n} \rightarrow 0$, it holds that:
\begin{equation*}
A'_{n} = \int_{0}^{T} \frac{1}{h_{n}} K \left( \frac{s-t}{h_{n}} \right) \mu_{s} \text{\upshape{d}}s - \mu_{t-} = O_{p} \left(h_{n}^{ \Gamma/2} + h_{n}^{B} \right).
\end{equation*}
This also applies if $\mu_{t}$ is replaced by $\sigma_{t}$.
\end{lemma}

\begin{proof}
Note that, by the properties of the kernel,
\begin{equation*}
\mu_{t-} = \mu_{t-} \int_{- \infty}^{0} K(x) \text{d}x = \mu_{t-} \left( \int_{- \infty}^{-t/h_{n}} K(x) \text{d}x + \int_{-t/h_{n}}^{0} K(x) \text{d}x \right),
\end{equation*}
so we can write:
\begin{align*}
|A'_{n}| &= \Big| \int_{0}^{T} \frac{1}{h_{n}}K \left( \frac{s-t}{h_{n}} \right) (\mu_{s} - \mu_{t-}) \text{d}s - \mu_{t-} \int_{- \infty}^{-t/h_{n}}K(x) \text{d}x \Big| \leq \int_{0}^{T} \frac{1}{h_{n}} K \left( \frac{s-t}{h_{n}} \right) | \mu_{s} - \mu_{t-}| \text{d}s + C h_{n}^{B},
\end{align*}
where $B$ is the constant in (K3). Then, by Jensen's inequality and Eq. \eqref{equation:lipschitz}:
\begin{equation*}
E_{s \wedge t} \left[ |\mu_s-\mu_{t-}| \right] \leq C|s-t|^{\Gamma/2}.
\end{equation*}
Together with (K4) and a change of variable, this implies that:
\begin{align*}
E\left[ \int_{0}^{T} \frac{1}{h_{n}}K \left( \frac{s-t}{h_{n}} \right)|\mu_s-\mu_{t-} |\text{d}s \right] & \leq \int_{0}^{T} \frac{1}{h_{n}} K \left( \frac{s-t}{h_{n}} \right)|s-t|^{\Gamma/2} \text{d}s = \int_{-t/h_{n}}^{0} K(x)|x|^{\Gamma/2}h_{n}^{\Gamma/2} \text{d}x \leq C h_{n}^{\Gamma/2}.
\end{align*}
This concludes the proof. \qed
\end{proof}

\begin{lemma} \label{lemma:spot-volatility}
Assume that the conditions of Assumptions \ref{assumption:model} and \ref{assumption:model2} -- \ref{assumption:kernel} hold. Then, for every fixed $t \in(0,T]$, as $n \rightarrow \infty$ and $h_{n} \rightarrow 0$ such that $nh_{n} \rightarrow \infty$, it holds that $\hat{ \sigma}_{t}^{n} \overset{p}{ \rightarrow} \sigma_{t-}$.
\end{lemma}
\begin{proof} The lemma extends \citet*[][Theorem 9.3.2]{jacod-protter:12a}
to a general kernel (as defined in Assumption \ref{assumption:kernel}). We compensate the large jump term and write $X_{t}' = \int_{0}^{t} \mu_{s}^{*} \text{d}s + \int_{0}^{t} \sigma_{s} \text{d}W_{s}$, where $\mu_{t}^{*} = \mu_{t} + \int_{ \mathbb{R}} \delta(t,x) I_{ \{| \delta(t,x)| > 1 \}} \lambda( \text{d}x)$ is bounded, and $X_{t}'' = X_{t} - X_{t}' = \int_{ \mathbb{R}} \delta(s,x) \left( \nu( \text{d}s, \text{d}x) - \tilde{ \nu}( \text{d}s, \text{d}x) \right)$. Now,  \citet*{mancini-mattiussi-reno:15a} established the convergence in probability:
\begin{equation} \label{equation:mmr15}
\frac{1}{h_{n}} \sum_{i=1}^{n} K \left( \frac{t_{i-1} - t}{h_{n}} \right) \left( \Delta_{i}^{n} X' \right)^{2} \overset{p}{ \rightarrow} \sigma^2_{t-}.
\end{equation}
Thus, it is enough to show that:
\begin{equation*}
R_{n}^{ \hat{ \sigma}} = \frac{1}{h_{n}} \sum_{i = 1}^{n} K \left( \frac{t_{i-1} - t}{h_{n}} \right) \left( \left( \Delta_{i}^{n} X \right)^{2} - \left( \Delta_{i}^{n} X' \right)^{2} \right) \overset{p}{\rightarrow} 0.
\end{equation*}
Note that for $\kappa \in (0,1)$, we may write $\Delta_{i}^{n} X = \Delta_{i}^{n} X' + \Delta_{i}^{n} X_{1}''( \kappa) + \Delta_{i}^{n} X_{2}''( \kappa)$, where
\begin{equation*}
\Delta_{i}^{n} X_{1}''( \kappa)= \int_{ \mathbb{R}} \delta(s,x)I_{ \{\overline{\Gamma}(x) \leq \kappa \}} \left( \nu( \text{d}s, \text{d}x) - \tilde{ \nu}( \text{d}s, \text{d}x) \right) \quad \text{and} \quad \Delta_{i}^{n} X_{2}''( \kappa) = \int_{ \mathbb{R}} \delta(s,x)I_{ \{\overline{ \Gamma}(x) > \kappa \}} \left( \nu( \text{d}s, \text{d}x) - \tilde{ \nu}( \text{d}s, \text{d}x) \right).
\end{equation*}
Applying the decomposition in \citet*[][Equation 9.3.9]{jacod-protter:12a}:
\begin{align*}
| R_{n}^{ \hat{ \sigma}} | &\leq  \frac{1}{h_{n}} \sum_{i=1}^{n} K \left( \frac{t_{i-1} - t}{h_{n}} \right)  \big| \left( \Delta_{i}^{n} X \right)^{2} -\left( \Delta_{i}^{n} X' \right)^{2} \big| \\[0.10cm]
&\leq \frac{1}{h_{n}} \sum_{i=1}^{n} K \left( \frac{t_{i-1} - t}{h_{n}} \right) \left( \epsilon \left( \Delta_{i}^{n} X' \right)^{2} + \frac{C}{ \epsilon} \left( \left( \Delta_{i}^{n} X_{1}''( \kappa) \right)^{2} + \left( \Delta_{i}^{n} X_{2}''( \kappa) \right)^{2} \right) \right),
\end{align*}
for $0 < \epsilon \leq 1$ and $C$ independent of $\epsilon$.

Next, define $\Omega_{n}( \psi, \kappa) \subseteq \Omega$ the set such that the Poisson process $\nu([0,t] \times \{x: \overline{ \Gamma}(x) > \kappa \})$ has no jumps in the interval $(t-h_{n}^{ \psi}, t]$, for $0< \psi <1$ and $0< \kappa <1$. Note that $\Omega_{n}( \psi, \kappa) \rightarrow \Omega$, as $n \rightarrow \infty$, since $h_{n}^{ \psi} \rightarrow 0$.
Now write, for all $c>0$,  $\mathcal{P} \left(|R_{n}^{ \hat{ \sigma}} | > c \right)  = \mathcal{P} \left( |R_{n}^{ \hat{ \sigma}} | > c \mid \Omega_{n}^{ \complement}( \psi, \kappa) \right) +  \mathcal{P} \left( |R_{n}^{ \hat{ \sigma}}| > c \mid \Omega_{n}( \psi, \kappa) \right) \leq \mathcal{P} \left( \Omega_{n}^{ \complement}( \psi, \kappa) \right) + E \Big[ |R_{n}^{ \hat{ \sigma}}| \mid \Omega_{n}( \psi, \kappa) \Big] / c$ by Markov's inequality.

On $\Omega_{n}( \psi, \kappa)$, we can bound the ``large'' jumps before $t-h_{n}^{ \psi}$ as follows
\begin{align*}
E \Bigg[ \frac{1}{h_{n}} \sum_{i=1}^{n} K \left( \frac{t_{i-1} - t}{h_{n}} \right) \frac{C}{ \epsilon} \left( \Delta_{i}^{n} X_{2}''( \kappa) \right)^{2} \mid \Omega_{n}( \psi, \kappa) \Bigg] &= \frac{1}{h_{n}} \sum_{t_{i-1} \leq t-h_{n}^{ \psi}} K \left( \frac{t_{i-1} - t}{h_{n}} \right) \frac{C}{ \epsilon} E \Big[ \left( \Delta_{i}^{n} X_{2}''( \kappa) \right)^{2} \Big] \\[0.10cm]
&\leq \frac{1}{h_{n}} \sum_{t_{i-1} \leq t-h_{n}^{ \psi}} K \left( \frac{t_{i-1} - t}{h_{n}} \right) \frac{C}{ \epsilon} \Delta_{i,n} \leq \frac{C}{ \epsilon}h_{n}^{B(1 - \psi)}
\end{align*}
by a Riemann approximation and (K3). Moreover,
\begin{equation*}
E \Big[ \left( \Delta_{i}^{n} X_{1}'' ( \kappa) \right)^{2} \Big] \leq C \Delta_{i,n} \int_{ \left\{x: \overline{ \Gamma}(x) \leq \kappa \right\}} \overline{ \Gamma}(x)^2 \lambda( \text{d}x),
\end{equation*}
so that, using again the convergence in Eq. \eqref{equation:mmr15};
\begin{equation*}
E \Big[|R_{n}^{ \hat{ \sigma}} | \mid \Omega_{n}( \psi, \kappa) \Big] \leq C \epsilon + \frac{C}{ \epsilon} \left( \int_{\left\{x: \overline{ \Gamma}(x) \leq \kappa \right\}} \overline{ \Gamma}(x)^{2} \lambda( \text{d}x) + h_{n}^{B(1- \psi)}
\right).
\end{equation*}
It follow that for all $\epsilon$ and $\kappa \in (0,1)$,
\begin{equation*}
\underset{n \rightarrow \infty}{ \limsup} \ \mathcal{P} \left( |R_{n}^{ \hat{ \sigma}}| > c \right) \leq \frac{1}{c} \left(C \epsilon + \frac{C}{ \epsilon} \left( \int_{ \left\{x: \overline{ \Gamma}(x) \leq \kappa \right\}} \overline{ \Gamma}(x)^{2} \lambda( \text{d}x) + h_{n}^{B(1- \psi)} \right) \right).
\end{equation*}
Setting $\epsilon = \left( \int_{ \left\{x: \overline{ \Gamma}(x) \leq \kappa \right\}} \overline{ \Gamma}(x)^{2} \lambda( \text{d}x) + h_{n}^{B(1 - \psi)} \right)^{ \psi'}$ with $0 < \psi' < 1$ and noticing that $\left( \int_{ \left\{x: \overline{ \Gamma}(x) \leq \kappa \right\}} \overline{ \Gamma}(x)^{2} \lambda( \text{d}x) + h_{n}^{B(1-\psi)} \right) \rightarrow 0$ as $\kappa \rightarrow 0$ and $n \rightarrow \infty$, we deduce that $\mathcal{P} \left( |R_{n}^{\hat{ \sigma}}| > c \right) \rightarrow 0$. The convergence in probability $\hat{ \sigma}_{t}^{n} \overset{p}{ \rightarrow} \sigma_{t-}$ then follows from these results in combination with Slutsky's Theorem. \qed
\end{proof}

\begin{proof}[Proof of Theorem \ref{theorem:null}.] We decompose $T_{t}^{n}$ into:
\begin{align*}
T_{t}^{n} &=  \underbrace{ \sqrt{ \frac{h_{n}}{K_2}} \frac{\left( \hat{ \mu}_{t}^{n} - \mu_{t-}^{*} \right)}{ \hat{\sigma}_{t}^{n}}}_{T_{1}} + \underbrace{ \sqrt{ \frac{h_{n}}{K_{2}}} \frac{\mu_{t-}^{*}}{ \hat{ \sigma}_{t}^{n}}}_{T_{2}},
\end{align*}
where $\mu_{t}^{*}$ is the jump compensated drift as defined in the proof of Lemma \ref{lemma:spot-volatility}. This, together with the boundedness of $\mu_{t}^{*}$, $\sigma_{t}$ and $\delta(t,x)$, yields the following:
\begin{equation*}
|T_{2}| \leq C \frac{ \sqrt{h_{n}}}{ \hat{ \sigma}_{t}^{n}} = O_{p} \left( \sqrt{h_{n}} \right).
\end{equation*}
Now, from Lemma \ref{lemma:discretization} -- \ref{lemma:localization} we can write:
\begin{equation*}
\hat{ \mu}_{t}^{n} - \mu_{t-}^{*} = \frac{1}{h_{n}} \sum_{i=1}^{n} K \left( \frac{t_{i-1}-t}{h_{n}} \right) \Delta_{i}^{n} X - \frac{1}{h_{n}} \sum_{i=1}^{n} K \left( \frac{t_{i-1}-t}{h_{n}} \right) \int_{t_{i-1}}^{t_{i}} \mu_{s}^{*} \text{d}s + O_{p} \left( \frac{1}{nh_{n}} + h_{n}^{ \Gamma/2} + h_{n}^{B} \right).
\end{equation*}
Hence,
\begin{align*}
\sqrt{h_{n}} \left( \hat{ \mu}_{t}^{n} - \mu_{t-}^{*} \right) &= \underbrace{ \frac{1}{ \sqrt{h_{n}}} \sum_{i=1}^{n} K \left( \frac{t_{i-1}-t}{h_{n}} \right) \int_{t_{i-1}}^{t_{i}} \sigma_{s} \text{d}W_{s}}_{G_{n,t}} \\[0.10cm]
&+ \underbrace{ \frac{1}{ \sqrt{h_{n}}} \sum_{i=1}^{n} K \left( \frac{t_{i-1}-t}{h_{n}} \right) \int_{t_{i-1}}^{t_{i}} \int_{ \mathbb{R}} \delta(s,x) \left( \nu( \text{d}s, \text{d}x) - \tilde{ \nu}( \text{d}s, \text{d}x) \right)}_{G'_{n,t}} + O_{p} \left( \frac{ \sqrt{h_{n}}}{nh_{n}} + h_{n}^{ \Gamma/2 + 1/2} + h_{n}^{B + 1/2} \right).
\end{align*}
The $O_{p}(\cdot)$ term is asymptotically negligible, as $n \rightarrow \infty$, $h_{n} \rightarrow 0$ and $nh_{n} \rightarrow \infty$. $G'_{n,t}$ also vanishes, which we show as in the proof of Lemma \ref{lemma:spot-volatility} by writing: $\int_{t_{i-1}}^{t_{i}} \int_{ \mathbb{R}} \delta(s,x) \left( \nu( \text{d}s, \text{d}x) - \tilde{ \nu}( \text{d}s, \text{d}x) \right) = \Delta_{i}^{n} X_{1}''( \kappa) + \Delta_{i}^{n} X_{2}''( \kappa)$ with $\kappa \in (0,1)$ and writing $G'_{n,t}= G_{n,t,1}'+G_{n,t,2}'$, where $G_{n,t,j}' = \frac{1}{ \sqrt{h_{n}}} \sum_{i=1}^{n} K \left( \frac{t_{i-1}-t}{h_{n}} \right) \Delta_{i}^{n} X_{j}''( \kappa)$ for $j = 1, 2$. Now, on the set $\Omega_{n}( \psi, \kappa)$:
\begin{align*}
E\Bigg[ \frac{1}{ \sqrt{h_{n}}} \sum_{i=1}^{n} K \left( \frac{t_{i-1}-t}{h_{n}} \right) | \Delta_{i}^{n} X_{2}''( \kappa)| \mid \Omega_{n}( \psi, \kappa)\Bigg] &= \frac{1}{\sqrt{h_{n}}} \sum_{t_{i-1} \leq t-h_{n}^{ \psi}} K \left( \frac{t_{i-1} - t}{h_{n}} \right) E \Big[ | \Delta_{i}^{n} X_{2}''( \kappa)|\Big] \\[0.10cm]
&\leq \frac{1}{ \sqrt{h_{n}}} \sum_{t_{i-1} \leq t-h_{n}^{ \psi}} K \left( \frac{t_{i-1} - t}{h_{n}} \right) C \Delta_{i,n} \int_{ \left\{x: \overline{ \Gamma}(x) > \kappa \right\}} \overline{ \Gamma}(x) \lambda( \text{d}x)  \leq C \sqrt{h_{n}}h_{n}^{B(1- \psi)},
\end{align*}
where the first inequality is by virtue of \citet*[][Lemma 2.1.5]{jacod-protter:12a} applied with $p = 1$. Since $\mathcal{P} \left(|G'_{n,t,2}| > c \right) \leq \mathcal{P} \left( \Omega_{n}^{ \complement}( \kappa, \psi) \right) + E \Big[|G'_{n,t,2} | \mid \Omega_{n}( \kappa, \psi) \Big] / c$, it again follows that
\begin{equation*}
\underset{n \rightarrow \infty}{ \limsup} ~ \mathcal{P} \left(|G_{n,2}'| > c \right) \leq \frac{1}{c}C \sqrt{h_{n}}h_{n}^{B(1- \psi)},
\end{equation*}
so that we deduce $G_{n,t,2}' \overset{p}{ \rightarrow} 0$ as $n \rightarrow \infty$.

Next, from Lemma 2.1.5 in \citet*{jacod-protter:12a},
\begin{equation*}
E \Big[|G'_{n,t,1}| \Big] \leq \frac{1}{\sqrt{h_{n}}} \sum_{i=1}^{n} K\left( \frac{t_{i-1}-t}{h_{n}} \right) E \Big[ |\Delta_{i}^{n} X_{1}''( \kappa) | \Big] \leq \frac{C}{\sqrt{h_{n}}} \sum_{i=1}^{n} K\left( \frac{t_{i-1}-1}{h_{n}} \right) \Delta_{i,n} \int_{\left\{x: \overline{ \Gamma}(x) \leq \kappa \right\}} \overline{ \Gamma}(x)\lambda( \text{d}x),
\end{equation*}
and the bound has order $O_{p}( \sqrt{h_{n}})$ so that $G_{n,t,1}' \overset{p}{ \rightarrow} 0$ and, therefore, $G_{n,t}' = o_{p}(1)$.

$G_{n,t}$ is the leading term, which we write as $G_{n,t} = \sum_{i = 1}^{n} \Delta_{i}^{n} u$ with $\Delta_{i}^{n}u = \frac{1}{ \sqrt{h_{n}}} K \left( \frac{t_{i-1}-t}{h_{n}} \right) \int_{t_{i-1}}^{t_{i}} \sigma_{s} \text{d}W_{s}$. The aim is to prove that $G_{n,t}$ -- and hence $\sqrt{h_{n}} \left( \hat{ \mu}_{t}^{n} - \mu_{t-}^{*} \right)$ -- converges stably in law to $N \left(0, K_{2} \sigma_{t-}^{2} \right)$. We exploit Theorem 2.2.14 in \citet*{jacod-protter:12a}, which lists four sufficient conditions for this to hold:
\begin{align}
&\sum_{i=1}^{n} E_{t_{i-1}} \Big[ \Delta_{i}^{n}u \Big] \overset{p}{ \rightarrow} 0, \label{g1} \\[0.10cm]
&\sum_{i=1}^{n} E_{t_{i-1}} \Big[ ( \Delta_{i}^{n}u)^{2} \Big] \overset{p}{ \rightarrow} K_{2} \sigma_{t-}^{2}, \label{g2} \\[0.10cm]
&\sum_{i=1}^{n} E_{t_{i-1}} \Big[ ( \Delta_{i}^{n}u)^{4} \Big] \overset{p}{ \rightarrow} 0, \label{g3} \\[0.10cm]
&\sum_{i=1}^{n} E_{t_{i-1}} \Big[ \Delta_{i}^{n}u \Delta_{i}^{n} Z \Big] \overset{p}{ \rightarrow} 0, \label{g4}
\end{align}
where either $Z_{t} = W_{t}$ or $Z_{t} = W'_{t}$ with $W'_{t}$ being orthogonal to $W_{t}$. The condition in Eq. \eqref{g1} is immediate. Next, from It\^{o}'s Lemma, we deduce that:
\begin{equation*}
\left( \int_{t_{i-1}}^{t_{i}} \sigma_{s} \text{d}W_{s} \right)^{2} = \int_{t_{i-1}}^{t_{i}} \sigma^{2}_{s} \text{d}s + 2 \int_{t_{i-1}}^{t_{i}} \sigma_{s} \left( \int_{t_{i-1}}^{s} \sigma_{u} \text{d}W_{u} \right) \text{d}W_{s},
\end{equation*}
so that
\begin{align*}
\sum_{i=1}^{n} E_{t_{i-1}} \Big[ ( \Delta_{i}^{n} u)^{2} \Big] &= \sum_{i=1}^{n} \frac{1}{h_{n}} K^{2} \left( \frac{t_{i-1}-t}{h_{n}} \right) E_{t_{i-1}} \Bigg[ \left( \int_{t_{i-1}}^{t_{i}} \sigma_{s} \text{d}W_{s} \right)^{2} \Bigg] \\[0.25cm]
&= \sum_{i=1}^{n} \frac{1}{h_{n}} K^{2} \left( \frac{t_{i-1}-t}{h_{n}} \right) E_{t_{i-1}} \Bigg[ \int_{t_{i-1}}^{t_{i}} \sigma_{s}^{2} \text{d}s \Bigg] \\[0.25cm]
&= \sum_{i=1}^{n} \frac{1}{h_{n}} K^{2} \left( \frac{t_{i-1}-t}{h_{n}} \right) \left( \sigma^{2}_{t_{i-1}} \Delta_{i,n} + E_{t_{i-1}} \Bigg[ \int_{t_{i-1}}^{t_{i}} \left( \sigma_{s}^{2} - \sigma_{t_{i-1}}^{2} \right) \text{d}s \Bigg] \right).
\end{align*}
The first term converges to $K_{2} \sigma_{t-}^{2}$, as shown in \citet{mancini-mattiussi-reno:15a}. The second term is negligible by the Lipschitz condition in Eq. \eqref{equation:lipschitz}:
\begin{equation*}
\sum_{i=1}^{n} \frac{1}{h_{n}} K^{2} \left( \frac{t_{i-1}-t}{h_{n}} \right) E_{t_{i-1}} \Bigg[ \int_{t_{i-1}}^{t_{i}} \left( \sigma_{s}^{2} - \sigma^{2}_{t_{i-1}} \right) \text{d}s \Bigg] \leq \sum_{i=1}^{n} \frac{1}{h_{n}} K^{2} \left( \frac{t_{i-1}-t}{h_{n}} \right) \Delta_{i,n} \Delta_{i,n}^{ \Gamma} = O_{p} \left( \Delta_{n}^{ \Gamma} \right).
\end{equation*}
To deal with the third condition in Eq. \eqref{g3}, we note that by the Burkholder-Davis-Gundy inequality and from the boundedness of $\sigma_{t}$, it holds that:
\begin{equation*}
E_{t_{i-1}} \left[ \left( \int_{t_{i-1}}^{t_{i}} \sigma_{s} \text{d}W_{s} \right)^{4} \right] \leq C ( \Delta_{i,n})^{2},
\end{equation*}
which leads to
\begin{align*}
\sum_{i=1}^{n} E_{t_{i-1}} \Big[ ( \Delta_{i}^{n} u)^{4} \Big] &= \sum_{i=1}^{n} \frac{1}{h_{n}^{2}} K^{4} \left( \frac{t_{i-1}-t}{h_{n}} \right) E_{t_{i-1}} \left[ \left( \int_{t_{i-1}}^{t_{i}} \sigma_{s} \text{d}W_{s} \right)^{4} \right] \\[0.25cm]
&\leq C \sum_{i=1}^{n} \frac{1}{h_{n}^{2}} K^{4} \left( \frac{t_{i-1}-t}{h_{n}} \right) (\Delta_{i,n})^{2} = O \left( \frac{ \Delta_{n}}{h_{n}} \right).
\end{align*}
To deal with Eq. \eqref{g4}, we first set $Z_{t} = W_{t}$. Then, using the Cauchy-Schwarz inequality:
\begin{align*}
E_{t_{i-1}} \Bigg[ \Delta_{i}^{n} W \int_{t_{i-1}}^{t_{i}} \sigma_{s} \text{d}W_{s} \Bigg] &\leq \sqrt{E_{t_{i-1}} \Big[ ( \Delta_{i}^{n} W)^{2} \Big]} \sqrt{E_{t_{i-1}} \left[ \left( \int_{t_{i-1}}^{t_{i}} \sigma_{s} \text{d}W_{s} \right)^{2} \right]} \\[0.25cm]
&= \sqrt{ \Delta_{i,n}} \sqrt{E_{t_{i-1}} \Bigg[ \int_{t_{i-1}}^{t_{i}} \sigma_{s}^{2} \text{d}s \Bigg]} = O_{p} \left( \Delta_{n} \right),
\end{align*}
and therefore
\begin{equation*}
\sum_{i=1}^{n} E_{t_{i-1}} \big[ \Delta_{i}^{n}u \Delta_{i}^{n} W \big] \leq C \frac{1}{ \sqrt{h_{n}}} \sum_{i=1}^{n} K \left( \frac{t_{i-1}-t}{h_{n}} \right) \Delta_{i,n} \rightarrow 0.
\end{equation*}
If $Z_{t} = W_{t}'$, the process $W_{t}' \int_{0}^{t} \sigma_{s} \text{d}W_{s}$ is a martingale by orthogonality, so that:
\begin{equation*}
E_{t_{i-1}} \Bigg[ \Delta_{i}^{n} W' \int_{t_{i-1}}^{t_{i}} \sigma_{s} \text{d}W_{s} \Bigg] = 0.
\end{equation*}
This verifies that $\sqrt{h_{n}} \left( \hat{ \mu}_{t}^{n} - \mu_{t-}^{*} \right) \overset{d}{ \rightarrow} N \left(0, K_{2} \sigma_{t-}^{2} \right)$, where the convergence is in law stably. Combined with Lemma \ref{lemma:spot-volatility}, this yields $T_{t}^{n} \overset{d}{ \rightarrow} N(0,1)$. \qed
\end{proof}

\begin{proof}[Proof of Theorem \ref{theorem:alternative}.]

Without loss of generality, we set $\tau_{\text{\upshape{db}}}= 1$. We write $\widetilde{X}_{t} = X_{t}+D_{t}+V_{t}$, where $D_{t} = \int_{0}^{t} c_{1,s}(1-s)^{- \alpha} \text{d}s$ and $V_{t} = \int_{0}^{t}c_{2,s}(1-s)^{- \beta} \text{d}W_{s}$.

Look at the term $\hat{ \mu}_{t}^{n}$. In Theorem \ref{theorem:null}, we already showed that $\frac{1}{h_{n}} \sum_{i=1}^{n}K \left( \frac{t_{i-1}-1}{h_{n}} \right) \Delta_{i}^{n} X = O_{p} \left( \frac{1}{ \sqrt{h_{n}}} \right)$. Now,
\begin{equation*}
A_n = \frac{1}{h_{n}} \sum_{i=1}^{n}K \left( \frac{t_{i-1}-1}{h_{n}} \right) \int_{t_{i-1}}^{t_i}c_{1,s} (1-s)^{-\alpha} \text{d}s = \frac{1}{h_{n}} \sum_{i=1}^{n}K \left( \frac{t_{i-1}-1}{h_{n}} \right) \Delta_{i,n}c_{1,\xi_{t_{i-1},t_i}}(1 - \xi_{t_{i-1},t_i})^{- \alpha},
\end{equation*}
where $t_{i-1} \leq \xi_{t_{i-1}, t_{i}} \leq t_{i}$.
The last term is, following a change of variable and Riemann summation, asymptotically equivalent to $c_{1,1}h_{n}^{- \alpha}m_{K}(- \alpha)$, where $m_{K}(- \alpha)$ is the constant in (K4).

As for the second term:
\begin{equation*}
B_{n} = \frac{1}{h_{n}} \sum_{i=1}^{n}K \left( \frac{t_{i-1}-1}{h_{n}} \right) \int_{t_{i-1}}^{t_i}c_{2,s} (1-s)^{- \beta} \text{d}W_{s},
\end{equation*}
we apply again \citet*[][Theorem 2.2.14]{jacod-protter:12a}. The crucial part is the variance (we omit the other three conditions that are uncomplicated here), and for this term we find:
\begin{align*}
\frac{1}{h_{n}^{2}} \sum_{i=1}^{n} K^{2} \left( \frac{t_{i-1}-1}{h_{n}} \right) E\left[ \left( \int_{t_{i-1}}^{t_i} c_{2,s} (1-s)^{- \beta} \text{d}W_{s} \right)^{2} \right] &= \frac{1}{h_{n}^{2}} \sum_{i=1}^{n} K^{2} \left( \frac{t_{i-1}-1}{h_{n}} \right) E \left[ \int_{t_{i-1}}^{t_{i}}c_{2,s}^{2} (1-s)^{-2 \beta} \text{d}{s} \right] \overset{p}{ \sim} h_{n}^{-1-2 \beta}m'_{K}(-2 \beta)c_{2,1}^{2},
\end{align*}
which implies that
\begin{equation*}
h_{n}^{1/2 + \beta} B_{n} \overset{d}{ \rightarrow} N \big(0,m_{K}'(-2 \beta)c_{2,1}^{2} \big),
\end{equation*}
where the convergence is stable in law. Thus, $\hat{ \mu}_{t}^{n}$ is dominated by $A_{n}$ when $\alpha- \beta>1/2$, by $B_{n}$ when $\alpha-\beta < 1/2$, whereas both terms influence the limit with $\alpha - \beta = 1/2$.

We turn to the denominator and set $( \hat{ \sigma}_{t}^{n})^{2} = \frac{1}{h_{n}} \sum_{i=1}^{n}K \left( \frac{t_{i-1}-1}{h_{n}} \right)\left[( \Delta_{i}^{n} D)^{2} + ( \Delta_{i}^{n} V)^{2} \right] + R_{n}'$, where
\begin{equation*}
R_{n}' = \frac{1}{h_{n}} \sum_{i=1}^{n}K \left( \frac{t_{i-1}-1}{h_{n}} \right) \left(( \Delta_{i}^{n} X + \Delta_{i}^{n} D + \Delta_{i}^{n} V)^{2} - ( \Delta_{i}^{n} D)^{2} - ( \Delta_{i}^{n} V)^{2} \right).
\end{equation*}
Since for all $\epsilon > 0$ and $a$, $b$ and $c$ real: $(a + b + c)^{2} - a^{2} - b^{2} \leq \epsilon (a^{2} + b^{2}) + \frac{1 + \epsilon}{ \epsilon}c^{2}$, we deduce the bound:
\begin{align*}
\begin{split}
|R_{n}'| \leq \frac{1}{h_{n}} \sum_{i=1}^{n}K \left( \frac{t_{i-1}-1}{h_{n}} \right) \left( \epsilon \left(( \Delta_{i}^{n} D)^{2} + ( \Delta_{i}^{n} V)^{2} \right) + \frac{1 + \epsilon}{ \epsilon} \left( \Delta_{i}^{n} X \right)^{2} \right) \\[0.10cm]
= \epsilon  \frac{1}{h_{n}} \sum_{i=1}^{n}K \left( \frac{t_{i-1}-1}{h_{n}} \right) \left( ( \Delta_{i}^{n} D)^{2} + ( \Delta_{i}^{n} V)^{2} \right)+ \frac{1+ \epsilon}{ \epsilon} O_{p}(1).
\end{split}
\end{align*}
Setting $\epsilon = h_{n}^{c}$, for a suitable $c>0$, $R_{n}'$ is at most $O_{p} \left(h_{n}^{-c} \right)$ and therefore negligible compared to the two leading terms. Next, setting $\Theta_{n} = \frac{1}{h_{n}} \sum_{i=1}^{n}K \left( \frac{t_{i-1}-1}{h_{n}} \right)( \Delta_{i}^{n} V)^{2}$ and using It\^{o}'s Lemma:
\begin{equation*}
( \Delta_{i}^{n} V)^{2} = \left( \int_{t_{i-1}}^{t_i}c_{2,s} (1-s)^{-\beta} \text{d}W_s\right)^{2} = 2 \int_{t_{i-1}}^{t_i}\left(\int_{t_{i-1}}^sc_{2,u} (1-u)^{-\beta} \text{d} W_u\right) c_{2,s} (1-s)^{-\beta}\text{d}W_s + \int_{t_{i-1}}^{t_i}c^2_{2,s} (1-s)^{-2\beta} \text{d}s,
\end{equation*}
and we split $\Theta_n = \Theta_{1,n} + \Theta_{2,n}$ accordingly. As for the term $A_{n}$, $\Theta_{2,n}$ is asymptotically equivalent to $c_{2,1}^{2} h_{n}^{-2 \beta}m_{K}(-2 \beta)$. We write $\Theta_{1,n}=\sum_{i=1}^{n} \Delta u'_i $, where
\begin{equation*}
\Delta u'_i = \frac{1}{h_n}K \left( \frac{t_{i-1}-1}{h_{n}}\right) 2\int_{t_{i-1}}^{t_i}c_{2,s}(1-s)^{-\beta}\left(\int_{t_{i-1}}^{s}c_{2,u}(1-u)^{-\beta} \text{d}W_{u} \right) \text{d} W_{s},
\end{equation*}
so $\Theta_{1,n}$ is a sum of martingale differences. Now, using a Taylor expansion on $c_{2,s}$,
\begin{align*}
\sum_{i=1}^{n} E_{t_{i-1}} \left[ \left( \Delta u_{i}' \right)^{2} \right] &= \frac{1}{h_n^2} \sum_{i=1}^{n}K^2 \left( \frac{t_{i-1}-1}{h_{n}} \right) \big(c^4_{2,t_{i-1}} + o_p(1) \big) E_{t_{i-1}} \left[\left(2 \int_{t_{i-1}}^{t_i}(1-s)^{- \beta} \left( \int_{t_{i-1}}^{s}(1-u)^{-\beta} \text{d}W_u \right) \text{d}W_s\right)^{2}\right] \\[0.10cm]
  & = \frac{4}{h_n^2}   \sum_{i=1}^{n}K^2 \left( \frac{t_{i-1}-1}{h_{n}} \right) \big(c^4_{2,t_{i-1}} +o_p(1) \big) \int_{t_{i-1}}^{t_i}(1-s)^{-2\beta}E_{t_{i-1}}\left[\left(\int_{t_{i-1}}^{s}(1-u)^{-\beta} \text{d}W_u\right)^2\right]\text{d}s \\[0.10cm]
  & =  \frac{4}{h_n^2}    \sum_{i=1}^{n}K^2 \left( \frac{t_{i-1}-1}{h_{n}} \right) \big(c^4_{2,t_{i-1}} +o_p(1) \big) \int_{t_{i-1}}^{t_i}(1-s)^{-2\beta}\left(\int_{t_{i-1}}^{s}(1-u)^{-2\beta} \text{d}u\right)\text{d}s \\[0.10cm]
  & =  \frac{2}{h_n^2}   \sum_{i=1}^{n}K^2 \left( \frac{t_{i-1}-1}{h_{n}} \right) \big(c^4_{2,t_{i-1}} +o_p(1) \big) \left((1-t_{i-1})^{-4 \beta} \Delta_{i,n}^{2} + O( \Delta_{i,n}^{3}) \right),
 \end{align*}
where the $o_{p}(1)$ term is uniform in $i$. If $\beta <1/4$, Lemma A.1 (ii) in \citet{mancini-mattiussi-reno:15a} implies:
\begin{equation*}
\frac{h_{n}^{1+4 \beta}}{ \Delta_{n}} \sum_{i=1}^{n} E_{t_{i-1}} \left[ \left( \Delta u_{i}' \right)^{2} \right] \overset{p}{ \rightarrow}2K_{2}m_{K}'(-4 \beta)H'(1)c_{2,1}^{4},
\end{equation*}
which means that $\Theta_{1,n}$ is dominated by $\Theta_{2,n}$. On the other hand, if $\beta \geq 1/4$, we deduce from Assumption \ref{assumption:times} and the boundedness of $K$ and $c_{2,t}$ that
\begin{equation*}
\frac{2}{h_{n}^{2}} \sum_{i=1}^{n} K^{2} \left( \frac{t_{i-1}-1}{h_{n}} \right) c_{2,t_{i-1}}^{4} (1-t_{i-1})^{-4 \beta} \Delta_{i,n}^{2}  \leq C\frac{1}{h_n^2} \sum_{i=1}^{n} (\Delta^-_n (n-(i-1))^{-4\beta}(\Delta_n^+)^2\leq C \frac{\Delta_n^{2-4\beta}}{h_n^2} \sum_{j=1}^{n} j^{-4\beta} \leq C \frac{\Delta_n^{2-4\beta}}{h_n^2} \log(n),
\end{equation*}
so this term is also dominated by $\Theta_{2,n}$.

With the remaining term in $( \hat{ \sigma}_{t}^{n})^{2}$, it holds that:
\begin{align*}
A_{n}' &= \frac{1}{h_{n}} \sum_{i=1}^{n}K \left( \frac{t_{i-1}-1}{h_{n}} \right)( \Delta_{i}^{n} D)^{2} = \frac{1}{h_{n}} \sum_{i=1}^{n}K \left( \frac{t_{i-1}-1}{h_{n}} \right) \left(c_{1,t_{i-1}}^{2} + o_{p}(1) \right) \left( \int_{t_{i-1}}^{t_i}(1-s)^{- \alpha} \text{d}s \right)^{2} \\[0.10cm]
&  \overset{p}{\sim} \frac{1}{h_{n}} \sum_{i=1}^{n}K \left( \frac{t_{i-1}-1}{h_{n}} \right)c_{1,t_{i-1}}^2(1-t_{i-1})^{-2\alpha}\Delta_{i,n}^2.
\end{align*}
As for the term $A_{n}$ above, when $\alpha<1/2 : \frac{h_{n}^{2 \alpha}}{ \Delta_{n}} A_{n}' \overset{p}{ \rightarrow} c_{1,1}^{2}H'(1)m_{K}(-2 \alpha)$. When $\alpha> 1/2$, we deduce (via the logic used to handle the term $\Theta_{1,n}$ above) that $A_{n}' \leq C \frac{ \Delta_{n}^{2-2 \alpha}}{h_{n}}$, which competes with $\Theta_{2,n}$. This bound is $A_{n}' \leq C \frac{ \Delta_{n} \log(n)}{h_{n}}$ (and is then dominated by $\Theta_{2,n}$) for $\alpha=1/2$.

We combine the various settings as follows. If $\alpha- \beta>1/2$ (implying $\alpha > 1/2$), $\hat{ \mu}_{t}^{n}$ is dominated by $A_{n}$ and has order $h_{n}^{- \alpha}$, whereas $( \hat{ \sigma}_{t}^{n})^{2}$ is at most of order $\Delta_{n}^{2-2 \alpha}/h_{n}$ (if $A_{n}'$ dominates), or $h_{n}^{-2 \beta}$ (if $\Theta_{2,n}$ dominates). In the first instance, the $t$-statistic is of order $(h_{n}/ \Delta_{n})^{1- \alpha} \rightarrow \infty$, while in the second it is $h_{n}^{1/2- \alpha+ \beta} \rightarrow \infty$. In the arbitrage-free setting $\alpha- \beta<1/2$, the leading term in $\hat{ \mu}_{t}^{n}$ is $B_{n}$. Note that $( \hat{ \sigma}_{t}^{n})^{2}$ is then dominated by $\Theta_{2,n}$, since the divergence rate of $A_{n}' \rightarrow \infty$ is always slower than $h_{n}^{-2 \beta}$, so the $t$-statistic is asymptotically normally distributed as stated in the theorem. In the borderline setting $\alpha- \beta = 1/2$, $( \hat{ \sigma}_{t}^{n})^{2}$ is still dominated by $\Theta_{2,n}$, while $\hat{ \mu}_{t}^{n}$ is influenced by both $A_{n}$ and $B_{n}$. The limit in distribution is then:
\begin{equation*}
c_{K, \beta} N(0,1) +  \frac{c_{1,1}}{c_{2,1}} \frac{m_{K}(- \alpha)}{(K_2 m_{K}(-2 \beta))^{1/2}}.
\end{equation*}
This concludes the proof. \qed
\end{proof}

\begin{proof}[Proof of Theorem \ref{theorem:gumbel-distribution}.] As in the proof of Theorem \ref{theorem:null}, we write:
\begin{equation} \label{equation:decomposition}
T_{t}^{n} = \sqrt{ \frac{h_{n}}{K_{2}}} \frac{( \hat{ \mu}_{t}^{n} - \mu_{t-}^{*})}{ \hat{ \sigma}_{t}^{n}}+ \sqrt{ \frac{h_{n}}{K_{2}}} \frac{ \mu_{t-}^{*}}{ \hat{ \sigma}_{t}^{n}} = \underbrace{ \sqrt{ \frac{h_{n}}{K_{2}}} \frac{( \hat{ \mu}_{t}^{n} - \mu_{t-}^{*})}{ \sigma_{t-}}}_{T_{t,1}^{n}} + \underbrace{ \sqrt{ \frac{h_{n}}{K_{2}}} \frac{( \hat{ \mu}_{t}^{n} - \mu_{t-}^{*})}{ \sigma_{t-}} \left(\frac{ \sigma_{t-}}{ \hat{ \sigma}_{t}^{n}}-1 \right)}_{T_{t,2}^{n}} + \underbrace{ \sqrt{ \frac{h_{n}}{K_2}} \frac{\mu_{t-}^{*}}{ \hat{ \sigma}_{t}^{n}}}_{O_{p}( \sqrt{h_{n}})}.
\end{equation}
The last term is handled as in the proof of Theorem \ref{theorem:null}, which shows that $\left|T_{t,3}^{n} \right| = O_{p}( \sqrt{h_{n}})$. This order is uniform in $t$, so that $\max_{j=1, \ldots ,m} \big|T_{t_{j}^{*},3}^{n} \big| = O_{p}( m \sqrt{h_{n}})$ by Boole's inequality. Next, we note that $T_{t,2}^{n} = T_{t,1}^{n} \left( \frac{ \sigma_{t-}}{ \hat{ \sigma}_{t}^{n}} - 1 \right) = T_{t,1}^{n} o_{p}(1)$ by Lemma \ref{lemma:spot-volatility}, so this term is negligible relative to $T_{t,1}^{n}$.

Turning to $T_{t,1}^{n}$, it follows from the proof of Theorem \ref{theorem:null}:
\begin{equation*}
T_{t,1}^{n} = \sqrt{ \frac{1}{K_{2}}} \frac{ G_{n,t} + G_{n,t}' + O_{p} \left( \frac{ \sqrt{h_{n}}}{nh_{n}} + h_{n}^{ \Gamma/2 + 1/2} + h_{n}^{B + 1/2} \right)}{ \sigma_{t-}},
\end{equation*}
where the $O_{p} \left( \frac{ \sqrt{h_{n}}}{nh_{n}} + h_{n}^{ \Gamma/2 + 1/2} + h_{n}^{B + 1/2} \right)$ term is uniform in $t$. Moreover, by \citet*[][Lemma 2.1.5]{jacod-protter:12a} $G_{n,t}'  = O_{p} \left( \sqrt{h_{n}} \right)$, uniformly in $t$.

We decompose $G_{n,t_{j}^{*}}/ \sigma_{t_{j}^{*}-}$, with $j=1, \ldots, m$, as follows:
\begin{align*}
\frac{G_{n,t_{j}^{*}}}{ \sigma_{t_{j}^{*}-}} &= \frac{1}{ \sigma_{t_{j}^{*}-} \sqrt{h_{n}}} \sum_{i=1}^{n} K \left( \frac{t_{i-1}-t_{j}^{*}}{h_{n}} \right)I_{ \left \{t_{j-1}^{*} < t_{i} \leq t_{j}^{*} \right \}} \int_{t_{i-1}}^{t_{i}} \sigma_{s} \text{d}W_{s} + \frac{1}{ \sigma_{t_{j}^{*}-} \sqrt{h_{n}}} \sum_{i=1}^{n} K \left( \frac{t_{i-1}-t_{j}^{*}}{h_{n}} \right)I_{ \left \{ t_{i} \leq t_{j-1}^{*} \right\}} \int_{t_{i-1}}^{t_{i}} \sigma_{s} \text{d}W_{s} \\[0.10cm]
&= \underbrace{ \frac{1}{ \sqrt{h_{n}}} \sum_{i=1}^{n} K \left( \frac{t_{i-1}-t_{j}^{*}}{h_{n}} \right)I_{ \left \{t_{j-1}^{*} < t_{i} \leq t_{j}^{*} \right \}} \left(W_{t_{i}} - W_{t_{i-1}} \right)}_{G_{t_{j}^{*},1}^{n}} + \underbrace{ \frac{1}{ \sigma_{t_{j}^{*}-} \sqrt{h_{n}}} \sum_{i=1}^{n} K \left( \frac{t_{i-1}-t_{j}^{*}}{h_{n}} \right)I_{ \left \{t_{j-1}^{*} < t_{i} \leq t_{j}^{*} \right \}} \int_{t_{i-1}}^{t_{i}} \left( \sigma_{s} - \sigma_{t_{j}^{*}-} \right) \text{d}W_{s}}_{G_{t_{j}^{*},2}^{n}} \\[0.10cm]
&+ \underbrace{ \frac{1}{ \sigma_{t_{j}^{*}-} \sqrt{h_{n}}} \sum_{i=1}^{n} K \left( \frac{t_{i-1}-t_{j}^{*}}{h_{n}} \right)
 I_{ \left \{ t_{i} \leq t_{j-1}^{*} \right \}} \int_{t_{i-1}}^{t_{i}} \sigma_{s} \text{d}W_{s}}_{G_{t_{j}^{*},3}^{n}}.
\end{align*}
Note $G_{t_{j_{1}}^{*},1}^{n}$ and $G_{t_{j_{2}}^{*},1}^{n}$ are independent for $j_{1} \neq j_{2}$. Moreover, $\sqrt{ \frac{1}{K_{2}}}G_{n,t_{j}^{*}}$ is, for each $j$, normally distributed with mean zero and unit asymptotic variance, so as in the proof of Lemma 1 in \citet{lee-mykland:08a}:
\begin{equation*}
a_{m} \left( \max_{j=1, \ldots,m} \sqrt{ \frac{1}{K_{2}}}G_{t^*_{j},1}^{n}-b_{m} \right)  \overset{d}{ \rightarrow} \xi.
\end{equation*}
For the second term, we write:
\begin{equation*}
\mathcal{P} \left( \max_{j=1, \ldots ,m} \big|G_{t_{j}^{*},2}^{n} \big| \geq \psi_{n,m} \right) \leq \mathcal{P} \left( \max_{j=1, \ldots ,m} \big|G_{t_{j}^{*},2}^{n} \big| \geq \psi_{n,m}, \max_{j=1, \ldots ,m} \left[G_{t_{j}^{*},2}^{n} \right] \leq \beta_{n,m} \right) + \mathcal{P} \left( \max_{j=1, \ldots ,m} \left[G_{t_{j}^{*},2}^{n} \right] > \beta_{n,m} \right),
\end{equation*}
with
\begin{align*}
\left[G_{t_{j}^{*},2}^{n} \right] &= \frac{1}{ \sigma_{t_{j}^{*}-} h_{n}^{2}} \sum_{i=1}^{n} K^{2} \left( \frac{t_{i-1}-t_{j}^{*}}{h_{n}} \right)I_{ \left \{t_{j-1}^{*} < t_{i} \leq t_{j}^{*} \right \}} \int_{t_{i-1}}^{t_{i}} \left( \sigma_{s} - \sigma_{t_{j}^{*}-} \right)^{2} \text{d}s \\[0.10cm]
&\leq Cm^{- \Gamma} \frac{1}{h_{n}} \sum_{i=1}^{n} K^{2} \left( \frac{t_{i-1}-t_{j}^{*}}{h_{n}} \right) \Delta_{i,n} \leq Cm^{- \Gamma},
\end{align*}
where $C$ is uniform over $j$.

Applying the exponential inequality from \citet*{dzhaparidze-zanten:01a}, we conclude that:
\begin{equation*}
\mathcal{P} \left( \max_{j=1, \ldots ,m} \big|G_{t_{j}^{*},2}^{n} \big| \geq C_{n}, \max_{j = 1, \ldots ,m} \left[ G_{t_{j}^{*},2}^{n} \right] \leq \beta_{n} \right) \leq 2m \exp \left(- \frac{C_{n}^{2}}{2 \beta_{n}} \right),
\end{equation*}
while Markov's inequality yields
\begin{equation*}
\mathcal{P} \left( \max_{j=1, \ldots, m} \left[G_{t_{j}^{*},2}^{n} \right] > \beta_{n,m} \right) \leq \frac{E \left( \max_{j=1, \ldots,m} \left[G_{t_{j}^{*},2}^{n} \right] \right)}{ \beta_{n,m}} \leq C \frac{m^{- \Gamma}}{ \beta_{n,m}}.
\end{equation*}
Setting $\beta_{n,m} = cm^{- \Gamma}$ and $\psi_{n,m} = cm^{- \Gamma/2} \sqrt{ \log m}$, this reduces to:
\begin{equation*}
\mathcal{P} \left( \frac{1}{ \psi_{n,m}} \max_{j=1, \ldots ,m} \big|G_{t_{j}^{*},2}^{n} \big| \geq c \right) \leq 2m^{1-c/2} + \frac{C}{c},
\end{equation*}
which shows that
\begin{equation*}
\max_{j=1, \ldots ,m} \big|G_{t_{j}^{*},2}^{n} \big| = O_{p} \left(m^{- \Gamma/2} \sqrt{ \log(m)} \right).
\end{equation*}
The fast vanishing tails of the kernel means that for $G_{t_{j}^{*},3}^{n}$:
\begin{align*}
\left[G_{t_{j}^{*},3}^{n} \right] &= \frac{1}{ \sigma_{t_{j}^{*}-}^{2} h_{n}} \sum_{i=1}^{n} K^{2} \left( \frac{t_{i-1}-t_{j}^{*}}{h_{n}} \right)
I_{ \left \{ t_{i} \leq t_{j-1}^{*} \right\}} \int_{t_{i-1}}^{t_{i}} \sigma^2_{s} \text{d}s \\[0.10cm]
\leq & C(mh_n)^{-2B},
\end{align*}
which implies, arguing as before,
\begin{equation*}
\max_{j=1, \ldots ,m} \big|G_{t_{j}^{*},3}^{n} \big| = O_{p} \left((mh_n)^{-B} \sqrt{ \log(m)} \right).
\end{equation*}
Thus, if $a_{m}\left(\frac{m}{ \sqrt{nh_{n}}}+(mh_n)^{-B} \sqrt{ \log(m)} + m^{-\Gamma/2} \sqrt{ \log(m)} \right) \rightarrow 0$,
$a_{m} \left( T_{m}^{*} - b_{m} \right) \overset{d}{ \rightarrow} \xi$. \qed
\end{proof}

\begin{proof}[Proof of Theorem \ref{theorem:fixed-jump}.] We again set $\tau_{ \text{db}} = \tau_{J} = T = 1$. Then,
\begin{align*}
T_{ \tau_{J}}^{n} = \sqrt{ \frac{h_{n}}{K_{2}}} \frac{ \hat{ \mu}_{t}^{n}}{ \hat{ \sigma}_{t}^{n}} &= \sqrt{ \frac{h_{n}}{K_{2}}} \frac{ \displaystyle \frac{1}{h_{n}} \sum_{i=1}^{n} K \left( \frac{t_{i-1}-1}{h_{n}} \right) \Delta_{i}^{n} X + \frac{1}{h_{n}} K \left( \frac{ \Delta_{n,n}}{h_{n}} \right)J}{ \displaystyle \left( \frac{1}{h_{n}} \sum_{i=1}^{n} K \left( \frac{t_{i-1}-1}{h_{n}} \right) ( \Delta_{i}^{n} X)^{2} + \frac{2}{h_{n}} K \left( \frac{ \Delta_{n,n}}{h_{n}} \right)J \Delta_{n}^{n} X + \frac{1}{h_{n}}K \left( \frac{ \Delta_{n,n}}{h_{n}} \right) J^{2} \right)^{1/2}} \\[0.25cm]
&= \frac{ \displaystyle N  \left(0, \sigma^{2}_{t-} \right) + \sqrt{ \frac{1}{K_{2}h_{n}}} K \left( \frac{ \Delta_{n,n}}{h_{n}} \right)J + o_{p}(1)}{ \displaystyle \left( \sigma_{t-}^{2} + \frac{2}{h_{n}}K \left( \frac{ \Delta_{n,n}}{h_{n}} \right) J O_{p} \left( \sqrt{ \Delta_{n}} \right) + \frac{1}{h_{n}} K \left( \frac{ \Delta_{n,n}}{h_{n}} \right) J^{2} + o_{p}(1) \right)^{1/2}} \overset{p}{ \rightarrow} \sqrt{ \frac{K(0)}{K_{2}}} \cdot \sign(J),
\end{align*}
as $n \rightarrow \infty$, $h_{n} \rightarrow 0$ and $nh_{n} \rightarrow \infty$. \qed
\end{proof}

\begin{proof}[Proof of Theorem \ref{theorem:null-noise}.]
As in the proof of Theorem \ref{theorem:null}, we write:
\begin{equation*}
\overbar{T}_{t}^{n} = \underbrace{ \sqrt{ \frac{h_{n}}{K_{2}}} \frac{( \hat{ \overbar{ \mu}}_{t}^{n} - \mu_{t-}^{*})}{ \sqrt{ \hat{ \overbar{ \sigma}}_{t}^{n}}}}_{ \overbar{T}_{1}} + \underbrace{ \sqrt{ \frac{h_{n}}{K_{2}}} \frac{\mu_{t-}^{*}}{ \sqrt{ \hat{ \overbar{ \sigma}}_{t}^{n}}}}_{ \overbar{T}_{2}},
\end{equation*}
where $\overbar{T}_{2} = O_{p}( \sqrt{h_{n}})$ is negligible. We dissect $\hat{ \overbar{ \mu}}_{t}^{n}$ into an efficient log-price and noise component:
\begin{equation*}
\hat{ \overbar{ \mu}}_{t}^{n} = \underbrace{ \frac{1}{h_{n}} \sum_{i=1}^{n-k_{n}+2} K \left( \frac{t_{i-1}-t}{h_{n}} \right) \Delta_{i-1}^{n} \overbar{X}}_{M_{X,n}} + \underbrace{ \frac{1}{h_{n}} \sum_{i=1}^{n-k_{n}+2} K \left( \frac{t_{i-1}-t}{h_{n}} \right) \Delta_{i-1}^{n} \overbar{ \epsilon}}_{M_{ \epsilon,n}}.
\end{equation*}
The strategy is again to verify Theorem 2.2.14 in \citet*{jacod-protter:12a}. This is more involved now because the summands in the drift estimator are $k_{n}$-dependent with $k_{n} \rightarrow \infty$ due to the pre-averaging. We therefore apply a block-splitting technique as in \citet*{jacod-li-mykland-podolskij-vetter:09a}.

We start with the noise term and write, for an integer $p \geq 2$,
\begin{equation*}
M_{ \epsilon,n} = M(p)_{t}^{n} + M'(p)_{t}^{n} + \hat{C}(p)_{t}^{n},
\end{equation*}
where
\begin{align*}
M(p)_{t}^{n} &= \frac{1}{h_{n}} \sum_{j=0}^{j_{n}(p)} \sum_{ \ell = \overbar{ \ell}_{j}^{n}(p)}^{ \overbar{ \ell}_{j}^{n}(p)+pk_{n}-1}K \left( \frac{t_{ \ell-1}-t}{h_{n}} \right) \Delta_{ \ell-1}^{n} \overbar{ \epsilon}, \\[0.10cm]
M'(p)_{t}^{n} &= \frac{1}{h_{n}} \sum_{j=0}^{j_{n}(p)} \sum_{ \ell = \overbar{ \ell}_{j}^{n}(p)+pk_{n}}^{ \overbar{ \ell}_{j}^{n}(p)+pk_{n}+k_{n}-1}K \left( \frac{t_{ \ell-1}-t}{h_{n}} \right) \Delta_{ \ell-1}^{n} \overbar{ \epsilon}, \\[0.10cm]
\hat{C}(p)_{t}^{n} &= \frac{1}{h_{n}} \sum_{ \ell = \overbar{ \ell}_{j_{n}(p)+1}^{n}(p)}^{n-k_{n}+2}K \left( \frac{t_{ \ell-1}-t}{h_{n}} \right) \Delta_{ \ell-1}^{n} \overbar{ \epsilon},
\end{align*}
with $j_{n}(p)= \left[ \frac{(n+1)}{(p+1)k_{n}} \right]-1$ and $\overbar{ \ell}_{j}^{n}(p) = j(pk_{n}-1)+jk_{n}+1$. The term $M(p)_{t}^{n}$ is a sum of ``big'' blocks of dimension $pk_{n}$, while the term $M'(p)_{t}^{n}$ is a sum of ``small'' blocks of dimension $k_{n}$, which separate the big blocks. $\hat{C}(p)_{t}^{n}$ is an end effect. $M(p)_{t}^{n}$ is shown to be the dominating term. It can be written as:
\begin{equation*}
M(p)_{t}^{n} \equiv \sum_{j=0}^{j_{n}(p)} u_{j}^{n},
\end{equation*}
where $u_{j}^{n} = \frac{1}{h_{n}} \sum_{ \ell = \overbar{ \ell}_{j}^{n}(p)}^{ \overbar{ \ell}_{j}^{n}(p)+pk_{n}-1}K \left( \frac{t_{ \ell-1}-t}{h_{n}} \right) \Delta_{ \ell-1}^{n} \overbar{ \epsilon}$ is by construction independent on $u_{j'}$ when $j' \neq j$ and $k_{n} > Q+1$. We can then employ Theorem 2.2.14 in \citet*{jacod-protter:12a} by taking expectations conditional on the discrete-time filtration $\mathcal{G}(p)_{j}^{n} = \mathcal{F}_{t_{ \overbar{ \ell}_{j}^{n}(p)-1}}$. We immediately get the orthogonality condition \eqref{g4} and
\begin{equation*}
\sum_{j=0}^{j_{n}(p)} E \left[u_{j}^{n} \mid \mathcal{G}(p)_{j}^{n} \right] = 0.
\end{equation*}
As for the conditional variance:
\begin{align*}
\sum_{j=0}^{j_{n}(p)}E \left[(u_{j}^{n})^{2} \mid \mathcal{G}(p)_{j}^{n} \right]  &= \frac{1}{h_{n}^{2}} \sum_{j=0}^{j_{n}(p)} E \left[ \left(\sum_{ \ell = \overbar{ \ell}_{j}^{n}(p)}^{ \overbar{ \ell}_{j}^{n}(p)+pk_{n}-1}K \left( \frac{t_{ \ell-1}-t}{h_{n}} \right) \Delta_{ \ell-1}^{n} \overbar{ \epsilon} \right)^{2} \mid \mathcal{G}(p)_{j}^{n} \right] \\[0.10cm]
&= \frac{1}{h_{n}^{2}} \sum_{j=0}^{j_{n}(p)}\sum_{ \ell = \overbar{ \ell}_{j}^{n}(p)}^{ \overbar{ \ell}_{j}^{n}(p)+pk_{n}-1}K^{2} \left( \frac{t_{ \ell-1}-t}{h_{n}} \right)E \left[ ( \Delta_{ \ell-1}^{n} \overbar{ \epsilon})^{2} \mid \mathcal{G}(p)_{j}^{n} \right] \\[0.10cm]
&+ \frac{2}{h_{n}^{2}} \sum_{j=0}^{j_{n}(p)} \sum_{ \ell = \overbar{ \ell}_{j}^{n}(p)}^{ \overbar{ \ell}_{j}^{n}(p)+pk_{n}-1} \sum_{ \ell'= \ell+1}^{\overbar{ \ell}_{j}^{n}(p)+pk_{n}-1 }K \left( \frac{t_{ \ell-1}-t}{h_{n}} \right)K \left( \frac{t_{ \ell'-1}-t}{h_{n}} \right)E \left[ \Delta_{ \ell-1}^{n} \overbar{ \epsilon} \Delta_{ \ell'-1}^{n} \overbar{ \epsilon} \mid \mathcal{G}(p)_{j}^{n} \right].
\end{align*}
By the mean value theorem,
\begin{equation*}
K \left( \frac{t_{ \ell'-1}-t}{h_{n}} \right)= K \left( \frac{t_{ \ell-1}-t}{h_{n}} \right) + K' \left( \frac{ \xi_{t_{ \ell-1},t_{ \ell'-1}}}{h_{n}} \right) \frac{t_{ \ell'-1}-t_{ \ell-1}}{h_{n}},
\end{equation*}
for a $\xi_{t_{ \ell-1},t_{ \ell'-1}} \in ]t_{ \ell'-1}-t,t_{ \ell-1}-t[$, so that
\begin{align*}
\sum_{j=0}^{j_{n}(p)}E \left[(u_{j}^{n})^{2} \mid \mathcal{G}(p)_{j}^{n} \right] &= \underbrace{ \frac{1}{h_{n}^{2}} \sum_{j=0}^{j_{n}(p)}  \sum_{ \ell = \overbar{ \ell}_{j}^{n}(p)}^{ \overbar{ \ell}_{j}^{n}(p)+pk_{n}-1}K^{2} \left( \frac{t_{ \ell-1}-t}{h_{n}} \right) \left(E \left[( \Delta_{ \ell-1}^{n} \overbar{ \epsilon})^{2} \mid \mathcal{G}(p)_{j}^{n} \right] + 2 \sum_{ \ell' = \ell+1}^{ \overbar{ \ell}_{j}^{n}(p)+pk_n-1}E \left[ \Delta_{ \ell-1}^{n} \overbar{ \epsilon} \Delta_{ \ell'-1}^{n} \overbar{ \epsilon} \mid \mathcal{G}(p)_{j}^{n} \right] \right)}_{V_{1,n}} \\[0.10cm]
&+ \underbrace{ \frac{2}{h_{n}^{2}} \sum_{j=0}^{j_{n}(p)} \sum_{ \ell = \overbar{ \ell}_{j}^{n}(p)}^{ \overbar{ \ell}_{j}^{n}(p)+pk_{n}-1} \sum_{ \ell'= \ell+1}^{\overbar{ \ell}_{j}^{n}(p)+pk_{n}-1 }K \left(\frac{t_{ \ell-1}-t}{h_{n}} \right)K' \left( \frac{ \xi_{t_{ \ell-1},t_{ \ell'-1}}}{h_{n}} \right) \frac{t_{ \ell'-1}-t_{ \ell-1}}{h_{n}}E \left[ \Delta_{ \ell-1}^{n} \overbar{ \epsilon} \Delta_{ \ell'-1}^{n} \overbar{ \epsilon} \mid \mathcal{G}(p)_{j}^{n} \right]}_{V_{2,n}}.
\end{align*}
The term inside the first summation is the long-run variance of a stationary time series that is over-differenced, so $V_{1,n} = 0$ \citep*[see e.g.][pp. 305--306]{hassler:16a}. The conditional variance is thus dominated by the limit of $V_{2,n}$.

We use the multinomial theorem for the fourth conditional moment:
\begin{align*}
&\frac{h_{n}^{4}}{k_{n}^{2}} \sum_{j=0}^{j_{n}(p)}E \left[(u_{j}^{n})^{4} \mid \mathcal{G}(p)_{j}^{n} \right] = \frac{1}{k_{n}^{2}} \sum_{j=0}^{j_{n}(p)}E \left[ \left( \sum_{ \ell = \overbar{ \ell}_{j}^{n}(p)}^{ \overbar{ \ell}_{j}^{n}(p)+pk_{n}-1}K \left( \frac{t_{ \ell-1}-t}{h_{n}} \right) \Delta_{ \ell-1}^{n} \overbar{ \epsilon} \right)^{4} \mid \mathcal{G}(p)_{j}^{n} \right] \\[0.10cm]
&= \frac{1}{k_{n}^{2}} \sum_{j=0}^{j_{n}(p)} \Bigg(\sum_{ \ell = \overbar{ \ell}_{j}^{n}(p)}^{ \overbar{ \ell}_{j}^{n}(p)+pk_{n}-1}K^{4} \left( \frac{t_{ \ell-1}-t}{h_{n}} \right)E \left[ \left( \Delta_{ \ell-1}^{n} \overbar{ \epsilon} \right)^{4} \right] \\[0.10cm]
&+ 4 \sum_{ \ell_{1} \neq \ell_{2}}K^{3} \left( \frac{t_{ \ell_1-1}-t}{h_{n}} \right)K \left( \frac{t_{ \ell_2-1}-t}{h_{n}} \right)E \left[ \left( \Delta_{ \ell_{1}-1}^{n} \overbar{ \epsilon} \right)^{3} \Delta_{ \ell_{2}-1}^{n} \overbar{ \epsilon} \right] \\[0.10cm]
&+ 6 \sum_{ \ell_{1} \neq \ell_{2}}K^{2} \left( \frac{t_{ \ell_{1}-1}-t}{h_{n}} \right)K^{2} \left( \frac{t_{ \ell_{2}-1}-t}{h_{n}} \right)E \left[ \left( \Delta_{ \ell_{1}-1}^{n} \overbar{ \epsilon} \right)^{2} \left( \Delta_{ \ell_{2}-1}^{n} \overbar{ \epsilon} \right)^{2} \right] \\[0.10cm]
&+ 12 \sum_{ \ell_{1} \neq \ell_{2} \neq \ell_{3}}K^{2} \left( \frac{t_{ \ell_1-1}-t}{h_{n}} \right)K \left( \frac{t_{ \ell_{2}-1}-t}{h_{n}} \right)K \left( \frac{t_{ \ell_{3}-1}-t}{h_{n}} \right)E \left[ \left( \Delta_{ \ell_{1}-1}^{n} \overbar{ \epsilon} \right)^{2} \Delta_{ \ell_{2}-1}^{n} \overbar{ \epsilon} \Delta_{ \ell_{3}-1}^{n} \overbar{ \epsilon} \right] \\[0.10cm]
& +24 \sum_{ \ell_{1} \neq \ell_{2} \neq \ell_{3} \neq \ell_{4}}K \left( \frac{t_{ \ell_{1}-1}-t}{h_{n}} \right)K \left( \frac{t_{ \ell_{2}-1}-t}{h_{n}} \right)K \left( \frac{t_{ \ell_{3}-1}-t}{h_{n}} \right)K \left( \frac{t_{ \ell_{4}-1}-t}{h_{n}} \right) E \left[ \Delta_{ \ell_{1}-1}^{n} \overbar{ \epsilon}  \Delta_{ \ell_{2}-1}^{n} \overbar{ \epsilon} \Delta_{ \ell_{3}-1}^{n} \overbar{ \epsilon} \Delta_{ \ell_{4}-1}^{n} \overbar{ \epsilon} \right] \Bigg).
\end{align*}
Now, since $\ell_{1}, \ell_{2}, \ell_{3}, \ell_{4}$ are no more than $O(pk_{n} \Delta_{n}/h_{n})$ terms apart, which is going to zero in the limit, we can mean value expand the kernel again to show that
\begin{equation*}
\frac{h_{n}^{4}}{k_{n}^{2}} \sum_{j=0}^{j_{n}(p)}E \left[(u_{j}^{n})^{4} \mid \mathcal{G}(p)_{j}^{n} \right] = Q_{1,n} + Q_{2,n},
\end{equation*}
where, as for the variance term,
\begin{align*}
Q_{1,n} &= \frac{1}{k_{n}^{2}} \sum_{j=0}^{j_{n}(p)} \Bigg( \sum_{ \ell = \overbar{ \ell}_{j}^{n}(p)}^{ \overbar{ \ell}_{j}^{n}(p)+pk_{n}-1}K^{4} \left( \frac{t_{ \ell-1}-t}{h_{n}} \right) \bigg( E \left[ \left( \Delta_{ \ell-1}^{n} \overbar{ \epsilon} \right)^{4} \right] +  4 \sum_{ \ell_{1} \neq \ell_{2}}E \left[ \left( \Delta_{ \ell_{1}-1}^{n} \overbar{ \epsilon} \right)^{3} \Delta_{ \ell_{2}-1}^{n} \overbar{ \epsilon} \right] + 6 \sum_{ \ell_{1} \neq \ell_{2}}E \left[ \left( \Delta_{ \ell_{1}-1}^{n} \overbar{ \epsilon} \right)^{2} \left( \Delta_{ \ell_{2}-1}^{n} \overbar{ \epsilon} \right)^{2} \right] \\[0.10cm]
&+ 12 \sum_{ \ell_{1} \neq \ell_{2} \neq \ell_{3}}E \left[ \left( \Delta_{ \ell_{1}-1}^{n} \overbar{ \epsilon} \right)^{2} \Delta_{ \ell_{2}-1}^{n} \overbar{ \epsilon} \Delta_{ \ell_{3}-1}^{n} \overbar{ \epsilon} \right] + 24 \sum_{ \ell_{1} \neq \ell_{2} \neq \ell_{3} \neq \ell_{4}}E \left[ \Delta_{ \ell_{1}-1}^{n} \overbar{ \epsilon} \Delta_{ \ell_{2}-1}^{n} \overbar{ \epsilon} \Delta_{ \ell_{3}-1}^{n} \overbar{ \epsilon} \Delta_{ \ell_{4}-1}^{n} \overbar{ \epsilon} \right] \bigg) \Bigg) = 0.
\end{align*}
With $\ell_{1}+ \ell_{2}+ \ell_{3}+ \ell_{4}=4$, the boundedness of the fourth moment of the noise and the first equation of (5.36) in \citet{jacod-li-mykland-podolskij-vetter:09a}, showing each $\Delta_{ \ell-1}^{n} \overbar{ \epsilon}$ is of order $k_n^{-1/2}$, it follows that
\begin{equation*}
E \bigg[ \left( \Delta_{ \ell_{1}-1}^{n} \overbar{ \epsilon} \right)^{ \ell_{1}} \left( \Delta_{ \ell_{2}-1}^{n} \overbar{ \epsilon} \right)^{ \ell_{2}} \left( \Delta_{ \ell_{3}-1}^{n} \overbar{ \epsilon} \right)^{ \ell_{3}} \left( \Delta_{ \ell_{4}-1}^{n} \overbar{ \epsilon} \right)^{ \ell_{4}} \bigg] \leq C k_{n}^{-2}.
\end{equation*}
We deduce that:
\begin{align*}
|Q_{2,n}| &\leq C \frac{pk_{n} \Delta_{n}}{h_{n}} \frac{1}{k_{n}^{4}} \sum_{j=0}^{j_{n}(p)} \Bigg( \sum_{ \ell_{1} \neq \ell_{2}}K^{3} \left( \frac{t_{ \ell_{1}-1}-t}{h_{n}} \right)K' \left( \frac{t_{ \ell_{2}-1}-t}{h_{n}} \right) + \sum_{ \ell_{1} \neq \ell_{2}}K^{2} \left( \frac{t_{ \ell_{1}-1}-t}{h_{n}} \right)(K^{2})' \left( \frac{t_{ \ell_{2}-1}-t}{h_{n}} \right) \\[0.10cm]
&+ \sum_{ \ell_{1} \neq \ell_{2} \neq \ell_{3}}K^{2} \left( \frac{t_{ \ell_{1}-1}-t}{h_{n}} \right) \left(K' \left( \frac{t_{ \ell_{2}-1}-t}{h_{n}} \right)K \left( \frac{t_{ \ell_{3}-1}-t}{h_{n}} \right)+K \left( \frac{t_{ \ell_{2}-1}-t}{h_{n}} \right)K' \left( \frac{t_{ \ell_{3}-1}-t}{h_{n}} \right) + O \left( \frac{pk_{n} \Delta_{n}}{h_{n}} \right) \right) \\[0.10cm]
&+ \sum_{ \ell_{1} \neq \ell_{2} \neq \ell_{3} \neq \ell_{4}}K \left( \frac{t_{ \ell_{1}-1}-t}{h_{n}} \right) \bigg(K' \left( \frac{t_{ \ell_{2}-1}-t}{h_{n}} \right)K \left( \frac{t_{ \ell_{3}-1}-t}{h_{n}} \right)K \left( \frac{t_{ \ell_{4}-1}-t}{h_{n}} \right)+K \left( \frac{t_{ \ell_{2}-1}-t}{h_{n}} \right)K' \left( \frac{t_{ \ell_{3}-1}-t}{h_{n}} \right)K \left( \frac{t_{ \ell_{4}-1}-t}{h_{n}} \right) \\[0.10cm]
&+K \left( \frac{t_{ \ell_{2}-1}-t}{h_{n}} \right)K \left( \frac{t_{ \ell_{3}-1}-t}{h_{n}} \right)K' \left( \frac{t_{ \ell_{4}-1}-t}{h_{n}} \right) + O \left( \frac{pk_{n} \Delta_{n}}{h_{n}} \right) \bigg) \Bigg) = O_{p} \left( \frac{1}{k^3_{n}} \right).
\end{align*}
$M'(p)_{t}^{n}$ can be bounded by Doob's inequality:
\begin{align*}
E \left[ \sup_{s \leq t} \left|M'(p)_{s}^{n} \right|^{2} \right] &\leq 4 \frac{1}{h_{n}^{2}} \sum_{j=0}^{j_{n}(p)}E \left[ \left( \sum_{ \ell= \overbar{ \ell}_{j}^{n}(p)+pk_{n}}^{ \overbar{ \ell}_{j}^{n}(p)+pk_{n}+k_{n}-1}K \left( \frac{t_{ \ell-1}-t}{h_{n}} \right) \Delta_{ \ell-1}^{n} \overbar{ \epsilon} \right)^{2} \right] \\[0.10cm]
&= \frac{1}{h_{n}^{2}} \sum_{j=0}^{j_{n}(p)} \sum_{ \ell} \sum_{ \ell' = \ell+1}K \left( \frac{t_{ \ell-1}-t}{h_{n}} \right)K' \left( \frac{ \xi_{t_{ \ell-1},t_{ \ell'-1}}}{h_{n}} \right) \frac{t_{ \ell'-1}-t_{ \ell-1}}{h_{n}}E \left[ \Delta_{ \ell-1}^{n} \overbar{ \epsilon} \cdot \Delta_{ \ell'-1}^{n} \overbar{ \epsilon} \right]
\end{align*}
which is negligible compared to $V_{2,n}$, as it asymptotically is an identical term summed over a small block.

Finally, the end-effect term can also be neglected, as
\begin{equation*}
\bigg| \frac{1}{h_{n}} \sum_{ \ell= \overbar{ \ell}_{j_{n}(p)+1}^{n}(p)}^{n-k_{n}+2}K \left( \frac{t_{ \ell-1}-t}{h_{n}} \right) \Delta_{ \ell-1}^{n} \overbar{ \epsilon} \bigg| \leq \frac{1}{h_{n}} \sum_{ \ell= \overbar{ \ell}_{j_{n}(p)+1}^{n}(p)}^{n-k_{n}+2}K \left( \frac{t_{ \ell-1}-t}{h_{n}} \right) \left| \Delta_{ \ell-1}^{n} \overbar{ \epsilon} \right| \leq C \frac{k_n^{1/2}}{h_{n}},
\end{equation*}
since $\Delta_{ \ell-1}^{n} \overbar{ \epsilon} = O_{p} \left( k_{n}^{-1/2} \right)$.

We next analyze the $M_{X,n}$ term and show this is negligible too. We write $M_{X,n} = \tilde{M}(p)_{t}^{n} + \tilde{M}'(p)_{t}^{n} + \widehat{ \widehat{C}}(p)_{t}^{n}$ with an identical decomposition as for the $M_{ \epsilon,n}$ term. Arguing as above, the dominating term is $\tilde{M}(p)_{t}^{n}$. We decompose $\tilde{M}(p)_{t}^{n} \equiv \sum_{j=0}^{j_{n}(p)} \tilde{u}_{j}^{n}$, where $\tilde{u}_{j}^{n} = \frac{1}{h_{n}} \sum_{ \ell= \overbar{ \ell}_{j}^{n}(p)}^{ \overbar{ \ell}_{j}^{n}(p)+pk_{n}-1}K \left( \frac{t_{ \ell-1}-t}{h_{n}} \right) \Delta_{ \ell-1}^{n} \overbar{X}$ and then compute:
\begin{align*}
\sum_{j=0}^{j_{n}(p)}E \left[(\tilde{u}_{j}^{n})^{2} \mid \mathcal{G}(p)_{j}^{n} \right] &= \frac{1}{h_{n}^{2}} \sum_{j=0}^{j_{n}(p)} E \left[ \left( \sum_{ \ell = \overbar{ \ell}_{j}^{n}(p)}^{ \overbar{ \ell}_{j}^{n}(p)+pk_{n}-1}K \left( \frac{t_{ \ell-1}-t}{h_{n}} \right) \Delta_{ \ell-1}^{n} \overbar{X} \right)^{2} \mid \mathcal{G}(p)_{j}^{n} \right] \\[0.10cm]
&= \frac{1}{h_{n}^{2}} \sum_{j=0}^{j_{n}(p)} \sum_{ \ell = \overbar{ \ell}_{j}^{n}(p)}^{ \overbar{ \ell}_{j}^{n}(p)+pk_{n}-1}K^{2} \left( \frac{t_{ \ell-1}-t}{h_{n}} \right)E \left[( \Delta_{ \ell-1}^{n} \overbar{X})^{2} \mid \mathcal{G}(p)_{j}^{n} \right] \\[0.10cm]
&+ \frac{2}{h_{n}^{2}} \sum_{j=0}^{j_{n}(p)} \sum_{ \ell = \overbar{ \ell}_{j}^{n}(p)}^{ \overbar{ \ell}_{j}^{n}(p)+pk_{n}-1} \sum_{ \ell'= \ell+1}^{\overbar{ \ell}_{j}^{n}(p)+pk_{n}-1 }K \left( \frac{t_{ \ell-1}-t}{h_{n}} \right)K \left( \frac{t_{ \ell'-1}-t}{h_{n}} \right)E \left[ \Delta_{ \ell-1}^{n} \overbar{X} \Delta_{ \ell'-1}^{n} \overbar{X} \mid \mathcal{G}(p)_{j}^{n} \right] \\[0.10cm]
&= \frac{2}{h_{n}^{2}} \sum_{j=0}^{j_{n}(p)} \sum_{ \ell = \overbar{ \ell}_{j}^{n}(p)}^{ \overbar{ \ell}_{j}^{n}(p)+pk_{n}-1} \sum_{ \ell'= \ell+1}^{\overbar{ \ell}_{j}^{n}(p)+pk_{n}-1 }K \left( \frac{t_{ \ell-1}-t}{h_{n}} \right)K' \left( \frac{ \xi_{t_{ \ell-1},t_{ \ell'-1}}}{h_{n}} \right) \frac{t_{ \ell'-1}-t_{ \ell-1}}{h_{n}}E \left[ \Delta_{ \ell-1}^{n} \overbar{X} \Delta_{ \ell'-1}^{n} \overbar{X} \mid \mathcal{G}(p)_{j}^{n} \right],
\end{align*}
where
\begin{equation*}
E \left[ \Delta_{ \ell-1}^{n} \overbar{X} \Delta_{ \ell'-1}^{n} \overbar{X} \mid \mathcal{G}(p)_{j}^{n} \right] =E \left[ \left( \sum_{j=0}^{k_{n}-1}H_{j}^{n} \Delta_{i}^{n} X_{ \ell+j-1} \right) \left( \sum_{j=0}^{k_{n}-1}H_{j}^{n} \Delta_{i}^{n} X_{ \ell'+j-1} \right) \mid \mathcal{G}(p)_{j}^{n} \right] = \sum_{j=0}^{k_{n}- \ell'+ \ell-1}H_{j}^{n}H_{j+ \ell'- \ell}^{n} \int_{t_{j-1}}^{t_{j}} \sigma_{s}^{2} \text{d}s.
\end{equation*}
We conclude that $\tilde{M}(p)_{t}^{n} = O_{p} \left( \frac{ \sqrt{\Delta_{n} k_{n}}}{h_{n}} \right)$ is of smaller order than $M_{ \epsilon,n}$.

Next, we investigate the estimator of the long-run variance of $\hat{ \overbar{ \mu}}_{t}^{n}$, i.e. $\hat{ \overbar{ \sigma}}_{t}^{n}$. We only analyze the market microstructure component, which is the leading term:
\begin{align*}
E \left[ \hat{ \overbar{ \sigma}}_{t}^{n} \right] = & \frac{1}{h_{n}} \sum_{i=1}^{n-k_{n}+2}K^{2} \left( \frac{t_{i-1}-t}{h_{n}} \right)E \left[ \left( \Delta_{i-1}^{n} \overbar{Y} \right)^{2} \right] + \frac{2}{h_{n}} \sum_{L=1}^{L_{n}}w \left( \frac{L}{L_{n}} \right) \sum_{i=1}^{n-k_{n}-L+2}K \left( \frac{t_{i-1}-t}{h_{n}} \right)K \left( \frac{t_{i+L-1}-t}{h_{n}} \right)E \left[ \Delta_{i-1}^{n} \overbar{Y} \Delta_{i-1+L}^{n} \overbar{Y} \right] \\[0.10cm]
= &  \frac{1}{h_{n}} \sum_{i=1}^{n-k_{n}+2}K^{2} \left( \frac{t_{i-1}-t}{h_{n}} \right)E \left[ \left( \Delta_{i-1}^{n} \overbar{ \epsilon} \right)^{2} \right] \\[0.10cm]
& + \frac{2}{h_{n}} \sum_{L=1}^{L_{n}}w \left( \frac{L}{L_{n}} \right) \sum_{i=1}^{n-k_{n}-L+2}K \left( \frac{t_{i-1}-t}{h_{n}} \right)K \left( \frac{t_{i+L-1}-t}{h_{n}} \right)E \left[ \Delta_{i-1}^{n} \overbar{ \epsilon} \Delta_{i-1+L}^{n} \overbar{ \epsilon} \right] + R_{n}'' \\[0.10cm]
& \equiv \sigma_{0,n}+ R_{n}'',
\end{align*}
where $R_{n}''$ is asymptotically negligible, following the calculation of the limiting variance of the drift estimator. Assuming $L_{n} \Delta_{n}/h_{n} \rightarrow 0$, we write $K \left( \frac{t_{i+L-1}-t}{h_{n}} \right) = K \left( \frac{t_{i-1}-t}{h_{n}} \right) + K' \left( \frac{ \xi_{i-1,i+L-1}}{h_{n}} \right) \frac{L \Delta_{n}}{h_{n}}$ and decompose $\sigma_{0,n}$ into:
\begin{align*}
\sigma_{0,n} &= \underbrace{ \frac{1}{h_{n}} \sum_{i=1}^{n-k_{n}+2}K^{2} \left( \frac{t_{i-1}-t}{h_{n}} \right)E \left[ \left( \Delta_{i-1}^{n} \overbar{ \epsilon} \right)^{2} \right] + \frac{2}{h_{n}} \sum_{L=1}^{L_{n}} \sum_{i=1}^{n-k_{n}-L+2}K^{2} \left( \frac{t_{i-1}-t}{h_{n}} \right)E \left[ \Delta_{i-1}^{n} \overbar{ \epsilon} \Delta_{i-1+L}^{n} \overbar{ \epsilon} \right]}_{ \sigma_{1,n}} \\[0.10cm]
&+ \underbrace{ \frac{2}{h_{n}} \sum_{L=1}^{L_{n}}w \left( \frac{L}{L_{n}} \right)\sum_{i=1}^{n-k_{n}-L+2}K \left( \frac{t_{i-1}-t}{h_{n}} \right)K' \left( \frac{ \xi_{i-1,i+L-1}}{h_{n}} \right)L \frac{ \Delta_{n}}{h_{n}}E \left[ \Delta_{i-1}^{n} \overbar{ \epsilon} \Delta_{i-1+L}^{n} \overbar{ \epsilon} \right]}_{ \sigma_{2,n}},
\end{align*}
where $\sigma_{1,n} = 0$. Since $w(L/L_{n}) = 1 + O \left( \frac{L}{L_{n}} \right)$ after the proper rescaling the limit of $\sigma_{2,n}$ and $V_{2,n}$ are identical. \qed
\end{proof}

\begin{proof}[Proof of Theorem \ref{theorem:alternative-noise}.]
As in the proof of Theorem \ref{theorem:alternative}, we set $\tau_{ \text{db}} = 1$ and further $c_{1,t} = c_{2,t} = 1$. We also write $\widetilde{X}_{t} = X_{t} + D_{t} + V_{t}$. Based on Theorem \ref{theorem:null-noise}, $\frac{1}{h_{n}} \sum_{i=1}^{n-k_{n}+1}K \left( \frac{t_{i-1}-t}{h_{n}} \right) \left( \Delta_{i-1}^{n} \overbar{X} + \Delta_{i-1}^{n} \overbar{ \epsilon} \right) = O_{p} \left( \frac{ \sqrt{k_{n}}}{h_{n}} \right)$.
Next, we note that for suitable $\xi_{j} \in[t_{j},t_{j+1}]$ and invoking again the mean-value theorem and ignoring end-effects:
\begin{align*}
\frac{1}{h_{n}} \sum_{i=1}^{n-k_{n}+2}K \left( \frac{t_{i-1}-1}{h_{n}} \right) \Delta_{i-1}^{n} \overbar{D} &= \frac{1}{h_{n}} \sum_{i=1}^{n-k_{n}+2}K \left( \frac{t_{i-1}-1}{h_{n}} \right) \sum_{j=1}^{k_{n}-1}g_{j}^{n}(D_{t_{i+j-1}}-D_{t_{i+j-2}}) \\[0.10cm]
&= \frac{1}{h_{n}} \sum_{i=1}^{n-k_{n}+2}K \left( \frac{t_{i-1}-1}{h_{n}} \right) \sum_{j=1}^{k_{n}-1}g_{j}^{n} \frac{(1-t_{i+j})^{1- \alpha}-(1-t_{i+j-1})^{1- \alpha}}{1- \alpha} \\[0.10cm]
&= \frac{1}{h_{n}} \sum_{i=1}^{n-k_{n}+2}K \left( \frac{t_{i-1}-1}{h_{n}} \right) \sum_{j=1}^{k_{n}-1}g_{j}^{n} (1- \xi_{i+j})^{- \alpha} \Delta_{i+j,n} \\[0.10cm]
&= \frac{1}{h_{n}} \sum_{i=1}^{n-k_{n}+2}K \left( \frac{t_{i-1}-1}{h_{n}} \right) \left(1-t_{i-1} + O(k_{n} \Delta_{n}) \right)^{- \alpha} \Delta_{i,n} \sum_{j=1}^{k_{n}-1}g_{j}^{n}.
\end{align*}
As $\frac{1}{k_{n}} \sum_{j=1}^{k_{n}-1}g_{j}^{n} \rightarrow \psi_{0} = \int_{0}^{1}g(s) \text{d}s$, the above has asymptotic rate $k_{n}h_{n}^{- \alpha}$ (as for the term $A_{n}$ in the proof of Theorem \ref{theorem:alternative}) if $k_{n} \Delta_{n} \rightarrow 0$, which is assured by assumption. This dominates the noise term if $\frac{k_{n}h_{n}^{- \alpha}}{ \frac{ \sqrt{k_{n}}}{h_{n}}} \rightarrow \infty$, i.e. $\sqrt{k_{n}}h_{n}^{1- \alpha} \rightarrow \infty$.

The pre-averaged volatility burst term can be written as
\begin{align*}
\frac{1}{h_{n}} \sum_{i=1}^{n-k_{n}+2}K \left( \frac{t_{i-1}-1}{h_{n}} \right) \Delta_{i-1}^{n} \overbar{V} &= \frac{1}{h_{n}} \sum_{i=1}^{n-k_{n}+2}K \left( \frac{t_{i-1}-1}{h_{n}} \right) \sum_{j=1}^{k_{n}-1}g_{j}^{n} \int_{t_{i+j-1}}^{t_{i+j}}(1-s)^{- \beta} \text{d}W_{s} \\[0.10cm]
&= \frac{1}{h_{n}} \sum_{i=1}^{n} \int_{t_{i-1}}^{t_{i}}(1-s)^{- \beta} \text{d}W_{s} \sum_{j=1}^{k_{n}-1 \wedge i-1}g_{j}^{n}K \left( \frac{t_{i-j-1}-1}{h_{n}} \right),
\end{align*}
whose variance is
\begin{align*}
&\frac{1}{h_{n}^{2}} \sum_{i=1}^{n} \int_{t_{i-1}}^{t_{i}}(1-s)^{-2 \beta} \text{d}s \left( \sum_{j=1}^{k_{n}-1 \wedge i-1}g_{j}^{n}K \left( \frac{t_{i-j-1}-1}{h_{n}} \right) \right)^{2} \\[0.10cm]
&= \frac{1}{h_{n}^{2}} \sum_{i=1}^{n} \int_{t_{i-1}}^{t_{i}}(1-s)^{-2 \beta} \text{d}s K^{2} \left( \frac{t_{i-1}-1}{h_{n}} + O \left( \frac{ \Delta_{n}k_{n}}{h_{n}} \right) \right) \Delta_{i,n} \left( \sum_{j=1}^{k_{n}-1 \wedge i-1}g_{j}^{n} \right)^{2} = O_{p} \left(k_{n}^{2}h_{n}^{-2 \beta-1} \right),
\end{align*}
so
\begin{equation*}
\frac{1}{h_{n}} \sum_{i=1}^{n-k_{n}+2}K \left( \frac{t_{i-1}-1}{h_{n}} \right) \Delta_{i-1}^{n} \overbar{V} = O_{p} \left( k_{n}h_{n}^{- \beta-1/2} \right),
\end{equation*}
which is negligible with respect to the previous term, since $\alpha- \beta>1/2$. Thus, $\hat{ \overbar{ \mu}}_{t}^{n}$ diverges as $k_{n}h_{n}^{- \alpha}$ when $\sqrt{k_{n}}h_{n}^{1- \alpha} \rightarrow \infty$. The leading orders in the denominator are:
\begin{align*}
& \frac{1}{h_{n}} \sum_{i=1}^{n-k_{n}+2}K^{2} \left( \frac{t_{i-1}-t}{h_{n}} \right) \left( \Delta_{i-1}^{n} \overbar{D} \right)^{2} + \frac{2}{h_{n}} \sum_{L=1}^{L_{n}}w \left( \frac{L}{L_{n}} \right) \sum_{i=1}^{n-k_{n}-L+2}K \left( \frac{t_{i-1}-t}{h_{n}} \right)K \left( \frac{t_{i+L-1}-t}{h_{n}} \right) \Delta_{i-1}^{n} \overbar{D} \Delta_{i+L-1}^{n} \overbar{D} \\[0.10cm]
\overset{p}{ \sim} & \frac{1}{h_{n}} \sum_{i=1}^{n-k_{n}+2}K^{2} \left( \frac{t_{i-1}-t}{h_{n}} \right) \left( \sum_{j=1}^{k_{n}-1}g_{j}^{n} \frac{ \left(1-t_{i-1} + O(k_{n} \Delta_{n}) \right)^{- \alpha}}{1- \alpha} \right)^{2} = O_{p} \left(k_{n}^{2} \frac{ \Delta_{n}^{2(1- \alpha)}}{h_{n}} \right),
\end{align*}
and
\begin{align*}
\frac{1}{h_{n}} \sum_{i=1}^{n-k_{n}+2}K^{2} \left( \frac{t_{i-1}-t}{h_{n}} \right) \left( \Delta_{i-1}^{n} \overbar{V} \right)^{2} + \frac{2}{h_{n}} \sum_{L=1}^{L_{n}}w \left( \frac{L}{L_{n}} \right) \sum_{i=1}^{n-k_{n}-L+2}K \left( \frac{t_{i-1}-t}{h_{n}} \right)K \left( \frac{t_{i+L-1}-t}{h_{n}} \right) \Delta_{i-1}^{n} \overbar{V} \Delta_{i+L-1}^{n} \overbar{V} =
O_{p} \left(k_{n}^{2}h_{n}^{-2 \beta} \right).
\end{align*}
Rationalizing as in the proof of Theorem \ref{theorem:alternative}, we arrive at the claim of this theorem. \qed
\end{proof}

\section{Critical value of drift burst $t$-statistic} \label{appendix:critical-value}

In the paper, we show that $(T_{t_{i}^{*}}^{n})_{i=1}^{m}$ is asymptotically a sequence of independent standard normal random variables, if $m$ does not increase too fast. A standard extreme value theory can then be applied. In practice, however, the route we follow with frequent sampling of our $t$-statistic leads $(T_{t_{i}^{*}}^{n})_{i=1}^{m}$ to be constructed from overlapping data. The extent of this depends on the interplay between the sampling frequency $n$, the grid points $\left( t_{i}^{*} \right)_{i=1}^{m}$, the kernel $K$, and the bandwidth $h_{n}$. In our implementation $(T_{t_{i}^{*}}^{n})_{i=1}^{m}$ exhibits a strong serial correlation, so that $m$ severely overstates the effective number of ``independent'' copies in a given sample. This implies our test is too conservative when evaluated against the Gumbel distribution.

With the left-sided exponential kernel advocated in the paper, the autocorrelation function (ACF) of $(T_{t_{i}^{*}}^{n})_{i=1}^{m}$ turns out to decay close to that of a covariance-stationary AR(1) process (see Panel A in Figure \ref{figure:ev_inference}):
\begin{equation}
\label{equation:armodel}
Z_{i} = \rho Z_{i-1} + \epsilon_{i}, \quad i = 1, \ldots, m,
\end{equation}
where $|\rho| < 1$ and $\epsilon_{i} \overset{ \text{i.i.d.}} \sim N \left(0, 1 - \rho^{2} \right)$. In this model, $Z_{i} \sim N(0,1)$ as consistent with the limit distribution of $T_{t}^{n}$, while the ACF is $\text{cor}(Z_{i}, Z_{j}) = \rho^{|i-j|}$.

To account for dependence in $(T_{t_{i}^{*}}^{n})_{i=1}^{m}$ and get better size and power properties of our test, we simulate the above AR(1) model. We input a value of $\rho$ that is found by conditional maximum likelihood estimation of Eq. \eqref{equation:armodel} from each individual series of $t$-statistics (i.e. OLS of Eq. \eqref{equation:armodel} based on $(T_{t_{i}^{*}}^{n})_{i = 1}^{m}$). We then generate a total of 100,000,000 Monte Carlo replica of the resulting process with a burn-in time of 10,000 observations that are discarded. In each simulation, we record the extreme value $Z_{m}^{*} = \max_{i = 1, \ldots, m} |Z_{i}|$. We tabulate the quantile function of the raw and normalized $Z_{m}^{*}$ series from Eq. \eqref{equation:maxabsTn} -- \eqref{equation:gumbel} across the entire universe of simulations and use this table to draw inference.\footnote{To speed this up for practical work, we prepared in advance a file with the quantile function of the raw and normalized $Z^{*}_{m}$ based on the above setup for several choices of $m$, $\rho$ and selected levels of significance $\alpha$. This file, along with an interpolation routine to find critical values for other $m$ and $\rho$, can be retrieved from the authors at request.}

\begin{figure}[t!]
\begin{center}
\caption{Autocorrelation function and critical value of $t$-statistic.
\label{figure:ev_inference}}
\begin{tabular}{cc}
\small{Panel A: ACF.} & \small{Panel B: Critical value.}\\
\includegraphics[height=0.4\textwidth,width=0.48\textwidth]{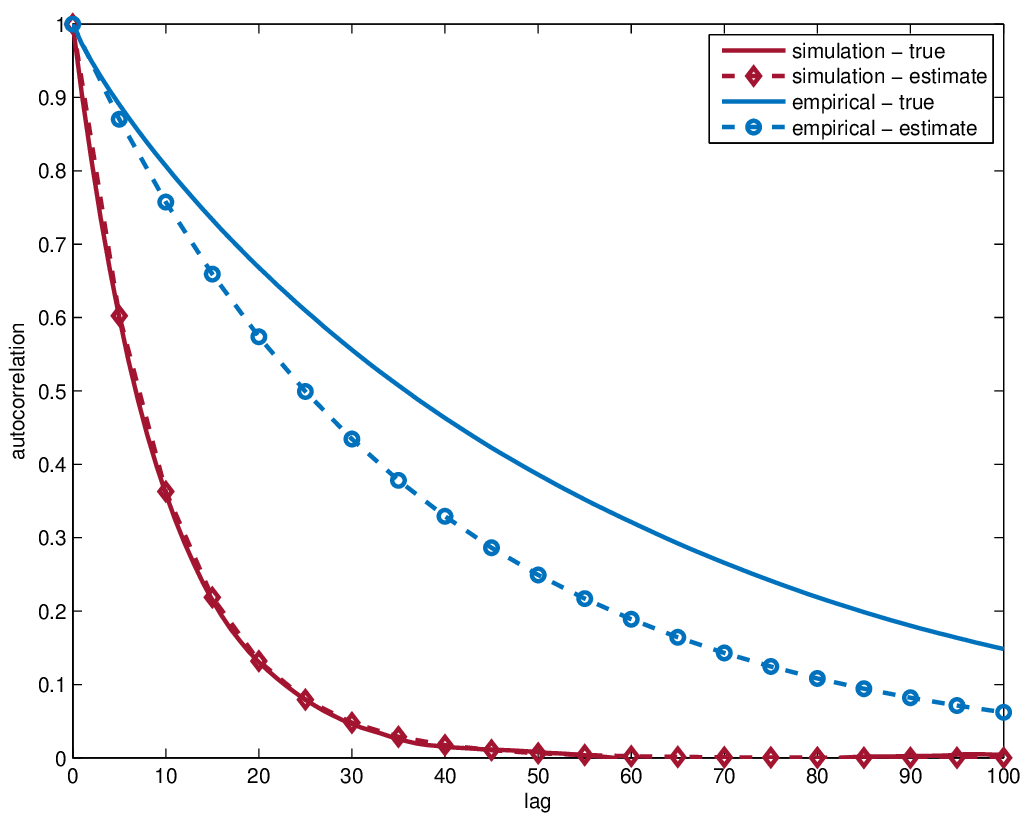} &
\includegraphics[height=0.4\textwidth,width=0.48\textwidth]{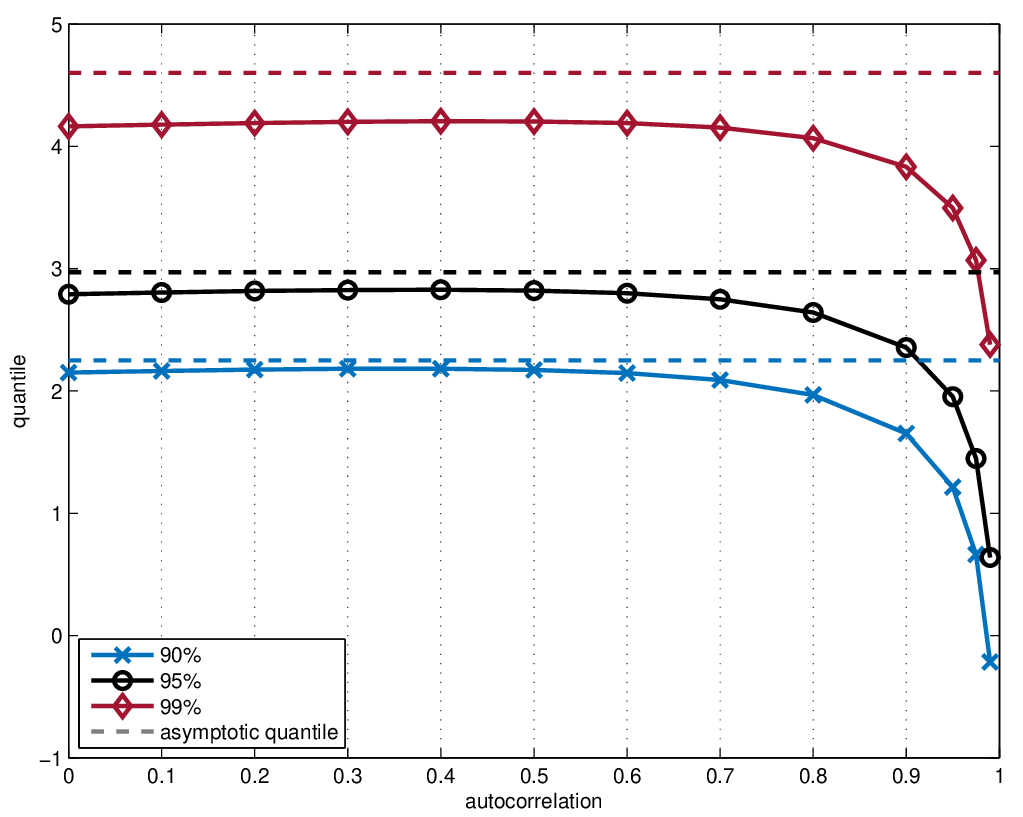}\\
\end{tabular}
\begin{scriptsize}
\parbox{0.95\textwidth}{\emph{Note.} In Panel A, we plot the ACF of $T_{t_{i}^{*}}^{n}$ from the simulation section (averaged across Monte Carlo replica) and the empirical application (averaged across asset markets and over time). The associated dashed curve is that implied by maximum likelihood estimation of the AR(1) approximation in Eq. \eqref{equation:armodel} based on the whole sequence $(T_{t_{i}^{*}}^{n})_{i=1}^{m}$ in our sample. The $t$-statistic is constructed as advocated in the main text. In Panel B, we plot the finite sample quantile found via simulation of the AR(1) model in Eq. \eqref{equation:armodel} by inputting the estimated autoregressive coefficient. This figure is based on $m = 2,500$ and shows the effect of varying the autocorrelation coefficient $\rho$ and confidence level $1 - \alpha$.}
\end{scriptsize}
\end{center}
\end{figure}

In Figure \ref{figure:ev_inference}, we provide an illustration of this approach. In Panel A, we show the ACF of our $t$-statistic for the stochastic volatility model considered in Section  \ref{section:simulation} and the empirical high-frequency data analyzed in Section \ref{section:empirical}. We also plot the curve fitted using the above AR(1) approximation. The estimated ACF is close to the observed one, although there is a slight attenuation bias for the empirical estimates. In Panel B, we report the simulated critical value, as a function of $\rho$ and $\alpha$ with $m$ fixed. These are compared to the ones from the Gumbel distribution. We note a pronounced gap between the finite sample and associated asymptotic quantile, which starts to grow noticeably wider in the region, where $\rho$ exceeds about 0.7 -- 0.8. Apart from that, the extreme value theory offers a decent description of the finite sample distribution for low confidence levels, if the degree of autocorrelation is small, while it gets materially worse, as we go farther into the tails. The latter is explained in part by the fact that even if the underlying sample is uncorrelated, and hence independent in our setting, convergence in law of the maximum term to the Gumbel is known to be exceedingly slow for Gaussian processes \citep*[e.g.][]{hall:79a}.


\section{A parametric test for drift bursts}\label{appendix:parametric}

In this section, as a robustness check, we propose an alternative drift burst test, which is based on a local parametric model.
We assume that, in the window $[0,T]$, the log-price $X$ follows the dynamics:
\begin{equation}\label{equation:model-parametric}
\text{d}X_{t} = \mu(T-t)^{- \alpha} \text{d}t + \sigma(T-t)^{- \beta} \text{d}W_{t}, \quad t \in [0,T],
\end{equation}
where $\alpha \in [0,1)$, $\beta \in [0, 1/2)$, $\mu \in \mathbb{R}$ and $\sigma > 0$ are constant.

The discretely sampled log-return is distributed as $\Delta_{i}^{n} X \overset{\text{d}}{ \sim} N( \mu_{i, \text{db}}, \sigma_{i, \text{db}}^{2})$ with:
\begin{equation}
\mu_{i, \text{db}} = \int_{t_{i-1}}^{t_{i}} \mu (T-s)^{- \alpha} \text{d}s = \frac{ \mu}{1-\alpha} \Big[ (T- t_{i-1})^{1- \alpha} - (T - t_{i})^{1- \alpha} \Big],
\end{equation}
and:
\begin{equation}
\sigma_{i, \text{db}}^{2} = \int_{t_{i-1}}^{t_{i}} \sigma^{2} (T-s)^{- 2 \beta} \text{d}s  = \frac{\sigma^{2}}{1-2 \beta} \Big[ (T- t_{i-1})^{1- 2 \beta} - (T - t_{i})^{1- 2 \beta} \Big].
\end{equation}
%
%
%
%
The log-likelihood for the model is:
\begin{align}
\log \mathcal{L}(\Theta; X) &= \sum_{j = 1}^{n} \log \phi( \mu_{j, \text{db}}, \sigma_{j, \text{db}}^{2}; \Delta_{j}^{n} X) = -\frac{n}{2} \log(2 \pi) - \frac{1}{2} \sum_{j = 1}^{n} \log ( \sigma_{j, \text{db}}^{2}) - \frac{1}{2} \sum_{j = 1}^{n} \frac{(\Delta_{j}^{n} X - \mu_{j, \text{db}})^{2}}{ \sigma_{j, \text{db}}^{2}} ,
\end{align}
where $\phi( \mu_{j, \text{db}}, \sigma_{j, \text{db}}^{2}; \Delta_{j}^{n} X)$ is the Gaussian density function.

The parameter vector $\Theta = (\mu, \sigma, \alpha, \beta)$ is estimated by maximum likelihood:
\begin{equation}
\hat{ \Theta}_{\text{ML}} = \underset{ \Theta}{\arg \max} \, \log \mathcal{L}(\Theta; X).
\end{equation}
We test the drift burst hypothesis:
\begin{equation}
\mathcal{H}_{0} : \alpha = 0 \quad \text{against} \quad \mathcal{H}_{1} : \alpha > 0
\end{equation}
with a likelihood ratio statistic $\Lambda = -2 ( \log \mathcal{L}(\hat{ \Theta}_{ \text{ML}}; X) -  \log \mathcal{L}(\Theta_{\mathcal{H}_{0}}; X)) \overset{a}{ \sim} \chi_{1}^{2}$ (a volatility burst is allowed, since $\beta$ is unrestricted). In the Online Appendix, based on simulations from the parametric model in Eq. \eqref{equation:model-parametric}, the test is found to perform reasonably in finite samples.

\begin{figure}[t!]
\begin{center}
\caption{Average drift burst log-return. \label{figure:dbsample}}
\begin{tabular}{cc}
\small{Panel A: Average log-return by direction.} & \small{Panel B: Average absolute log-return by asset.}\\
\includegraphics[height=0.4\textwidth,width=0.48\textwidth]{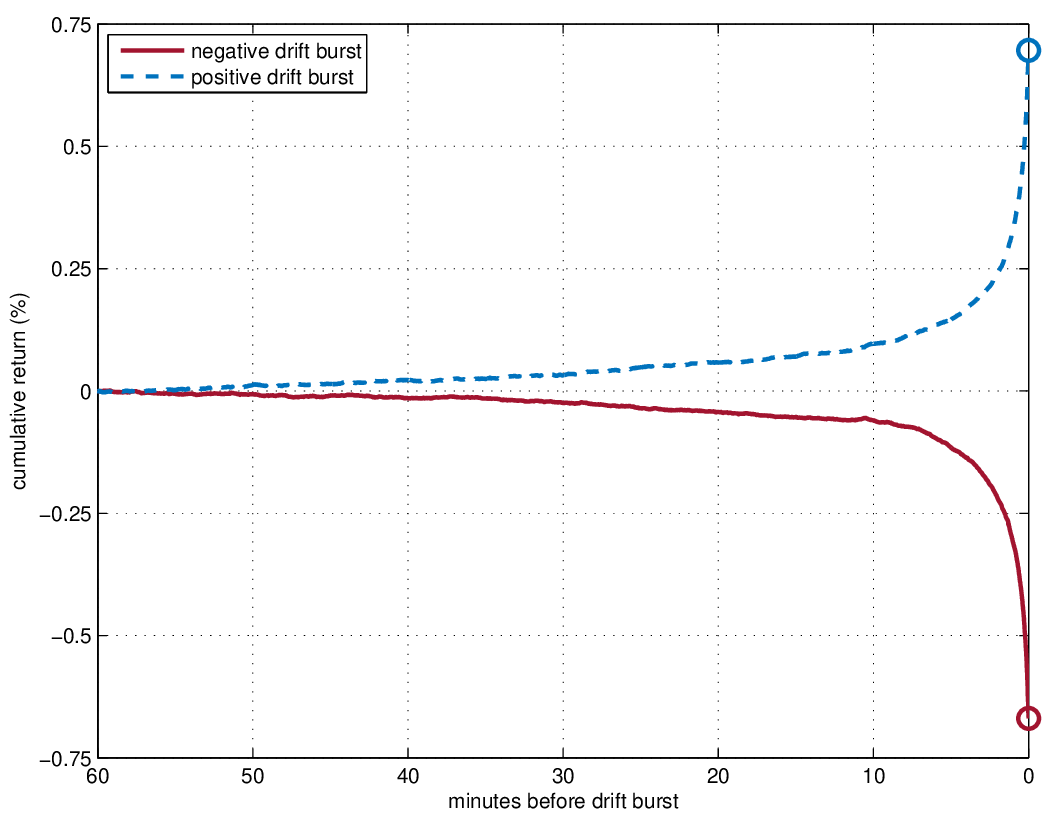} &
\includegraphics[height=0.4\textwidth,width=0.48\textwidth]{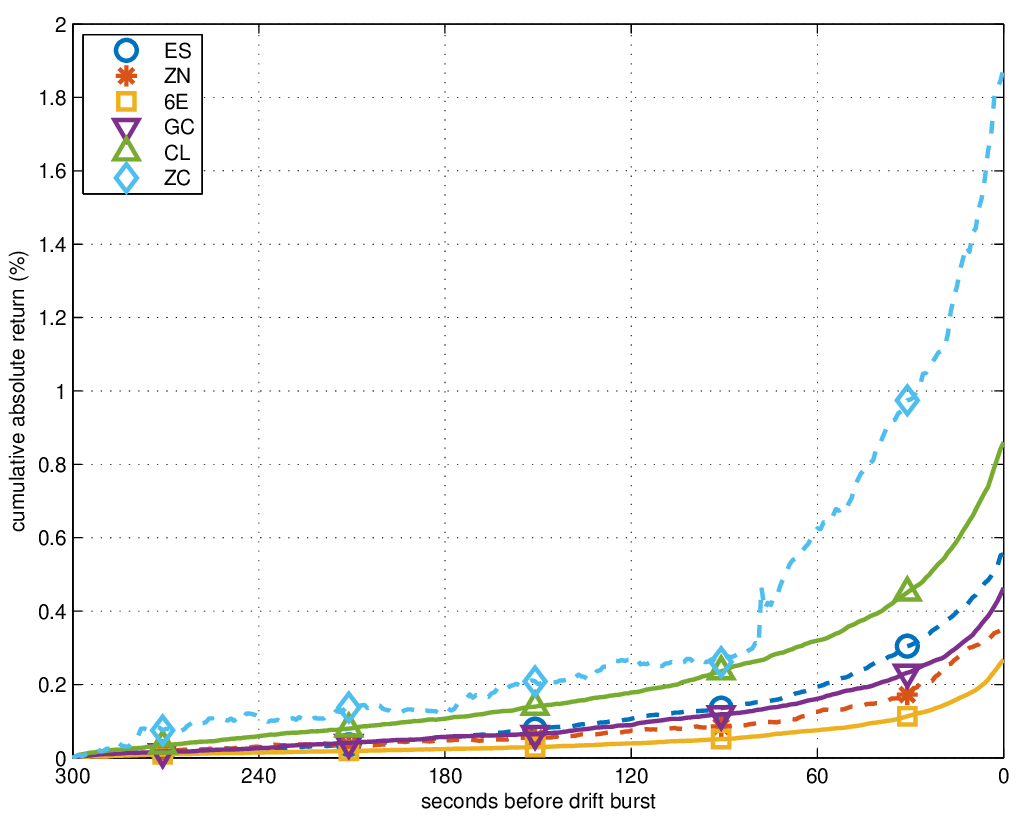} \\
\end{tabular}
\begin{scriptsize}
\parbox{\textwidth}{\emph{Note.} In Panel A, we show the cumulative log-return for positive and negative drift bursts, averaged across assets, in a one-hour window prior to the event. In Panel B, we extend this analysis by plotting the average absolute cumulative log-return by asset in a five-minute window prior to the event.}
\end{scriptsize}
\end{center}
\end{figure}
We implement the parametric test on two samples. The first, detected by the nonparametric test with a maximum $t$-statistic of at least five (in absolute value) and labeled the ``drift burst sample,'' consists of second-by-second prices sampled in a one-hour run-up window before the drift burst. The sample has 933 events. In Figure \ref{figure:dbsample}, we illustrate a typical drift burst in these data. In Panel A, we plot the average cumulative log-return of negative and positive drift bursts in the event window. In Panel B, we show the average absolute cumulative log-return for different asset classes in the last five minutes. The second ``no signal'' control sample has the corresponding data (when it is available) a week after each drift burst, which is 855 events.

\begin{figure}[t!]
\begin{center}
\caption{Empirical analysis of the parametric drift burst testing procedure. \label{figure:lr}}
\begin{tabular}{cc}
\small{Panel A: Likelihood ratio statistic.} & \small{Panel B: $\hat{ \alpha}_{ \text{ML}}$ and $\hat{ \beta}_{ \text{ML}}$.} \\
\includegraphics[height=0.4\textwidth,width=0.48\textwidth]{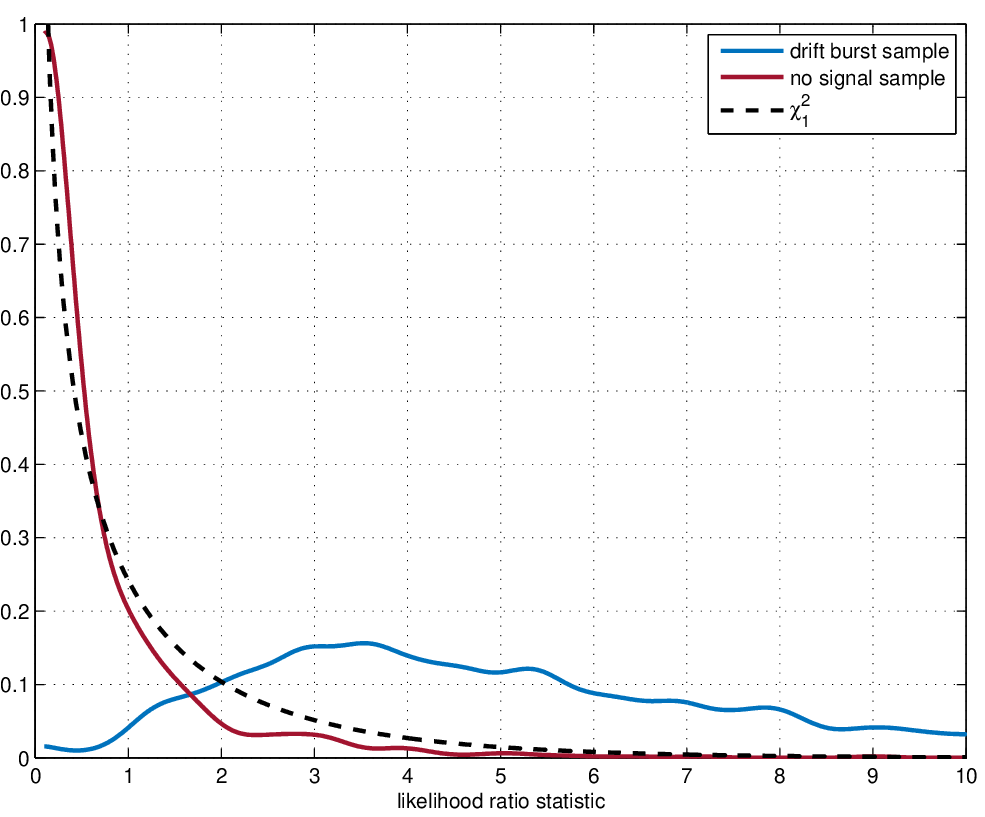} &
\includegraphics[height=0.4\textwidth,width=0.48\textwidth]{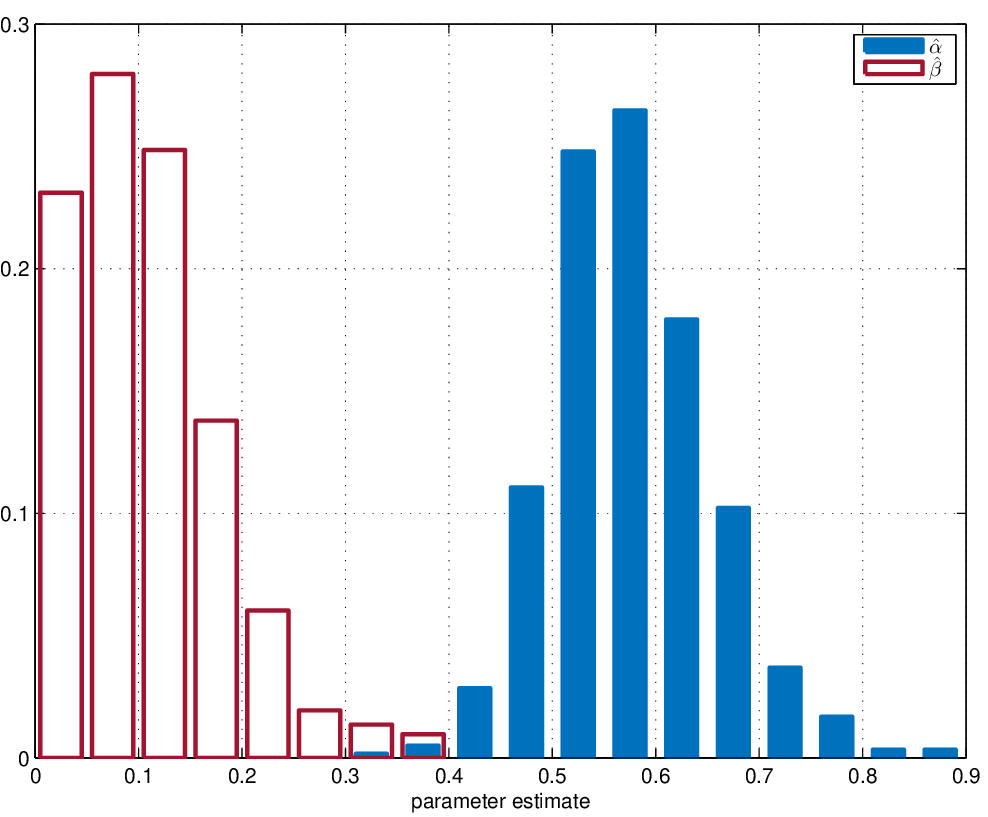} \\
\small{Panel C: $\hat{ \alpha}_{ \text{ML}} - \hat{ \beta}_{ \text{ML}}$.} & \small{Panel D: $\hat{ \alpha}_{ \text{ML}}$ versus $\hat{ \beta}_{ \text{ML}}$.} \\
\includegraphics[height=0.4\textwidth,width=0.48\textwidth]{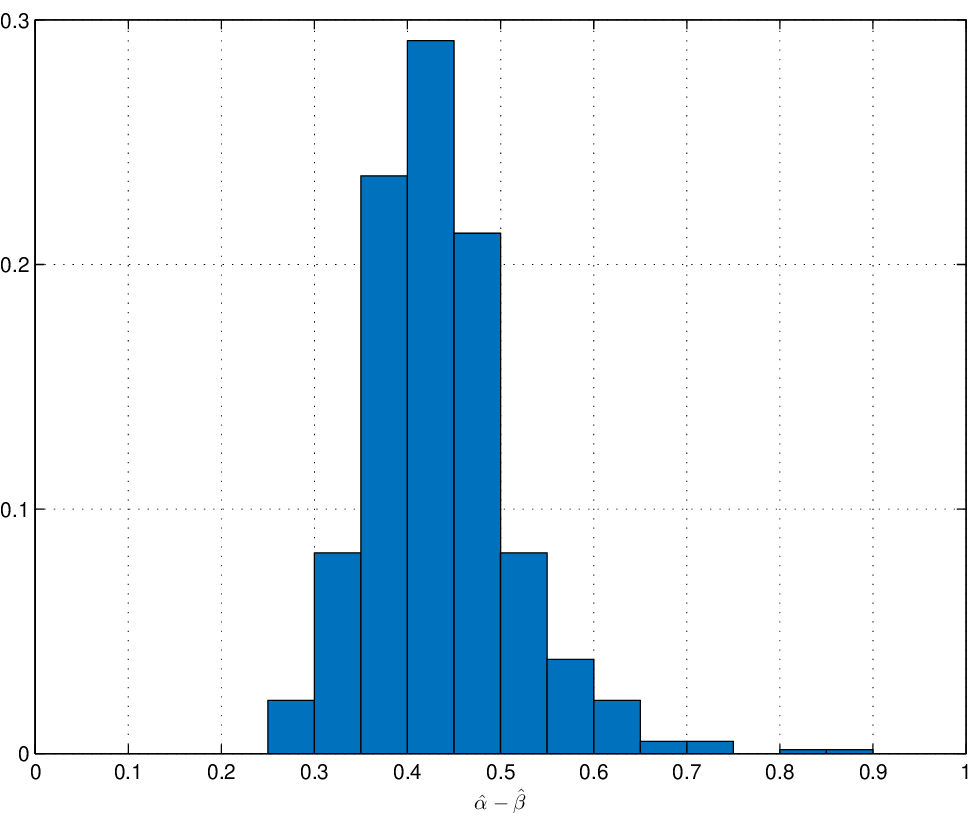} &
\includegraphics[height=0.4\textwidth,width=0.48\textwidth]{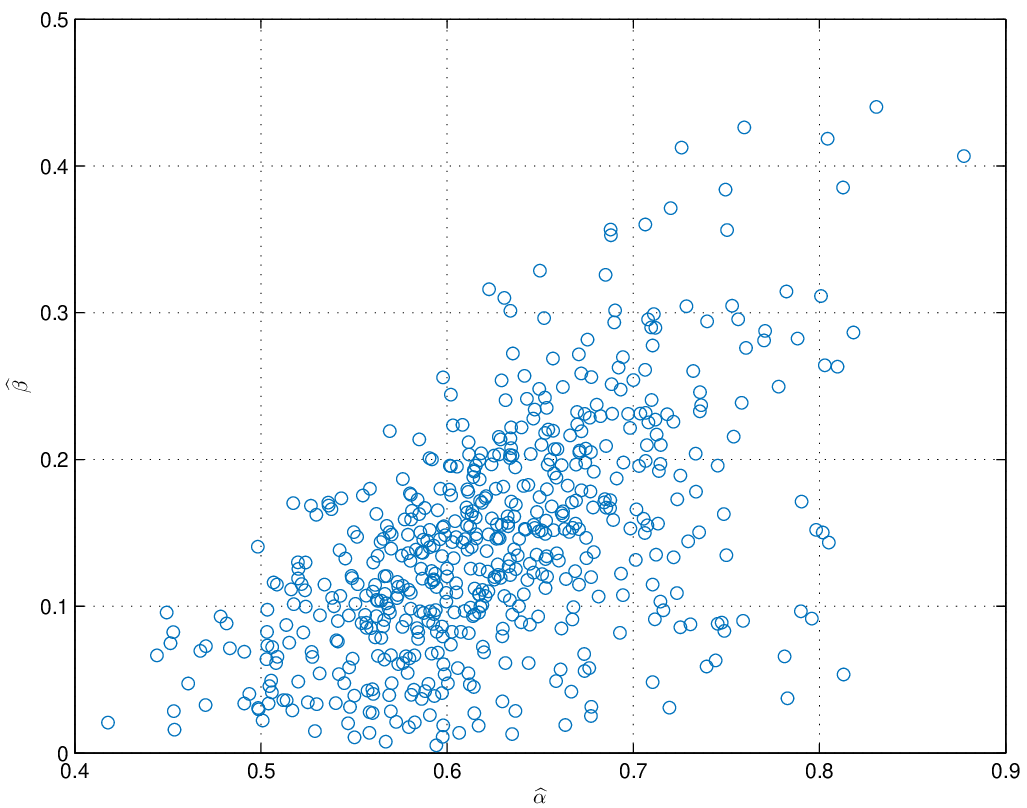}
\end{tabular}
\begin{scriptsize}
\parbox{\textwidth}{\emph{Note.} The figure reports on the outcome of the empirical analysis based on the parametric test for drift bursts. In Panel A, we show a kernel estimate of the distribution of the likelihood ratio statistic in the drift burst and no signal sample, which are compared to the $\chi_{1}^{2}$ distribution. In Panel B -- C, we plot a histogram of the maximum likelihood estimates $\hat{ \alpha}_{ \text{ML}}$, $\hat{ \beta}_{ \text{ML}}$, and $\hat{ \alpha}_{ \text{ML}} - \hat{ \beta}_{ \text{ML}}$. In Panel D, a scatter plot of $\hat{ \alpha}_{ \text{ML}}$ against $\hat{ \beta}_{ \text{ML}}$ is shown.}
\end{scriptsize}
\end{center}
\end{figure}	
In Panel A of Figure \ref{figure:lr}, we report the distribution of the likelihood ratio statistic in these two samples. As seen, $\Lambda$ is close to chi-square distributed in the no signal sample, while it is strongly skewed to the right in the drift burst sample.\footnote{As a robustness check, we applied three implementation windows: one hour, thirty minutes and ten minutes. In the no signal sample, and at a 5\% significance level, the likelihood ratio statistic rejects $\mathcal{H}_{0}$ in $1.87\%$, $1.87\%$ and $1.17\%$ of the events, while the comparable numbers in the drift burst sample are $63.99\%$, $54.88\%$ and $42.66\%$. Thus, drift bursts identified by the nonparametric statistic are also more likely to be identified by the parametric test as more data are included. This is consistent with the simulation analysis in Section \ref{section:simulation}.} In Panel B, we show a histogram of $\hat{ \alpha}_{ \text{ML}}$ and $\hat{ \beta}_{ \text{ML}}$ based on the 597 events in the drift burst sample, where the likelihood ratio test statistic is significant at a $5\%$ level. The sample averages are $\bar{ \hat{ \alpha}}_{ \text{ML}} = 0.6250$ and $\bar{ \hat{ \beta}}_{ \text{ML}} = 0.1401$. $\hat{ \beta}_{ \text{ML}}$ is typically much smaller than $\hat{ \alpha}_{ \text{ML}}$ (otherwise it is hard for the nonparametric test to detect an event). This aligns with the theory in Section \ref{section:dbh}. The average difference $\bar{ \hat{ \alpha}}_{ \text{ML}} - \bar{ \hat{ \beta}}_{ \text{ML}} = 0.4849$ is borderline with the no-arbitrage boundary of 0.5 (the full histogram is in Panel C). A high correlation of 51.94\% between $\hat{ \alpha}_{ \text{ML}}$ and $\hat{ \beta}_{ \text{ML}}$ is noticed, as also evident by the scatter plot in Panel D. It indicates that a stronger drift burst is accompanied by a stronger volatility burst, as consistent with the theory.

The parametric model can also be used to test a volatility burst:
\begin{equation}
\mathcal{H}_{0}^{ \prime} : \beta = 0 \quad \text{against} \quad \mathcal{H}_{1}^{ \prime} : \beta > 0
\end{equation}
with $\Lambda^{ \prime} = -2 ( \log \mathcal{L}( \hat{ \Theta}_{ \text{ML}}; X) -  \log \mathcal{L}( \Theta_{\mathcal{H}^{ \prime}_{0}}; X)) \overset{a}{ \sim} \chi_{1}^{2}$. If we apply this to the drift burst sample, $\mathcal{H}_{0}^{ \prime}$ is discarded at the 5\% level for 93.68\% of the events, showing that volatility is almost always co-exploding during a drift burst.

\clearpage

\renewcommand{\baselinestretch}{1.0}
\small
\bibliography{burst}

\end{document}